\documentclass[journal,10pt,draftclsnofoot,onecolumn]{IEEEtran}
\usepackage{amsmath,amssymb,amsthm,graphicx,subfigure}
\usepackage{mathrsfs,algorithm,algpseudocode}
\usepackage{booktabs,multirow,bm}
\usepackage{cite,caption2}
\usepackage{xcolor}

\graphicspath{{figures/}}

\begin{document}

\newtheorem{theorem}{Theorem}
\newtheorem{exmp}{Example}
\newtheorem{lemma}{Lemma}
\newtheorem{proposition}{Proposition}
\newtheorem{definition}[exmp]{Definition}
\renewcommand{\algorithmicrequire}{\textbf{Input:}}
\renewcommand{\algorithmicensure}{\textbf{Output:}}
\newcommand{\tabincell}[2]{\begin{tabular}{@{}#1@{}}#2\end{tabular}}


\title{\emph{Pacos}: Modeling Preference Reversals In Users' Context-Dependent Choices}
\author{Qingming~Li and H.Vicky~Zhao\thanks{The authors are with the Department of Automation, Beijing National Research Center for Information Science and Technology, Tsinghua University, Beijing 100084 P. R. China (email: qingmingli45@163.com, vzhao@tsinghua.edu.cn).}}
\maketitle

\begin{abstract}
Choice problems refer to the problem of selecting the best choices from several available items, and learning users' preferences in choice problems is of great importance in understanding users' decision making mechanisms and providing personalized services. Existing works typically assume that people evaluate items independently. In practice, however, users' preferences depend on the market in which items are placed, which is known as the context effects; and the order of users' preferences for two items may even be reversed, which is called to preference reversals. In this work, we identify three factors contributing to the context effects: users' adaptive weights, the inter-item comparison, and display positions. We propose a context-dependent preference model named \emph{Pacos} as a unified framework to address three factors simultaneously, and consider two design methods including an additive method with high interpretability and an ANN-based method with high accuracy. We study the conditions for preference reversals to occur and provide a theoretical proof of the effectiveness of \emph{Pacos} in predicting when preference reversals would occur. Experimental results show that the proposed method has better performance than prior works in predicting users' choices, and has great interpretability to help understand the cause of preference reversals.

\end{abstract}



\section{Introduction}
\label{sec:intro}

Choice problems, such as purchasing a festival gift or picking a restaurant,  involve comparing several available items. Previous works on preference modeling and analysis typically assume that people evaluate items independently, and the relative preference between two items is fixed regardless of other competing options \cite{mcfadden1973conditional}. However, numerous studies show that the above independence assumption is frequently violated in reality \cite{benson2016relevance,rieskamp2006extending}. It is essential to model how the relative preference is influenced by competing options and figure out how people select their best choices. This study can help understand users' decision making mechanisms and offer personalized services, and provide important guidelines on pricing strategies and sales forecasts.

\begin{figure}[tbp]
\centering
\begin{minipage}[b]{.32\linewidth}
  \centering
  \centerline{\includegraphics[width=5cm]{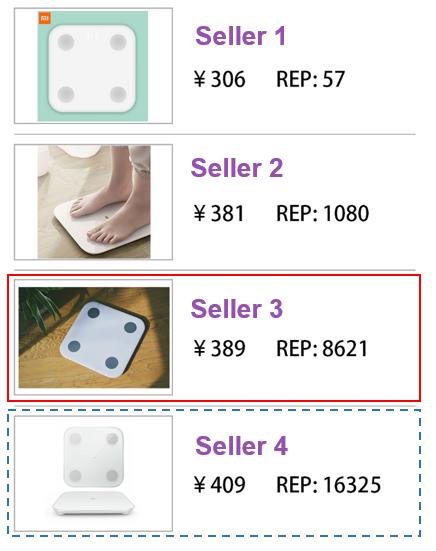}}
  \centerline{(a) market I}\medskip
\end{minipage}
\hfill
\begin{minipage}[b]{.32\linewidth}
  \centering
  \centerline{\includegraphics[width=5cm]{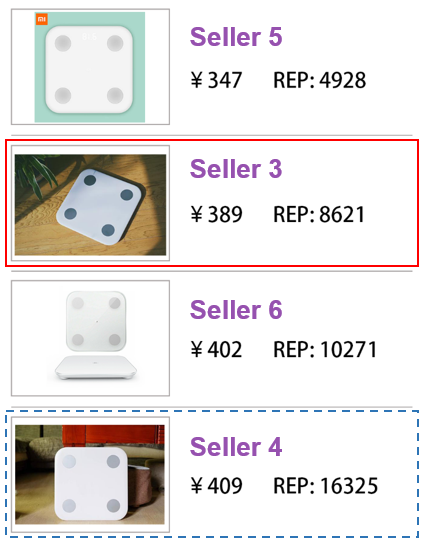}}
  \centerline{(b) market II }\medskip
\end{minipage}
\begin{minipage}[b]{.32\linewidth}
  \centering
  \centerline{\includegraphics[width=4cm]{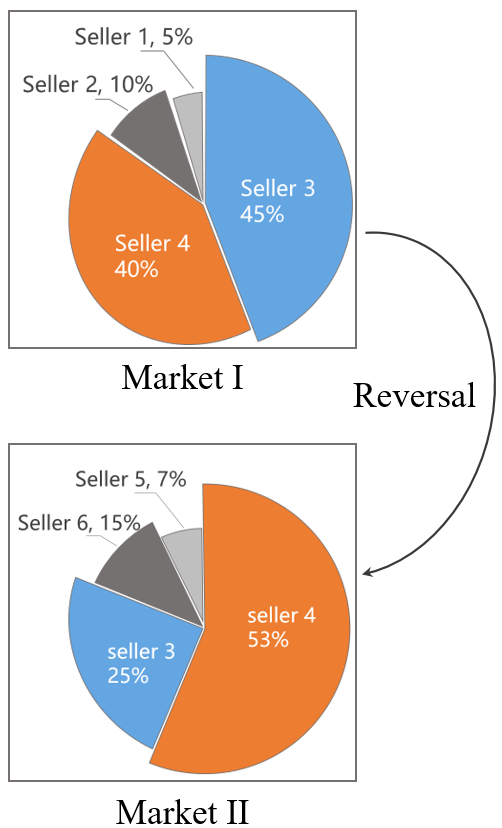}}
  \centerline{(c) the selection data}\medskip
\end{minipage}
\caption{An example of preference reversals.}
\label{fig:market_pref_revs}
\end{figure}

To show this independence violation, we conduct a real user test. In our test, we set two markets of Xiaomi scale, as shown in Fig. \ref{fig:market_pref_revs} (a) and (b). In these two markets, we consider sellers described by two attributes: price (\textyen) and seller reputation (REP). We invite 647 subjects to select their favorite sellers from each market. The selection results are displayed in Fig. \ref{fig:market_pref_revs} (c). Specifically, in market I, $45\%$ of users choose seller 3 and $40\%$ of users select seller 4, indicating that more users prefer seller 3. While in market II, $25\%$ of users choose seller 3 and $53\%$ of users select seller 4, revealing that more users prefer seller 4. In these two markets, the order of users' preferences for seller 3 and seller 4 is reversed, which is called \emph{preference reversals} in economics \cite{grether1979economic}. From this test, we can see that users' preferences depend on the market and especially the competing items, and this dependence is called the \emph{context effects} in economics \cite{rieskamp2006extending}. Context effects can take various forms in different scenarios, among which the preference reversal is one of the most typical ones and is the focus of our work.


\textbf{Our Goal.}\quad In this work, to better understand the impact of context effects on users' choices, we investigate the following three questions: 1) which factors contribute to context effects; 2) how to model users' context-dependent preferences with consideration of these factors; and 3) how to verify that the proposed method can better model and predict users' preference reversals.  To answer question 1, note that there are works looking at context effects in terms of the way how options are presented or described  \cite{kahneman1984choices,thomadsen2018context}, and situational factors such as weather or consumption environment when people make a choice \cite{busse2015psychological,cheng2009effect}. These factors are not easily observed and not the concern of our work. After extensively study of prior works on context effects, we identify three major factors: adaptive weights, inter-item comparison, and display positions. Although these factors have been studied separately in \cite{ariely1995seeking,tomlinson2021learning,rieskamp2006extending,bar2015position}, there is no unified framework considering all three factors simultaneously. In addition, we study the permutation of competing items and find that these factors can be divided into two categories: adaptive weights and inter-item comparison are both permutation invariant, while the factor of display positions is permutation sensitive. That is, adaptive weights and inter-item comparison are independent of the display order of the available items, while the factor of display positions depends on the order in which items are arranged.

The first factor is \textbf{adaptive weights}. A large stream of literature finds that markets could affect the weights people assign to attributes \cite{ariely1995seeking,tomlinson2021learning}. For example,  when two items are similar, people would easily recognize the difference between their attributes and assign a larger weight to the attribute that differs the most (larger difference, larger weight)\cite{ariely1995seeking}. Furthermore, it is observed that when the average value of an attribute is larger than other attributes,  people will give a larger weight to this attribute (larger average value, larger weight) \cite{tomlinson2021learning}. Preference reversals may occur when the change in weights leads to a change in the order of items' utilities. We find that the factor of adaptive weights has the permutation-invariant property. This is intuitive as disturbing the order of items do not affect the difference between attributes, nor the average values of attributes.

The second factor is the \textbf{inter-item comparison}. Items compete in a market in a complex way. For example, from the study in \cite{rieskamp2006extending}, users would decrease their preferences for an item if a similar option joins the market, or increase their preferences for an item when a decoy item (an item that is slightly inferior in all attributes) exists. In our work, we refer to the competition among items as the inter-item comparison. Preference reversals happen when the inter-item comparison changes the order of items' utilities. We find that the factor of inter-item comparison is permutation invariant, as rearrangement of items does not change their similarity, nor the existence of a decoy item.

The third factor is the \textbf{display positions}. Extensive literature shows that users prefer items displayed in certain ranks. For example, some users always choose sellers displayed in the top rank, and some users prefer items in the middle \cite{bar2015position}. Swapping the display positions of two items can cause a change in their utilities, which may potentially lead to the preference reversal. We adopt display positions to account for the permutation-sensitive part in users' preferences, which is usually ignored in prior preference modeling works  \cite{tomlinson2021learning,osogami2014restricted,seshadri2019discovering}.



\textbf{Our Contributions.}\quad In this work, based on the above discussions, we propose a context-dependent \textbf{P}reference model \textbf{A}ddressing preference reversals in \textbf{C}h\textbf{O}ice problem\textbf{S} (\emph{\textbf{Pacos}}).  To model users' context-dependent preferences, we propose a context-aware utility function, which contains three utility modules to account for three contributing factors. Based on the proposed utility function, we consider two design methods, including an additive method with better interpretability and an ANN-based method to achieve high accuracy. A novel learning algorithm is proposed to learn unknown parameters in the model. Experimental results show that the proposed method has better performance than prior works in predicting users' choices. The contributions of our work are as follows.
\begin{itemize}
  \item We identify three important factors that contribute to context effects, and propose \emph{Pacos} as a unified framework to address three factors simultaneously. In addition, we propose two models to achieve high interpretability and high accuracy, respectively.
  \item For the proposed additive method, we study the conditions for preference reversals to occur, and provide a theoretical proof of its effectiveness in analyzing preference reversals.
  \item We design a preference reversal prediction experiment and collect real user data. Experiments show that \emph{Pacos} can effectively predict preference reversals, and has good interpretability to understand the cause of preference reversals.
\end{itemize}

The rest of the paper is organized as follows. Section \ref{sec:related_work} is the literature review. Section \ref{sec:framework} presents the \emph{Pacos} framework and the two module design methods. Section \ref{sec:calculation} theoretically analyze the effectiveness of the additive method. Section \ref{sec:test} shows the experimental results. Conclusions are drawn in Section \ref{sec:conclusion}.

\section{Related Works}
\label{sec:related_work}
Previous works on preference modeling can be classified into multinomial logit model (MNL) based methods and machine learning-based methods, which will be discussed in detail.

\subsection{MNL-Based Methods}
The multinomial logit model, which is proposed by McFadden in 1973, is the cornerstone in the field of preference theory \cite{mcfadden1973conditional}. MNL adopts the utility function to describe users' preferences, and a larger utility value indicates a higher preference. MNL formulates the utility of an item as two parts: a deterministic part that is determined by its attributes, and a random part that accounts for unobserved factors \cite{train2009discrete}. MNL provides a simple and interpretable framework to analyze the decision making process. However, MNL evaluates each item independently, and has the \emph{independence of irrelevant alternatives} (IIA) issue. Consider a pair of sellers A and B, and define the relative preference of seller B to seller A as the ratio of the probability of choosing seller B to the likelihood of selecting seller A. Then, the IIA issue states that the relative preference of seller B to seller A is fixed, and is not affected by other options. The IIA issue prevents MNL from addressing preference reversals, that is, it is not possible to find parameters to fit the probability distribution when the preference reversal occurs \cite{rieskamp2006extending}.


Based on MNL, new theoretical methods have been proposed to address context effects. For example, in Linear Context Logit (LCL) model, a mapping matrix is adopted to quantify the adaptive weights, and theoretical analysis demonstrates that the mapping matrix is identifiable and can be effectively trained from data \cite{tomlinson2021learning}. Besides, Context-Dependent Model (CDM) assumes that context effects come from the pairwise comparisons between items \cite{seshadri2019discovering}, which can be computed by the inner product between a target vector and a context vector. Moreover, PRIMA++ observes that more competition exists between items with similar attributes, and introduces the indifference curves from microeconomics to model the competition \cite{li2021prima++}. Additionally, the work in \cite{mohr2017attraction} adapts the sequential sampling model \cite{baker2022degenerate} to account for context effects under risk, and investigates context effects from a psychological perspective. These MNL-based methods do not suffer from the IIA issue, and theoretically, are able to cope with context effects in certain specific scenarios. However, none of them explicitly discuss preference reversals.

\subsection{Machine Learning-Based Methods}

These machine learning-based works transform the choice problem into a classification or ranking problem, and then apply machine learning methods to solve it. For example, the choice problem is translated into a binary classification problem in \cite{mottini2018understanding,gao2021extrapolation}. In the training stage, selected items are labeled as positive, and the rest are labeled as negative. Then classification methods are applied to learn how to distinguish two kinds of samples. In the testing stage, available items are evaluated by the probability of being classified as positive, and the item with the highest probability is regarded as the predicted choice. In addition, the choice problem is sometimes translated into a pairwise ranking problem \cite{liu2021pair}. Specifically, the choice records are converted into multiple item pairs, and each pair contains one selected item and one unselected item. Then the pairwise ranking methods, like RankNet \cite{burges2005learning} and RankSVM \cite{fine2001efficient}, are applied to identify the selected items in these pairs. However, these methods do not consider context effects, and are subject to potential preference information loss in the process of transformation. For example, in Fig. \ref{fig:market_pref_revs} (a),  users prefer seller 3 to seller 4 in market I. While after transformation, only the information about users preferring seller 3 to seller 4 is available, and the information about the market I is lost.

Recently, many machine learning-based methods have been proposed to explicitly address choice problems. These methods do not need an extra transformation process and avoid possible information loss. For example, the pointer neural network (PNN) \cite{vinyals2015pointer}, which is an encoder-decoder network based on recurrent neural network and attention mechanism, is utilized as the mapping function between the available items in the markets and users' choices. The available items are fed into PNN one by one, and PNN would point to the item that is predicted as the best choice \cite{mottini2017deep}. In \cite{rosenfeld2020predicting}, a function-aggregation method is proposed, which assumes that one utility function can capture only a part of users' preferences, and the whole preferences can be infinitely approximated by introducing enough utility functions. The function-aggregation method adopts multiple utility functions simultaneously, and the total utility equals to their weighted sum. The restricted Boltzmann machine is adopted in \cite{osogami2014restricted} to model context effects. It is proved that the restricted Boltzmann machine can be translated into the MNL formulation, together with an additional term accounting for the comparisons among items.


\subsection{Summary}

These prior methods, including MNL-based and machine learning-based methods, have several limitations in handling context effects and preference reversals. First, these methods cannot comprehensively address all three factors. Specifically, the MNL-based methods study only one factor, for example, the LCL emphasizes the adaptive weights and the CDM stresses on the inter-item comparison, while none consider all three factors. The machine learning-based methods are used like a black box and does not explicitly discuss these three factors. Furthermore, neither methods explicitly analyze preference reversals, let alone the conditions for it to occur or the ability to cope with it. These limitations motivate us to propose a unified framework that jointly considers all three factors and is effective in predicting preference reversals.

\section{The Proposed Context-Dependent Preference Model}
\label{sec:framework}
Consider the scenario that a user is going to buy a specific product online, and there are several matching items for sale, each corresponding to a different seller. In the following, the terms ``item'' and ``seller'' are used interchangeably. We refer to the list containing all matching sellers as the market $\mathcal{S}$. Without loss of generality, we consider a market with at most $N$ sellers, and we use $\boldsymbol{s}_i$ to denote the $i$-th seller.

We assume that each seller is described by two types of features. The first type is the item's attributes, such as the item's price, the seller's reputation, etc. Let the number of attributes be $d$, and we have $\boldsymbol{s}_i\in\mathbb{R}^d$. For convenience, we use $\mathcal{S}_{att}=[\boldsymbol{s}_1,\cdots,\boldsymbol{s}_N]\in\mathbb{R}^{d\times N}$ as the attribute matrix, which stores the $d$ attributes of $N$ sellers. The second type of features is their display positions. We use the one-hot vector $pos(\boldsymbol{s}_i|\mathcal{S})\in\mathbb{R}^N$ to encode the display position of seller $\boldsymbol{s}_i$. Formally, if the seller $\boldsymbol{s}_i$ is displayed at the $i$-th position in market $\mathcal{S}$, then the $i$-th element of $pos(\boldsymbol{s}_i|\mathcal{S})$ is 1 and all other elements are zero \cite{yao2013recurrent}. For convenience, we let $\mathcal{S}_{pos}=[pos(\boldsymbol{s}_1|\mathcal{S}),\cdots,pos(\boldsymbol{s}_N|\mathcal{S})]\in\mathbb{R}^{N\times N}$ be the position matrix, which encodes the display positions of $N$ sellers. 


\begin{figure}[!t]
  \centering
  \centerline{\includegraphics[width=15cm]{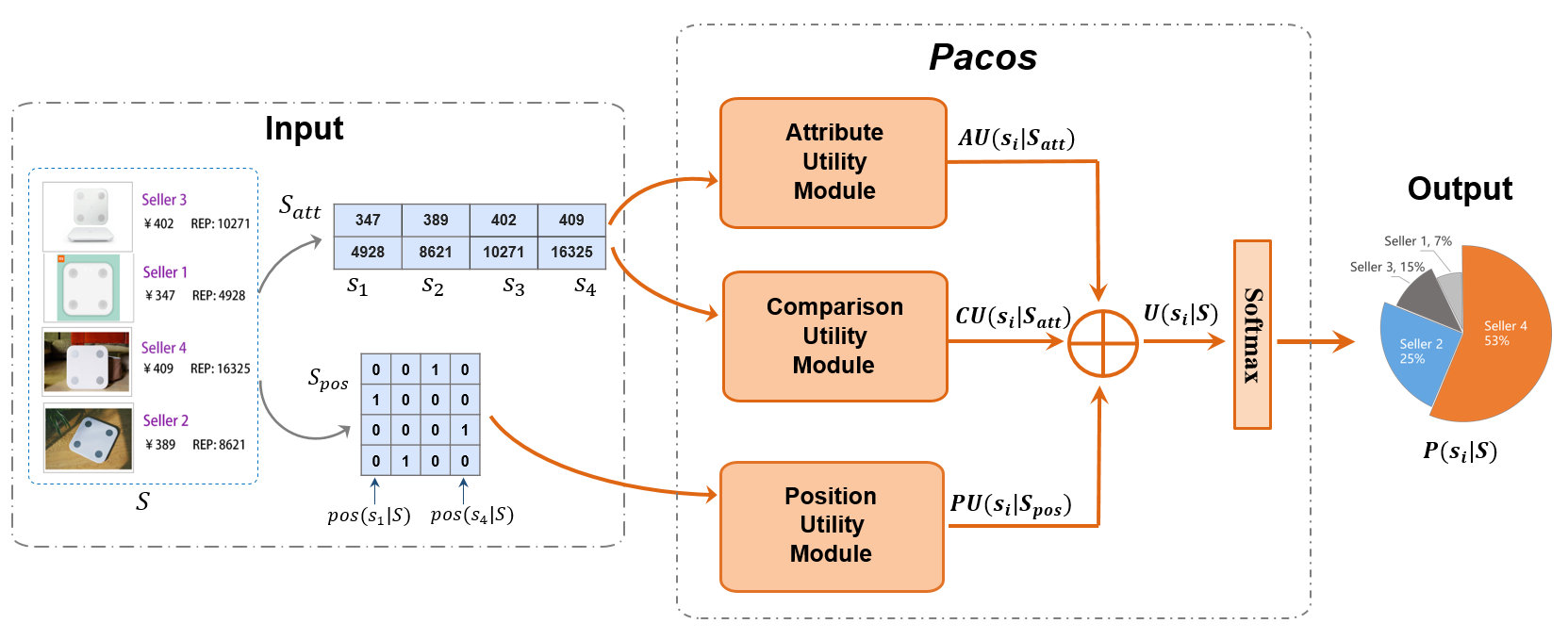}}
\caption{The framework of \emph{Pacos}.}
\label{fig:overall_structure}
\end{figure}

In this section, we study how to model context effects and propose a unified framework \emph{Pacos} to quantify users' preferences in consideration of the context effects. Fig. \ref{fig:overall_structure} shows the framework of \emph{Pacos}. Here, the input is the market $\mathcal{S}$, which can be represented by $\mathcal{S}_{att}$ and $\mathcal{S}_{pos}$. After the input, three utility modules are connected in parallel to model three factors, respectively. Next, outputs of the utility modules are combined to obtain the total utility $U(\boldsymbol{s}_i|\mathcal{S})$. The softmax function is followed to translate each item's utility to its probability of being selected $\mathcal{P}(\boldsymbol{s}_i|\mathcal{S})$, which is the final output of our model.

In Fig. \ref{fig:overall_structure}, issues need to be addressed include: (1) how to construct the utility function $U(\boldsymbol{s}_i|\mathcal{S})$, (2) how to design the three utility modules, and (3) how to learn unknown parameters in these modules. To address these issues, we explain the construction of the utility function in Section \ref{sec:construction}, show details of the module designs in Section \ref{sec:design}, and propose a learning algorithm in Section \ref{sec:learning}.

\subsection{The Utility Function}
\label{sec:construction}


In our work, we propose a context-aware utility function $U(\boldsymbol{s}_i|\mathcal{S})$, which contains three utility modules that are used to account for the three factors, respectively. Details of the three utility modules are as follows.

The first is the \textbf{attribute utility module.} Prior works typically use a fixed vector $\boldsymbol \beta$ to store weights, which cannot cope with the dependence between users' adaptive weights and competing items in the market \cite{li2021prima++,han2020neural}. In our work, we propose an adaptive weight vector $\boldsymbol\beta(\mathcal{S}_{att})$ to address such dependence. The vector $\boldsymbol\beta(\mathcal{S}_{att})$ is a permutation-invariant function of the attribute matrix $\mathcal{S}_{att}$. Details of how to design the vector $\boldsymbol\beta(\mathcal{S}_{att})$ are in Section \ref{sec:design}. We refer to this part of utility as the attribute utility, and the attribute utility of seller $\boldsymbol{s}_i$ is computed by $AU(\boldsymbol{s}_i|\mathcal{S}_{att})=\boldsymbol\beta(\mathcal{S}_{att})^T\cdot \boldsymbol{s}_i$.


The second is the \textbf{comparison utility module.} We use the comparison utility $CU(\boldsymbol{s}_i|\mathcal{S}_{att})$ to capture inter-item comparison, and propose that $CU(\boldsymbol{s}_i|\mathcal{S}_{att})$ is a permutation-invariant function of the attribute matrix $\mathcal{S}_{att}$. When item $\boldsymbol{s}_i$ receives preference gain from the inter-item comparison, the comparison utility $CU(\boldsymbol{s}_i|\mathcal{S}_{att})$ is positive. While when the item loses preference, the value is negative. Details of the comparison utility module will be discussed in Section \ref{sec:design}.



The third is the \textbf{position utility module.} We refer to the part of the utility resulting from display positions as the position utility $PU(\boldsymbol{s}_i|\mathcal{S}_{pos})$. In this module, a fixed vector $\boldsymbol\alpha\in\mathbb{R}^N$ is used to store users' preferences for display positions. When the $i$-th element of $\boldsymbol\alpha$ is larger, users' preferences for the seller displayed at the $i$-th position are higher. Given the position matrix $\mathcal{S}_{pos}$, the position utility of  seller $\boldsymbol{s}_i$ is computed by $PU(\boldsymbol{s}_i|\mathcal{S}_{pos})=\boldsymbol\alpha^T\cdot pos(\boldsymbol{s}_i|\mathcal{S})$.


For simplicity, we consider the additive form and formulate the total utility function as
\begin{align}
\label{equ:uti}
U(\boldsymbol{s}_i|\mathcal{S})&= AU(\boldsymbol{s}_i|\mathcal{S}_{att})
+ CU(\boldsymbol{s}_i|\mathcal{S}_{att})
+PU(\boldsymbol{s}_i|\mathcal{S}_{pos})\\   \notag
&= \boldsymbol\beta(\mathcal{S}_{att})^T\cdot \boldsymbol{s}_i+CU(\boldsymbol{s}_i|\mathcal{S}_{att})+
\boldsymbol\alpha^T\cdot pos(\boldsymbol{s}_i|\mathcal{S}).
\end{align}
We will investigate other complicated utility functions in our future work. The proposed utility function in Eq. \eqref{equ:uti} provides a unified framework to jointly consider the three factors. Then, we use the softmax function to convert each item's utility $U(\boldsymbol{s}_i|\mathcal{S})$ to its probability of being selected $\mathcal{P}(\boldsymbol{s}_i|\mathcal{S})$, that is,
\begin{align}
\label{equ:prob}
\mathcal{P}(\boldsymbol{s}_i|\mathcal{S})
=\frac{\exp{\left[U(\boldsymbol{s}_i|\mathcal{S})\right]}}
{\sum_{\boldsymbol{s}_k\in \mathcal{S}}\exp{\left[U(\boldsymbol{s}_k|\mathcal{S})\right]}}.
\end{align}

\subsection{Two Module Design Methods}
\label{sec:design}
The design of the position utility module is easy and intuitive. Specifically, its input is the position matrix $\mathcal{S}_{pos}$, and the output is each seller's position utility $PU(\boldsymbol{s}_i|\mathcal{S})$. In this module, the only parameter to estimate is the fixed vector $\boldsymbol\alpha$, and its estimation method is in Section \ref{sec:learning}.

The design of the attribute and the comparison utility module is more complex. As discussed in Section \ref{sec:intro}, both the adaptive weights and inter-item comparison are permutation invariant, and how to capture the  permutation-invariant property is the crucial point in the module design process. In this study, we propose two module design methods: an additive method with better interpretability, and an ANN-based method offering higher accuracy. Details of the two module design methods are as follows.


\textbf{The Additive Method.} \quad To achieve high interpretability, the basic principle of the additive method is to expand the context-dependent preferences into a series of additive and interpretable terms. 



We begin with the design of the \textbf{attribute utility module}, whose structure is shown in Fig. \ref{fig:add_design} (a). The additive method assumes that each item individually and independently makes a contribution to the weights, and that the weights $\boldsymbol\beta(\mathcal{S}_{att})$ can be approximated by additively combining these contributions together. In this module, the input is the attribute matrix, and the row vector  $g(\boldsymbol s_i)$ computes the contribution that seller $\boldsymbol s_i$ makes to the weights. The element at the $i$-th row and the $j$-th column of the matrix $g(\mathcal{S}_{att})\in\mathbb{R}^{d\times N}$ stores the contribution that seller $\boldsymbol s_j$ makes to the weight that users assign to the $i$-th attribute. Then, users' adaptive weights $\boldsymbol\beta(\mathcal{S}_{att})$ equals to the sum of all $g(\boldsymbol s_i)$ with
\begin{align}
\label{equ:beta}
\boldsymbol\beta(\mathcal{S})=\sum_{\boldsymbol s_i\in\mathcal{S}}g(\boldsymbol s_i).
\end{align}
This method can be regarded as the first-order expansion of the adaptive weights on the item level. Finally, the weights are multiplied by the attribute matrix to obtain the attribute utility $AU(\boldsymbol{s}_i|\mathcal{S}_{att})$.


Furthermore, we propose that the function  $g(\boldsymbol{s}_i)$ can be formulated as
\begin{align}
\label{equ:g}
\boldsymbol g_1 &= \text{ReLU}(\Gamma_1\boldsymbol s_i),\\ \notag
g(\boldsymbol s_i) &= \text{ReLU}\left(\Gamma_2\boldsymbol g_1\right).
\end{align}
In Eq. \eqref{equ:g}, $\Gamma_1\in\mathbb{R}^{d\times d}$ and $\Gamma_2\in\mathbb{R}^{d\times d}$ are two matrices to implement linear transformations. ReLU function is used to ensure that the weights are always positive. In our work, we repeat the linear transformation and ReLU twice, as experiments show that one operation is not enough to learn the adaptive weights.

To show that the proposed additive method has the permutation-invariant property, note that theoretical analysis in \cite{zaheer2017deep} show that any permutation-invariant function can be written in the form $\rho\left(\sum_{x\in \mathcal{X}}\phi (x)\right)$, where $\phi$ and $\rho$ are transformations determined by the task of interest. The proposed method in Eq. (\ref{equ:beta}) obeys this form with $\phi=g(\cdot)$ and $\rho(x)=x$. Additionally, the proposed method is similar to the LCL model \cite{tomlinson2020learning}, which calculates the adaptive weights through
$\boldsymbol\beta(\mathcal{S}_{att})=\frac{1}{N}\sum_{\boldsymbol s_i\in\mathcal{S}}A\boldsymbol s_i=A\boldsymbol s_m$, where $\boldsymbol s_m$ is the average seller whose attributes equal the average value of all sellers. LCL is interpretable and is permutation invariant with $\phi=A\boldsymbol s_i$ and $\rho(x)=x$. However, LCL cannot distinguish between two markets where the sellers' attributes are different but the ``average'' seller is the same. In our work, the proposed additive method does not have the ``average'' seller issue, and experiments show that it has better predictive performance than LCL. In summary, our proposed additive method has the permutation-invariant property, while it also avoids the ``average'' seller issue and achieves better performance.

\begin{figure}[!t]
\begin{minipage}[b]{0.5\linewidth}
  \centering
  \centerline{\includegraphics[width=7.6cm]{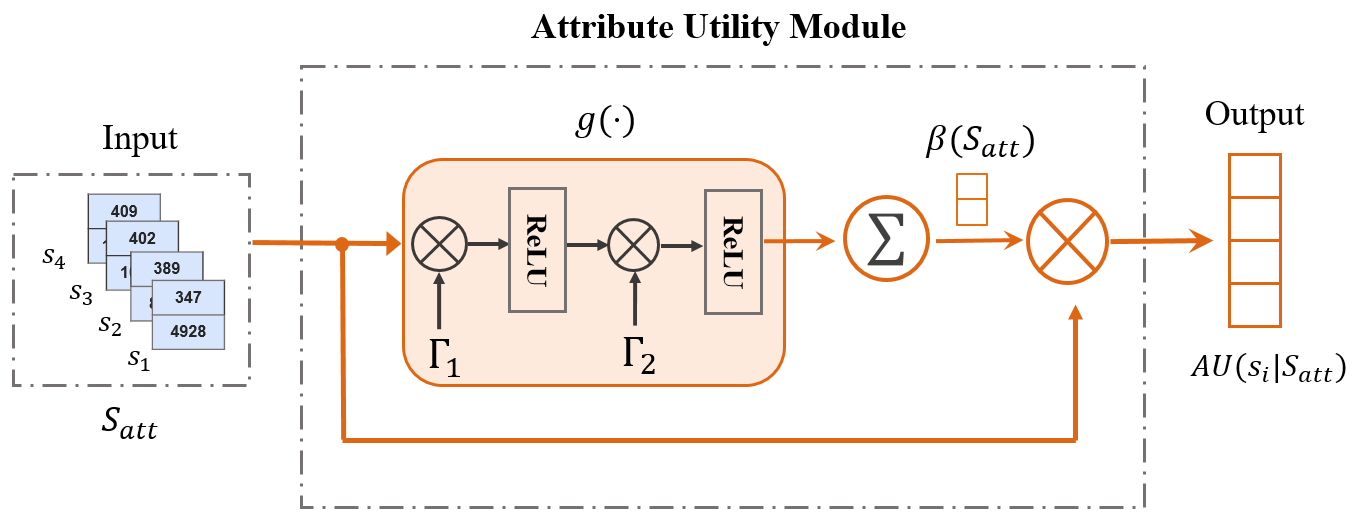}}
  \centerline{(a) }\medskip
\end{minipage}
\hfill
\begin{minipage}[b]{0.47\linewidth}
  \centering
  \centerline{\includegraphics[width=8cm]{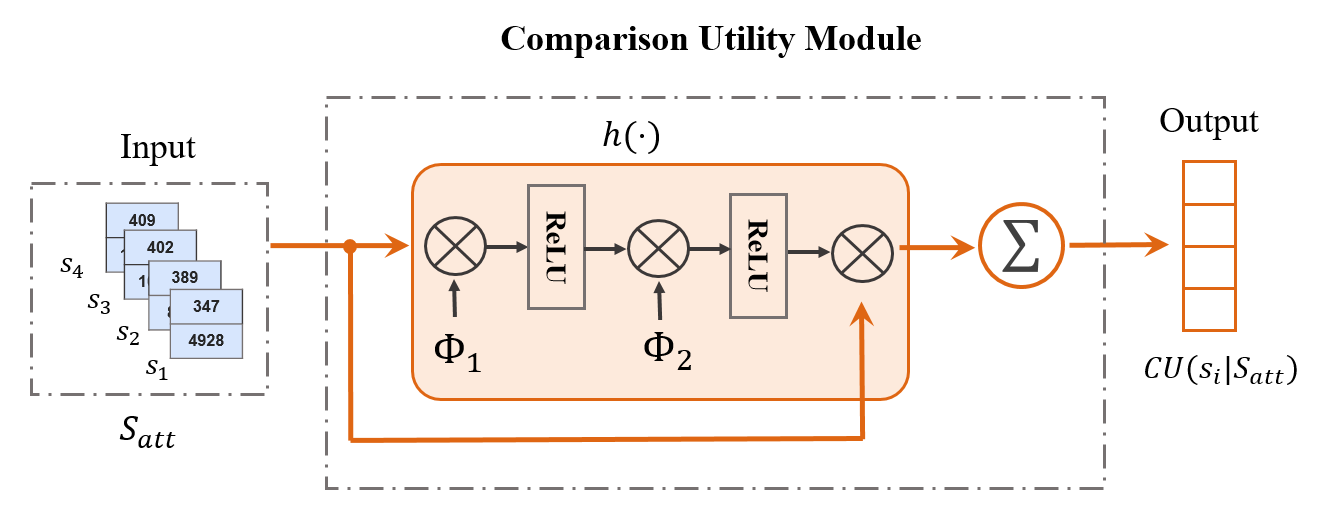}}
  \centerline{(b) }\medskip
\end{minipage}
\caption{ (a) The attribute utility module and (b) the comparison utility module in the additive method.}
\label{fig:add_design}
\end{figure}

Next, we discuss the design of the \textbf{comparison utility module}, whose structure is shown in Fig. \ref{fig:add_design} (b). The additive method approximates the inter-item comparison by the pairwise comparisons between two items, and assumes that the comparison utility of an item can be estimated by adding all pairwise comparisons together. In this module, the input is the attribute matrix $\mathcal{S}_{att}$, followed by a function  $h(\boldsymbol{s}_i,\boldsymbol{s}_j)$ to calculate the pairwise comparison between seller $\boldsymbol{s}_i$ and seller $\boldsymbol{s}_j$. The element at the $i$-th row and $j$-th column of the matrix $h(\mathcal{S}_{att},\mathcal{S}_{att})\in\mathbb{R}^{N\times N}$ stores the comparison utility of seller $\boldsymbol{s}_i$ caused by the pairwise comparison between seller $\boldsymbol{s}_i$ and $\boldsymbol{s}_j$. The comparison utility $CU(\boldsymbol{s}_i|\mathcal{S}_{att})$ equals to the sum of all pairwise comparisons with
\begin{align}
\label{equ:cu}
CU(\boldsymbol{s}_i|\mathcal{S})=\sum_{\boldsymbol{s}_j\in \mathcal{S} } h(\boldsymbol{s}_i,\boldsymbol{s}_j).
\end{align}

Furthermore, we propose that the operation $h(\boldsymbol{s}_i,\boldsymbol{s}_j)$ can be formulated as
\begin{align}
\label{equ:h}
\boldsymbol h_1(\boldsymbol s_i) &= \text{LeakyReLU}\left(\Phi_1\boldsymbol s_i\right),\\    \notag
\boldsymbol h_2(\boldsymbol s_i) &= \text{LeakyReLU}\left[\Phi_2\boldsymbol h_1(\boldsymbol s_i)\right],\\   \notag
h(\boldsymbol{s}_i,\boldsymbol{s}_j) &=\boldsymbol h_2(\boldsymbol s_i)^T\cdot\boldsymbol{s}_j,
\end{align}
where $\Phi_1\in\mathbb{R}^{d\times d}$ and $\Phi_2\in\mathbb{R}^{d\times d}$ are two matrices to implement linear transformations. In Eq. \eqref{equ:h}, there are four points to note. First of all, consider the competition of two sellers $\boldsymbol{s}_i$ and $\boldsymbol{s}_j$, when $\boldsymbol{s}_i$ is a decoy item (slightly inferior in all attributes), it increases users' preferences for $\boldsymbol{s}_j$, while when $\boldsymbol{s}_i$ is an similar item, it decreases users' preferences for $\boldsymbol{s}_j$. Therefore, we propose that the pairwise comparison $h(\boldsymbol{s}_i,\boldsymbol{s}_j)$ can be both positive and negative, and use LeakyReLU rather than ReLU as the activation function. Also, we propose that the inter-item comparison module has different characteristics in different scenarios with proper selection of $\Phi_1$ and $\Phi_2$. For example, the inter-item comparison can be asymmetric with $h(\boldsymbol{s}_i,\boldsymbol{s}_j)\neq h(\boldsymbol{s}_j,\boldsymbol{s}_i)$. This happens when the decoy item ($\boldsymbol{s}_i$) increases users' preferences for another item ($\boldsymbol{s}_j$) while at the same time decreases users' preferences for itself, i.e., $h(\boldsymbol{s}_i,\boldsymbol{s}_j)<0$ and $h(\boldsymbol{s}_j,\boldsymbol{s}_i)>0$. In addition, the inter-item comparison can be non-zero-sum with $h(\boldsymbol{s}_i,\boldsymbol{s}_j)+h(\boldsymbol{s}_j,\boldsymbol{s}_i)\neq 0$, as similar items reduce users' preferences for each other with $h(\boldsymbol{s}_i,\boldsymbol{s}_j)<0$ and $h(\boldsymbol{s}_j,\boldsymbol{s}_i)<0$. Thirdly, we preserve the term $h(\boldsymbol{s}_i,\boldsymbol{s}_i)$, because it facilitates the theoretical proof of the effectiveness of the method, and as will be shown in Section \ref{sec:test}, experiments show that keeping this term in Eq. (\ref{equ:cu}) achieves higher accuracy. Finally, the proposed method in Eq. (\ref{equ:cu}) has the permutation-invariant property, and the analysis is similar to the attribute utility module.

With the above settings, the total utility function in additive method can be expressed as
\begin{align}
\label{equ:uti_add}
U(\boldsymbol{s}_i|\mathcal{S})= \left[\sum_{\boldsymbol s_j\in\mathcal{S}}g(\boldsymbol s_j)\right]^T\cdot \boldsymbol{s}_i
+\sum_{\boldsymbol{s}_j\in \mathcal{S}}h(\boldsymbol{s}_j,\boldsymbol{s}_i)+\boldsymbol\alpha^T\cdot pos(\boldsymbol{s}_i|\mathcal{S}).
\end{align}

\textbf{The ANN-based Method.}\quad In pursuit of high flexibility and accuracy, the ANN-based method utilizes the artificial neural network to learn users' context-dependent preferences. Among many neural network architectures, the deep sets \cite{zaheer2017deep} has the permutation-invariant property, thus it is used as an example to illustrate the attribute and the comparison utility module in the ANN-based method. The structures of the two modules are shown in Fig. \ref{fig:G_and_H}. The design of the position utility module is the same as that in the additive method and is omitted here.

For the attribute utility module, its input is the attribute matrix $\mathcal{S}_{att}$. Given the input $\mathcal{S}_{att}$, the deep sets is used to learn users' adaptive weights $\boldsymbol\beta(\mathcal{S}_{att})$. Then, these weights are multiplied with the attribute matrix $\mathcal{S}_{att}$ to obtain the attribute utility $AU(\boldsymbol{s}_i|\mathcal{S}_{att})$. Specifically, the deep sets has three components: two multilayer perceptrons and a sum operation in between. Following the work in \cite{zaheer2017deep}, to achieve the permutation-invariant property, we design the deep sets in the attribute utility module to implement $\rho\left(\sum_{x\in \mathcal{X}}\phi (x)\right)$, where $\rho$ and $\phi$ both are implemented by the multilayer perceptrons, as shown in Fig. \ref{fig:G_and_H} (a). In our work, we let the multilayer perceptrons have three layers, and experiments show that more layers do not give a significant gain in model accuracy. In addition, to ensure $\boldsymbol\beta(\mathcal{S}_{att})>0$, the activation functions in the attribute utility module are all ReLU functions.

Similarly, for the comparison utility module, its input is the attribute matrix $\mathcal{S}_{att}$, followed by the deep sets to calculate the inter-item comparison and obtain the comparison utility of each item. The multilayer perceptrons in the deep sets have three layers, and  the activation function here is the LeakyReLU function. 



\begin{figure}[!t]
\begin{minipage}[b]{.50\linewidth}
  \centering
  \centerline{\includegraphics[width=8.5cm]{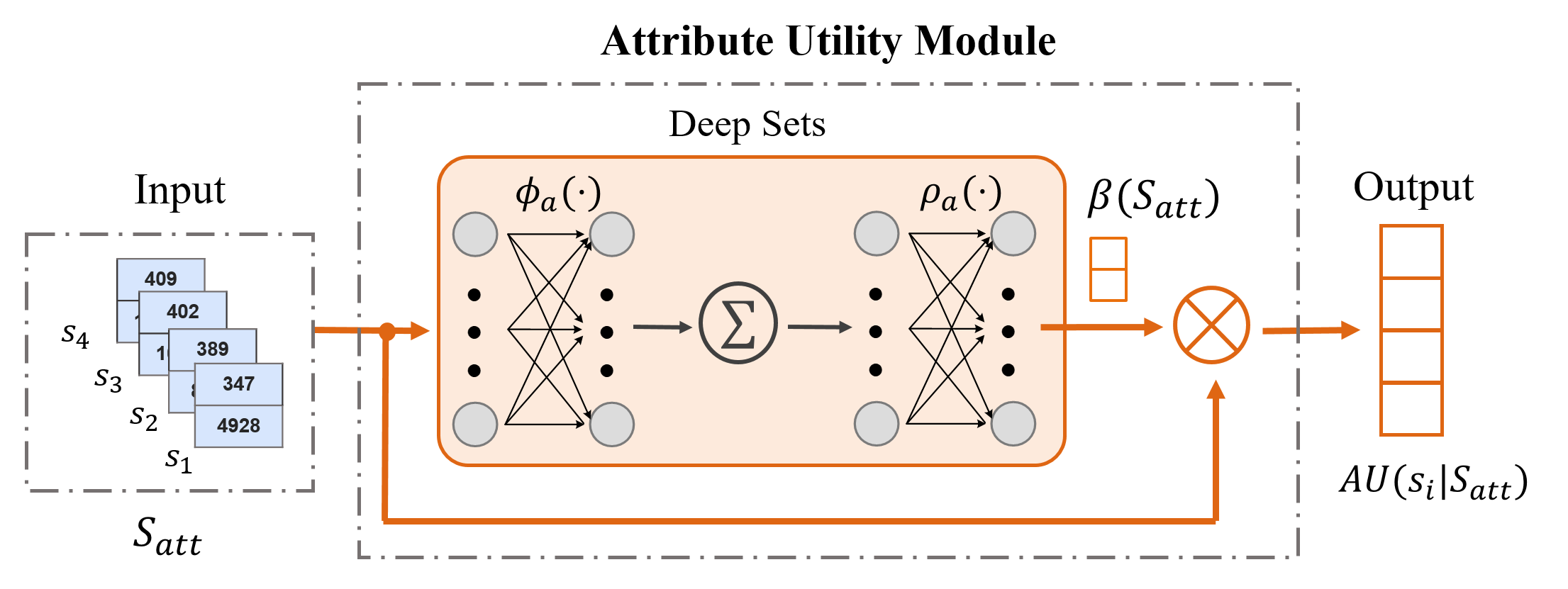}}
  \centerline{(a) }\medskip
\end{minipage}
\hfill
\begin{minipage}[b]{0.47\linewidth}
  \centering
  \centerline{\includegraphics[width=8.0cm]{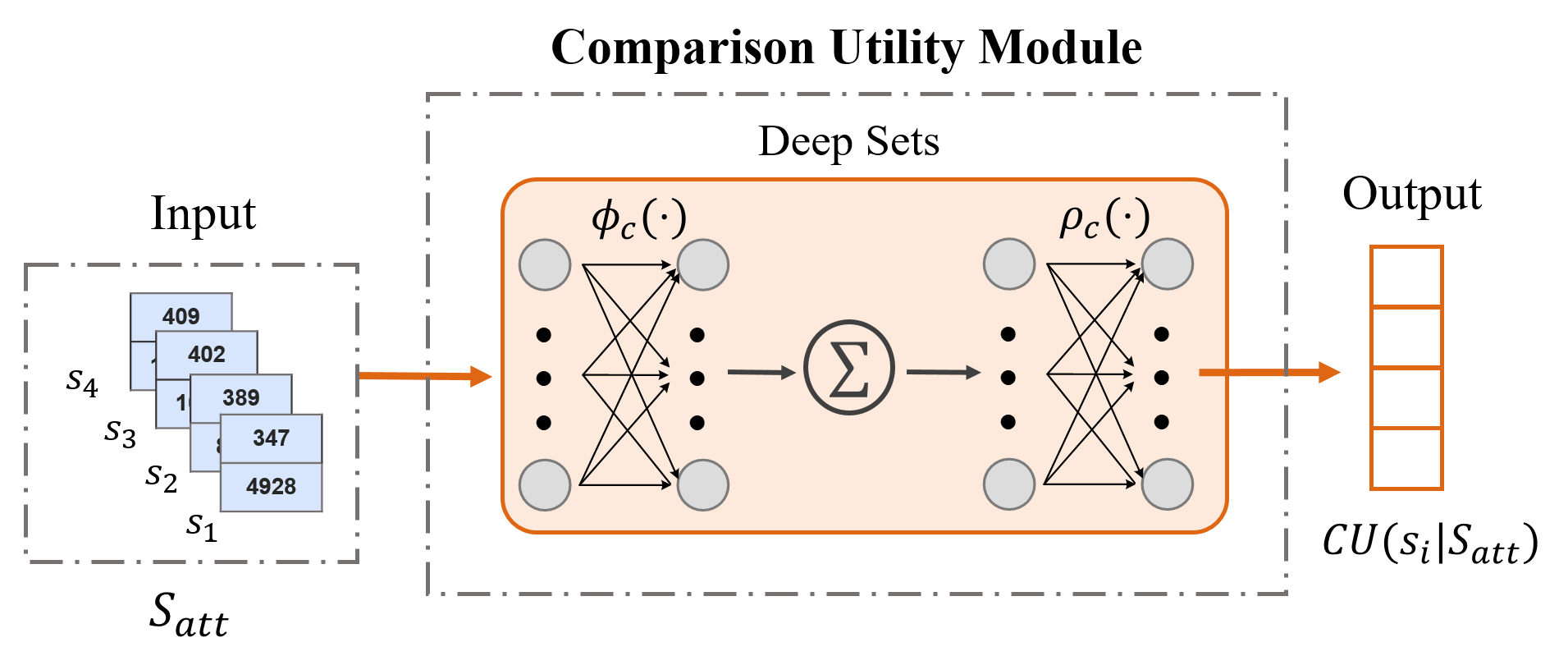}}
  \centerline{(b) }\medskip
\end{minipage}
\caption{(1) The attribute utility module and (b) comparison utility module in the ANN-based method.}
\label{fig:G_and_H}
\end{figure}

\subsection{The Learning Algorithm}
\label{sec:learning}

\begin{algorithm}[!t]
  \caption{The Learning Algorithm}
  \begin{algorithmic}[1]
  \Require $T$ records, and the number of iterations $K$
  \Ensure The optimized model
  \State Initialize parameters
  \For{$k=1$ to $K$}
  \If{$k<K_{init}$}
  \State $\bar B=0$
  \Else
  \State Compute $\bar B$ according to Eq. (\ref{equ:update})
  \EndIf
  \State Obtain $AU(\boldsymbol{s}_i|\mathcal{S}_{att})$, $CU(\boldsymbol{s}_i|\mathcal{S}_{att})$, and $PU(\boldsymbol{s}_i|\mathcal{S}_{pos})$ according to Eq. (\ref{equ:uti_add})
  \State Refine $CU(\boldsymbol{s}_i|\mathcal{S}_{att})$ and $PU(\boldsymbol{s}_i|\mathcal{S}_{pos})$ according to Eq. (\ref{equ:bar})
  \State Get $\mathcal{P}(\boldsymbol s_i|\mathcal{S})$ according to Eq. (\ref{equ:prob})
  \State Calculate the current loss $L$ according to Eq. (\ref{equ:loss})
  \State Update parameters that need to be optimized
  \EndFor
  \end{algorithmic}
   \label{algo:training}
\end{algorithm}

\textbf{Parameters to Learn.} \quad For the position utility module, the only parameter to be optimized is $\boldsymbol\alpha$, which stores users' preferences for display positions. While for the other two utility modules, the parameters to be optimized in the additive method are different from that in the ANN-based method. Specifically, in the additive method, parameters to be optimized include $\Phi_1$ and $\Phi_2$ in the comparison utility module, and $\Gamma_1$ and $\Gamma_2$ in the attribute utility module. While in the ANN-based method, parameters to be optimized are the variables in the multilayer perceptrons.

\textbf{The Imbalance Issue.} \quad Note that the complexity of the three utility modules is imbalanced. The position utility module includes only one vector multiplication, thus its structure is the simplest. On the contrary, the attribute utility module needs to first obtain $\boldsymbol\beta(\mathcal{S}_{att})$ and then compute the attribute utility, thus its structure is the most complex. Since the complexity of the modules may affect the convergence speeds of the learning algorithm, the imbalance issue may occur in the training process. That is, it is possible that parameters in the position utility module or the comparison utility module are already optimized, while parameters in the attribute utility module are fully or partially ignored. The imbalance issue among utility modules is observed in both the additive and ANN-based methods.


The imbalance issue may cause poor performance in new markets. In our real user test, we observe the following phenomenon caused by the imbalance issue. Consider the scenario that all markets in the training set contain four items, which are sorted in the ascending order of their reputations. A user always selects sellers with the highest reputations. Note that the selected sellers also happen to be displayed at the fourth position. However, classical optimization algorithms, such as Adam or Adagrad \cite{kingma2014adam}, often optimize the position utility module alone, and neglect parameters in the attribute utility module. When given a new market with five sellers, the user selects the fifth seller with the highest reputation, but the learned model still regards the fourth seller as the best choice and outputs inaccurate prediction results. Similarly, we also observe in our test that these classical optimization algorithms sometimes optimize the comparison utility module alone and neglect parameters in the attribute utility module.

\textbf{The Proposed Learning Algorithm.} \quad To address the imbalance issue, we propose a learning algorithm shown in Algorithm \ref{algo:training}.

Following prior works \cite{mottini2017deep,rosenfeld2020predicting}, we use the cross entropy \cite{de2005tutorial}
\begin{align}
\label{equ:loss}
  L = -\sum_{t=1}^T\sum_{i=1}^N y_i\log \mathcal{P}(\boldsymbol s_i|\mathcal{S}),
\end{align}
as the loss function, where $\mathcal{P}(\boldsymbol s_i|\mathcal{S})$ is the probability of being selected in Eq. (\ref{equ:prob}), and $y_i$ is an indicator with $y_i=1$ when $\boldsymbol s_i$ is selected and $y_i=0$ otherwise. The unknown parameters are optimized to minimize the loss function.

To avoid the scenario where the learning algorithm continuously optimizes the position and the comparison utility modules but ignores the attribute utility module, we introduce a variable $\bar B$ to denote the upper limit of the position and the comparison utility, that is,
\begin{align}
\label{equ:bar}
CU(\boldsymbol s_i|\mathcal{S}_{att}) = \min\{CU(\boldsymbol s_i|\mathcal{S}_{att}),\bar B\}, \quad \text{and}\quad PU(\boldsymbol s_i|\mathcal{S}_{pos})=\min\{PU(\boldsymbol s_i|\mathcal{S}_{pos}),\bar B\}.
\end{align}
With $\bar B$, the loss function cannot be minimized by optimizing the position or the comparison utility module alone and the learning algorithm is forced to optimize the attribute utility module, and thus, the imbalance issue is alleviated.

Note that the selection of $\bar B$ may affect the model's performance, we propose a novel update rule to dynamically adjust its value in the training process. In the first $K_{init}$ epochs, $\bar B$ is initialized as 0. Then, after every $\Delta K$ epochs, $\bar B$ increases its value by $\Delta B$ until it reaches $B_{max}$. That is,
\begin{align}
\label{equ:update}
\bar B = \min\left\{\left\lfloor\frac{k-K_{init}}{\Delta K}\right\rfloor\cdot \Delta B, B_{max}\right\}.
\end{align}
In Eq. (\ref{equ:update}), $k$ is the current epoch during the training process. The update rule forces the learning algorithm to optimize the attribute utility module in the first $K_{init}$ epochs, and parameters in the comparison and position utility module are optimized gradually during subsequent learning epochs. Through the cross-validation experiment, we set $K_{init}=10$, $\Delta K=10$, $\Delta B=0.2$ and $B_{max}=2.0$, and analysis of the parameter selection is in Section \ref{sec:test_update}.


\section{Theoretical Analysis of the Additive Method}
\label{sec:calculation}
As discussed in Section \ref{sec:related_work}, existing MNL-based methods with the IIA issue cannot cope with preference reversals. Specifically, it is not possible to find parameters to fit the probability distributions when the preference reversal occurs. In this section, we theoretically analyze the conditions for preference reversals to occur, and demonstrate that the proposed additive method can always find parameters that fit the actual probability distributions and thus can effectively addressing preference reversals. Note that the ANN-based method is very difficult to analyze theoretically due to its low interpretability. We will use simulations to show its effectiveness in predicting preference reversals in this work, and plan to theoretically study its performance in our future work. In addition, the restricted Boltzmann machine-based method \cite{osogami2014restricted} has poor scalability, and needs to adjust the model structure when new sellers enter the market. We will show in this section that our proposed additive model has high scalability and can easily adapt to new markets with new/leaving sellers.


In this section, we consider the scenario where new items join the market causing users' preference reversals, and the display position of items may also change. For other scenarios, such as preference reversals due to the withdrawal of existing items, or due to the withdrawal of existing items together with the addition of new items, the analysis is similar and omitted here. In the following, we begin with a mathematical definition of preference reversals, and then theoretically analyze the scalability and effectiveness of the proposed additive method.

\subsection{Definition of Preference Reversals}
Consider a market $\mathcal{S}$ with at least two sellers, i.e., $\boldsymbol{s}_A$ and $\boldsymbol{s}_B$, and the probabilities of $\boldsymbol{s}_A$ and $\boldsymbol{s}_B$ being selected is $\mathcal{P}(\boldsymbol{s}_A|\mathcal{S})$ and $\mathcal{P}(\boldsymbol{s}_B|\mathcal{S})$, respectively. Suppose a new item $\boldsymbol{s}_C$ joins the market, and the probabilities of being selected become $\mathcal{P}(\boldsymbol{s} _A|\mathcal{S}\cup\boldsymbol{s}_C)$
and $\mathcal{P}(\boldsymbol{s}_B|\mathcal{S}\cup\boldsymbol{s}_C)$, respectively. Then, preference reversals occur when the order of the probabilities of being selected for seller $\boldsymbol{s}_A$ and $\boldsymbol{s}_B$ is reversed after the addition of seller $\boldsymbol{s}_C$. That is,  $\mathcal{P}(\boldsymbol{s}_A|\mathcal{S})>\mathcal{P}(\boldsymbol{s}_B|\mathcal{S})$ and $\mathcal{P}(\boldsymbol{s}_A|\mathcal{S}\cup\boldsymbol{s}_C)
<\mathcal{P}(\boldsymbol{s}_B|\mathcal{S}\cup\boldsymbol{s}_C)$, or $\mathcal{P}(\boldsymbol{s}_A|\mathcal{S})<\mathcal{P}(\boldsymbol{s}_B|\mathcal{S})$ and $\mathcal{P}(\boldsymbol{s}_A|\mathcal{S}\cup\boldsymbol{s}_C)
>\mathcal{P}(\boldsymbol{s}_B|\mathcal{S}\cup\boldsymbol{s}_C)$.
The two cases can be summarized as
\begin{align}
\label{equ:pref_rvs}
\left[\mathcal{P}(\boldsymbol{s}_A|\mathcal{S})-\mathcal{P}(\boldsymbol{s}_B|\mathcal{S})\right]
\cdot
\left[\mathcal{P}(\boldsymbol{s}_A|\mathcal{S}\cup\boldsymbol{s}_C)
-\mathcal{P}(\boldsymbol{s}_B|\mathcal{S}\cup\boldsymbol{s}_C)\right]<0.
\end{align}
For a model to accurately predict the preference reversal with the addition of the new seller $\boldsymbol{s}_C$, the sufficient and necessary condition is that it is possible to find parameters satisfying Eq. (\ref{equ:pref_rvs}).

However, Proposition \ref{prop:mnl} in Section \ref{sec:appendix} states that the IIA issue in prior MNL-based works \cite{mcfadden1973conditional,burges2005learning,fine2001efficient,chen2018behavior2vec} contradicts Eq.  (\ref{equ:pref_rvs}), and prevents these models from predicting preference reversals.

\subsection{Effectiveness of The Additive Method}

We first discuss the scalability of the additive method, as this scalability analysis will be used in the proof of the effectiveness of the additive method.

\textbf{The Scalability.}\quad Consider the scenario where a new seller $\boldsymbol{s}_{C}$ joins the market $\mathcal{S}$. Define $\lambda(\boldsymbol{s} _i|\mathcal{S})=\exp{\left[U(\boldsymbol{s}_i|\mathcal{S})\right]}$, where $U(\boldsymbol{s}_i|\mathcal{S})$ is in Eq. \eqref{equ:uti_add}. Note that $\lambda(\boldsymbol{s} _i|\mathcal{S})$ is the numerator of $\mathcal{P}(\boldsymbol{s}_i|\mathcal{S})$. Since the denominator of $\mathcal{P}(\boldsymbol{s}_i|\mathcal{S})$ is the same for all sellers in a market, $\lambda(\boldsymbol{s} _i|\mathcal{S})$ is larger when the probability  $\mathcal{P}(\boldsymbol{s}_i|\mathcal{S})$ is higher. Then, we have the following Lemma  \ref{lemma:scalability}. It shows that the additive method has good scalability, as we only need to add an update term $f(\boldsymbol{s}_i,\boldsymbol{s} _{C})$ to $\lambda (\boldsymbol{s}_i|\mathcal{S})$ when a new seller $\boldsymbol{s} _{C}$ joins the market.

\begin{lemma}
\label{lemma:scalability}
If a new seller $\boldsymbol{s}_{C}$ joins the market $\mathcal{S}$, then
\begin{align}
\label{equ:update_lambda_pos_changes}
\lambda(\boldsymbol{s}_i|\mathcal{S}\cup\boldsymbol{s}_{C})=\lambda (\boldsymbol{s}_i|\mathcal{S})\cdot f(\boldsymbol{s}_{i},\boldsymbol{s}_{C}),
\end{align}
where $f(\boldsymbol{s}_i,\boldsymbol{s}_{C})\triangleq\exp{\left[g(\boldsymbol{ s}_{C})^T\cdot \boldsymbol{s}_i+h(\boldsymbol{s}_i,\boldsymbol{s}_{C})
+\boldsymbol\alpha^T\cdot \Delta pos(\boldsymbol{s}_i,\boldsymbol{s}_{C} )\right]}$. Here, $\Delta pos(\boldsymbol{s}_i,\boldsymbol{s}_{C})=pos(\boldsymbol{s}_i|\mathcal{S}\cup\boldsymbol{s}_{C})
-pos(\boldsymbol{s}_i|\mathcal{S})$ is a change in the display position caused by the addition of seller $\boldsymbol{s}_{C}$, and $f(\boldsymbol{s}_i,\boldsymbol{s} _{C})$ can be seen as the exponential of the change in the total utility of seller $\boldsymbol{s}_i$ caused by the addition of seller $\boldsymbol{s}_{C}$.
\end{lemma}

\begin{proof}
For the additive method, define $\lambda(\boldsymbol{s} _i|\mathcal{S})=\exp{\left[U(\boldsymbol{s}_i|\mathcal{S})\right]}$ based on Eq. \eqref{equ:uti_add}, then we have
\begin{align}
\label{equ:ori_lambda}
\lambda(\boldsymbol{s}_i|\mathcal{S})&=\exp\left\{\left[\sum_{\boldsymbol s_j\in\mathcal{S}}g(\boldsymbol s_j)\right]^T\cdot \boldsymbol{s}_i+\sum_{\boldsymbol{s}_j\in \mathcal{S}}h(\boldsymbol{s}_i,\boldsymbol{s}_j)+
\boldsymbol\alpha^T\cdot pos(\boldsymbol{s}_i|\mathcal{S})\right\}
\end{align}
When a new seller $\boldsymbol{s}_{\text{new}}$ joins the market $\mathcal{S}$, and the display position of seller $\boldsymbol{s} _i$ is $pos(\boldsymbol{s} _i|\mathcal{S}\cup\boldsymbol{s}_{\text{new}} )$ after the addition of the new seller. Then,
\begin{align}
\label{equ:lambda_new}
\lambda(\boldsymbol{s}_i|\mathcal{S}\cup\boldsymbol{s}_{C})&=\exp\left\{\left[\sum_{\boldsymbol s_j\in\mathcal{S}\cup\boldsymbol{s}_{C}}g(\boldsymbol s_j)\right]^T\cdot \boldsymbol{s}_i+\sum_{\boldsymbol{s}_j\in \mathcal{S}\cup\boldsymbol{s}_{C}}h(\boldsymbol{s}_i,\boldsymbol{s}_j)+
\boldsymbol\alpha^T\cdot pos(\boldsymbol{s}_i|\mathcal{S}\cup\boldsymbol{s}_{C})\right\}      \\    \notag
&= \exp\left\{\left[\sum_{\boldsymbol s_j\in\mathcal{S}}g(\boldsymbol s_j)\right]^T\cdot \boldsymbol{s}_i+\sum_{\boldsymbol{s}_j\in \mathcal{S}}h(\boldsymbol{s}_i,\boldsymbol{s}_j)+
\boldsymbol\alpha^T\cdot pos(\boldsymbol{s}_i|\mathcal{S})\right\} \\   \notag
&\quad \quad \cdot \exp{\left\{g(\boldsymbol{ s}_{C})^T\cdot \boldsymbol{s}_i+h(\boldsymbol{s}_i,\boldsymbol{s}_{C})
+\boldsymbol\alpha^T\cdot \left[pos(\boldsymbol{s}_i|\mathcal{S}\cup\boldsymbol{s}_{C})
-pos(\boldsymbol{s}_i|\mathcal{S})\right]\right\}} \\   \notag
&=\lambda (\boldsymbol{s}_i|\mathcal{S})\cdot f(\boldsymbol{s}_{i},\boldsymbol{s}_{C}).
\end{align}
Eq. \eqref{equ:lambda_new} means that $\lambda (\boldsymbol{s}_i|\mathcal{S})$ is simply multiplied by an update term $f(\boldsymbol{s}_i,\boldsymbol{s} _{C})$ when a new seller $\boldsymbol{s}_{C}$ joins the market.
\end{proof}


\textbf{Conditions For Preference Reversals To Occur.} \quad Based on the definition of $\lambda(\boldsymbol{s} _i|\mathcal{S})$, Eq. (\ref{equ:pref_rvs}) is equivalent to
\begin{align}
\label{equ:lambda_pref_rvs}
\left[\lambda(\boldsymbol{s}_A|\mathcal{S})-\lambda(\boldsymbol{s}_B|\mathcal{S})\right]
\cdot
\left[\lambda(\boldsymbol{s}_A|\mathcal{S}\cup\boldsymbol{s}_C)
-\lambda(\boldsymbol{s}_B|\mathcal{S}\cup\boldsymbol{s}_C)\right]<0.
\end{align}
We analyze the sufficient and necessary conditions for Eq. \eqref{equ:lambda_pref_rvs} to hold, and obtain Theorem \ref{theo:condition}. From Theorem \ref{theo:condition}, if users originally prefer seller $\boldsymbol{s}_A$ to $\boldsymbol{s}_B$ in market $\mathcal{S}$, then preference reversals occur when the addition of seller $\boldsymbol{s}_{C}$ increases users' preference for seller $\boldsymbol{s}_{B}$ more than $\boldsymbol{s}_{A}$, and the ratio between the update terms is greater than the ratio between users' original preferences. Similarly, if users originally prefer seller $\boldsymbol{s}_B$ to $\boldsymbol{s}_A$ in market $\mathcal{S}$, then preference reversals occur when the addition of seller $\boldsymbol{s}_{C}$ increases users' preference for seller $\boldsymbol{s}_{B}$ less than $\boldsymbol{s}_{A}$, and the ratio between the update terms is less than the ratio between users' original preferences.

\begin{theorem}
\label{theo:condition}
For the additive method, the sufficient and necessary condition for Eq. (\ref{equ:lambda_pref_rvs}) to hold is that it is possible to find parameters satisfying
\begin{align}
\label{equ:condition}
\frac{f(\boldsymbol{s}_{B},\boldsymbol{s}_{C})}{f(\boldsymbol{s}_{A},\boldsymbol{s}_{C})}>
\frac{\lambda(\boldsymbol{s}_A|\mathcal{S})}{\lambda(\boldsymbol{s}_B|\mathcal{S})}>1 \quad \text{or} \quad
\frac{f(\boldsymbol{s}_{B},\boldsymbol{s}_{C})}{f(\boldsymbol{s}_{A},\boldsymbol{s}_{C})}<
\frac{\lambda(\boldsymbol{s}_A|\mathcal{S})}{\lambda(\boldsymbol{s}_B|\mathcal{S})}<1.
\end{align}
\end{theorem}

\begin{proof}
Based on Lemma  \ref{lemma:scalability}, we have
\begin{align}
\label{equ:lambda_AB}
\lambda(\boldsymbol{s}_A|\mathcal{S}\cup\boldsymbol{s}_{C})=\lambda (\boldsymbol{s}_A|\mathcal{S})\cdot f(\boldsymbol{s}_{A},\boldsymbol{s}_{C}), \quad \text{and} \quad
\lambda(\boldsymbol{s}_B|\mathcal{S}\cup\boldsymbol{s}_{C})=\lambda (\boldsymbol{s}_B|\mathcal{S})\cdot f(\boldsymbol{s}_{B},\boldsymbol{s}_{C}).
\end{align}

(1) When $\lambda(\boldsymbol{s}_A|\mathcal{S})-\lambda(\boldsymbol{s}_B|\mathcal{S})>0$, Eq. (\ref{equ:lambda_pref_rvs}) is equivalent to $\left[\lambda(\boldsymbol{s}_A|\mathcal{S}\cup\boldsymbol{s}_C)
-\lambda(\boldsymbol{s}_B|\mathcal{S}\cup\boldsymbol{s}_C)\right]<0$. Combining with Eq. \eqref{equ:lambda_AB}, we have
\begin{align}
\label{equ:proof_lambda}
&\lambda (\boldsymbol{s}_A|\mathcal{S})\cdot f(\boldsymbol{s}_{A},\boldsymbol{s}_{C})-\lambda (\boldsymbol{s}_B|\mathcal{S})\cdot f(\boldsymbol{s}_{B},\boldsymbol{s}_{C})<0 \\  \notag
\Leftrightarrow \quad  &\lambda (\boldsymbol{s}_B|\mathcal{S})\cdot f(\boldsymbol{s}_{B},\boldsymbol{s}_{C}) > \lambda (\boldsymbol{s}_A|\mathcal{S})\cdot f(\boldsymbol{s}_{A},\boldsymbol{s}_{C})\\  \notag
\Leftrightarrow \quad & \frac{f(\boldsymbol{s}_{B},\boldsymbol{s}_{C})}{f(\boldsymbol{s}_{A},\boldsymbol{s}_{C})}
>\frac{\lambda(\boldsymbol{s}_A|\mathcal{S})}{\lambda(\boldsymbol{s}_B|\mathcal{S})} > 1.
\end{align}

If $\lambda(\boldsymbol{s}_A|\mathcal{S})-\lambda(\boldsymbol{s}_B|\mathcal{S})>0$, that is, users originally prefer seller $\boldsymbol{s}_A$ to $\boldsymbol{s}_B$ in market $\mathcal{S}$. Let ${\lambda(\boldsymbol{s}_A|\mathcal{S})}\big /{\lambda(\boldsymbol{s}_B|\mathcal{S})}$ be the ratio between users' original preferences for the two sellers. After the addition of seller $\boldsymbol{s}_C$, the update term of users' preferences for seller $\boldsymbol{s}_A$ is $f(\boldsymbol{s}_{A},\boldsymbol{s}_{C})$, and the update term for seller $\boldsymbol{s}_B$ is $f(\boldsymbol{s}_{B},\boldsymbol{s}_{C})$. Let ${f(\boldsymbol{s}_{B},\boldsymbol{s}_{C})}\big /{f(\boldsymbol{s}_{A},\boldsymbol{s}_{C})}$ be the ratio between the two update terms. Then, the preference reversal occurs when the ratio between the update terms is greater than the ratio between users' original preferences, i.e., ${f(\boldsymbol{s}_{B},\boldsymbol{s}_{C})}\big /{f(\boldsymbol{s}_{A},\boldsymbol{s}_{C})}>{\lambda(\boldsymbol{s}_A|\mathcal{S})}\big /{\lambda(\boldsymbol{s}_B|\mathcal{S})}$ .

(2) When $\lambda(\boldsymbol{s}_A|\mathcal{S})-\lambda(\boldsymbol{s}_B|\mathcal{S})<0$, Eq. (\ref{equ:lambda_pref_rvs}) is equivalent to $\left[\lambda(\boldsymbol{s}_A|\mathcal{S}\cup\boldsymbol{s}_C)
-\lambda(\boldsymbol{s}_B|\mathcal{S}\cup\boldsymbol{s}_C)\right]>0$. Similar to the analysis for the first case, we have
\begin{align}
\frac{f(\boldsymbol{s}_{B},\boldsymbol{s}_{C})}{f(\boldsymbol{s}_{A},\boldsymbol{s}_{C})}
<\frac{\lambda(\boldsymbol{s}_A|\mathcal{S})}{\lambda(\boldsymbol{s}_B|\mathcal{S})}< 1.
\end{align}
If $\lambda(\boldsymbol{s}_A|\mathcal{S})<\lambda(\boldsymbol{s}_B|\mathcal{S})$, that is, users originally prefer seller $\boldsymbol{s}_B$ to $\boldsymbol{s}_A$ in market $\mathcal{S}$. Then, the preference reversal occurs when the ratio between the update terms is less than the ratio between users' original preferences, i.e., ${f(\boldsymbol{s}_{B},\boldsymbol{s}_{C})}\big /{f(\boldsymbol{s}_{A},\boldsymbol{s}_{C})}<{\lambda(\boldsymbol{s}_A|\mathcal{S})}\big /{\lambda(\boldsymbol{s}_B|\mathcal{S})}$ .
\end{proof}

\textbf{The Effectiveness.} \quad We find possible solutions to Eq. \eqref{equ:condition}, and summarize it as Theorem \ref{theo:effectiveness}. Specifically, for the attribute utility module, we calculate the inner product of  $\left[\boldsymbol{s}_A-\boldsymbol{s}_B\right]$ and $\sum_{\boldsymbol s_j\in\mathcal{S}}g(\boldsymbol s_j)$ in market $\mathcal{S}$, and the inner product of $\left[\boldsymbol{s}_A-\boldsymbol{s}_B\right]$ and $\sum_{\boldsymbol s_j\in\mathcal{S}\cup\boldsymbol{s}_C}g(\boldsymbol s_j)$ in the new market of $\mathcal{S}\cup\boldsymbol{s}_C$. For the comparison utility module, we calculate the inner product of $\left[\boldsymbol h_2(\boldsymbol s_B)-\boldsymbol  h_2(\boldsymbol s_A)\right]$ and $\sum_{\boldsymbol{s}_j\in \mathcal{S}}\boldsymbol{s}_j$ in the market $\mathcal{S}$, and the inner product of $\left[\boldsymbol h_2(\boldsymbol s_B)-\boldsymbol  h_2(\boldsymbol s_A)\right]$ and $\sum_{\boldsymbol{s}_j\in \mathcal{S}\cup\boldsymbol{s}_C}\boldsymbol{s}_j$ in the new market of $\mathcal{S}\cup\boldsymbol{s}_C$. And for the position utility module, we calculate the inner product of $\boldsymbol\alpha$ and $\left[pos(\boldsymbol{s}_A|\mathcal{S})-pos(\boldsymbol{s}_B|\mathcal{S})\right]$ in the market $\mathcal{S}$, and the inner product of  $\boldsymbol\alpha$ and $\left[pos(\boldsymbol{s}_A|\mathcal{S}\cup\boldsymbol{s}_C)-pos(\boldsymbol{s}_B|\mathcal{S}\cup\boldsymbol{s}_C)\right]$ in the new market of $\mathcal{S}\cup\boldsymbol{s}_C$. Then, one possible solution to make Eq. \eqref{equ:condition} hold is that the addition of seller $\boldsymbol{s}_C$ leads to a different sign of the above inner products in all three utility modules. It means that the additive method is effective as it can always find parameters to make Eq. \eqref{equ:condition} hold by optimizing parameters in the three utility modules.

\begin{theorem}
\label{theo:effectiveness}
Based on the proposed additive method, we find one possible solution for ${f(\boldsymbol{s}_{B},\boldsymbol{s}_{C})}\big / {f(\boldsymbol{s}_{A},\boldsymbol{s}_{C})}>
{\lambda(\boldsymbol{s}_A|\mathcal{S})}\big /{\lambda(\boldsymbol{s}_B|\mathcal{S})}>1$ to hold is that $\Gamma_1$ and $\Gamma_2$ in the attribute utility module, $\Phi_1$ and $\Phi_2$ in the comparison utility module and $\boldsymbol\alpha$ in the position utility module satisfy all the following conditions
\begin{align}
\left\{\begin{aligned}
& \left[\sum_{\boldsymbol s_j\in\mathcal{S}}g(\boldsymbol s_j)\right]^T\cdot \left(\boldsymbol{s}_A-\boldsymbol{s}_B\right)>0,
\quad \text{and} \quad \left[\sum_{\boldsymbol s_j\in\mathcal{S}\cup\boldsymbol{s}_C}g(\boldsymbol s_j)\right]^T\cdot \left(\boldsymbol{s}_A-\boldsymbol{s}_B\right)<0, \\
&\left[\boldsymbol h_2(\boldsymbol s_A)-\boldsymbol  h_2(\boldsymbol s_B)\right]^T\cdot \sum_{\boldsymbol{s}_j\in \mathcal{S}}\boldsymbol{s}_j>0, \quad \text{and} \quad  \left[\boldsymbol h_2(\boldsymbol s_A)-\boldsymbol  h_2(\boldsymbol s_B)\right]^T\cdot \sum_{\boldsymbol{s}_j\in \mathcal{S}\cup\boldsymbol{s}_C}\boldsymbol{s}_j<0,\\
&\boldsymbol\alpha^T\cdot \left[pos(\boldsymbol{s}_A|\mathcal{S})-pos(\boldsymbol{s}_B|\mathcal{S})\right]>0, \quad \text{and} \quad \boldsymbol\alpha^T\cdot \left[pos(\boldsymbol{s}_A|\mathcal{S}\cup\boldsymbol{s}_C)-pos(\boldsymbol{s}_B|\mathcal{S}\cup\boldsymbol{s}_C)\right]<0,  \\
\end{aligned}\right.
\end{align}
and one possible solution for ${f(\boldsymbol{s}_{B},\boldsymbol{s}_{C})}\big / {f(\boldsymbol{s}_{A},\boldsymbol{s}_{C})}<
{\lambda(\boldsymbol{s}_A|\mathcal{S})}\big /{\lambda(\boldsymbol{s}_B|\mathcal{S})}<1$ to hold is that $\Gamma_1$, $\Gamma_2$, $\Phi_1$, $\Phi_2$ and $\boldsymbol\alpha$ satisfy
\begin{align}
\left\{\begin{aligned}
& \left[\sum_{\boldsymbol s_j\in\mathcal{S}}g(\boldsymbol s_j)\right]^T\cdot \left(\boldsymbol{s}_A-\boldsymbol{s}_B\right)<0,
\quad \text{and} \quad \left[\sum_{\boldsymbol s_j\in\mathcal{S}\cup\boldsymbol{s}_C}g(\boldsymbol s_j)\right]^T\cdot \left(\boldsymbol{s}_A-\boldsymbol{s}_B\right)>0, \\
&\left[\boldsymbol h_2(\boldsymbol s_A)-\boldsymbol  h_2(\boldsymbol s_B)\right]^T\cdot \sum_{\boldsymbol{s}_j\in \mathcal{S}}\boldsymbol{s}_j<0, \quad \text{and} \quad  \left[\boldsymbol h_2(\boldsymbol s_A)-\boldsymbol  h_2(\boldsymbol s_B)\right]^T\cdot \sum_{\boldsymbol{s}_j\in \mathcal{S}\cup\boldsymbol{s}_C}\boldsymbol{s}_j>0,\\
&\boldsymbol\alpha^T\cdot \left[pos(\boldsymbol{s}_A|\mathcal{S})-pos(\boldsymbol{s}_B|\mathcal{S})\right]<0, \quad \text{and} \quad \boldsymbol\alpha^T\cdot \left[pos(\boldsymbol{s}_A|\mathcal{S}\cup\boldsymbol{s}_C)-pos(\boldsymbol{s}_B|\mathcal{S}\cup\boldsymbol{s}_C)\right]>0.  \\
\end{aligned}\right.
\end{align}
\end{theorem}

\begin{proof}
We begin with ${f(\boldsymbol{s}_{B},\boldsymbol{s}_{C})}\big / {f(\boldsymbol{s}_{A},\boldsymbol{s}_{C})}>
{\lambda(\boldsymbol{s}_A|\mathcal{S})}\big /{\lambda(\boldsymbol{s}_B|\mathcal{S})}>1$. First for ${\lambda(\boldsymbol{s}_A|\mathcal{S})}\big /{\lambda(\boldsymbol{s}_B|\mathcal{S})}>1$, based on the definition of $\lambda(\boldsymbol{s}_i|\mathcal{S})$ in Eq. \eqref{equ:ori_lambda}, we have
\begin{align}
\frac{\lambda(\boldsymbol{s}_A|\mathcal{S})}{\lambda(\boldsymbol{s}_B|\mathcal{S})}
&= \frac{\exp\left\{\left[\sum_{\boldsymbol s_j\in\mathcal{S}}g(\boldsymbol s_j)\right]^T\cdot \boldsymbol{s}_A+\sum_{\boldsymbol{s}_j\in \mathcal{S}}h(\boldsymbol{s}_A,\boldsymbol{s}_j)+
\boldsymbol\alpha^T\cdot pos(\boldsymbol{s}_A|\mathcal{S})\right\}}
{\exp\left\{\left[\sum_{\boldsymbol s_j\in\mathcal{S}}g(\boldsymbol s_j)\right]^T\cdot \boldsymbol{s}_B+\sum_{\boldsymbol{s}_j\in \mathcal{S}}h(\boldsymbol{s}_B,\boldsymbol{s}_j)+
\boldsymbol\alpha^T\cdot pos(\boldsymbol{s}_B|\mathcal{S})\right\}} \\ \notag
&= \exp\Bigg\{\left[\sum_{\boldsymbol s_j\in\mathcal{S}}g(\boldsymbol s_j)\right]^T\cdot \left(\boldsymbol{s}_A-\boldsymbol{s}_B\right)
+\sum_{\boldsymbol{s}_j\in \mathcal{S}} \left[h(\boldsymbol{s}_A,\boldsymbol{s}_j)-h(\boldsymbol{s}_B,\boldsymbol{s}_j)\right]\\ \notag
& \qquad \qquad \qquad \qquad \qquad \qquad \qquad \quad
+\boldsymbol\alpha^T\cdot \left[pos(\boldsymbol{s}_A|\mathcal{S})-pos(\boldsymbol{s}_B|\mathcal{S})\right]\Bigg\}.
\end{align}
Recall that $h(\boldsymbol{s}_i,\boldsymbol{s}_j)=\boldsymbol h_2(\boldsymbol s_i)^T\cdot\boldsymbol{s}_j$, then
\begin{align}
&\frac{\lambda(\boldsymbol{s}_A|\mathcal{S})}{\lambda(\boldsymbol{s}_B|\mathcal{S})} = \exp\Bigg\{\left[\sum_{\boldsymbol s_j\in\mathcal{S}}g(\boldsymbol s_j)\right]^T\cdot \left(\boldsymbol{s}_A-\boldsymbol{s}_B\right)
+\left[\boldsymbol h_2(\boldsymbol s_A)-\boldsymbol  h_2(\boldsymbol s_B)\right]^T\cdot \sum_{\boldsymbol{s}_j\in \mathcal{S}}\boldsymbol{s}_j \\ \notag
& \qquad \qquad \qquad \qquad \qquad \qquad \qquad \qquad \qquad \qquad \qquad \quad
+\boldsymbol\alpha^T\cdot \left[pos(\boldsymbol{s}_A|\mathcal{S})-pos(\boldsymbol{s}_B|\mathcal{S})\right]\Bigg\}.
\end{align}
To make ${\lambda(\boldsymbol{s}_A|\mathcal{S})}\big /{\lambda(\boldsymbol{s}_B|\mathcal{S})}>1$, one possible solution is
\begin{align}
\left[\sum_{\boldsymbol s_j\in\mathcal{S}}g(\boldsymbol s_j)\right]^T\cdot \left(\boldsymbol{s}_A-\boldsymbol{s}_B\right)>0, \quad
&\left[\boldsymbol h_2(\boldsymbol s_A)-\boldsymbol  h_2(\boldsymbol s_B)\right]^T\cdot \sum_{\boldsymbol{s}_j\in \mathcal{S}}\boldsymbol{s}_j>0, \\ \notag
&\quad \quad \quad \quad \text{and} \quad \boldsymbol\alpha^T\cdot \left[pos(\boldsymbol{s}_A|\mathcal{S})-pos(\boldsymbol{s}_B|\mathcal{S})\right]>0.
\end{align}
Next, for ${f(\boldsymbol{s}_{B},\boldsymbol{s}_{C})}\big / {f(\boldsymbol{s}_{A},\boldsymbol{s}_{C})}>
{\lambda(\boldsymbol{s}_A|\mathcal{S})}\big /{\lambda(\boldsymbol{s}_B|\mathcal{S})}$, it is equivalent to $\lambda(\boldsymbol{s}_A|\mathcal{S}\cup\boldsymbol{s}_C)
-\lambda(\boldsymbol{s}_B|\mathcal{S}\cup\boldsymbol{s}_C)<0$ (see the proof of Theorem \ref{theo:condition}), that is, ${\lambda(\boldsymbol{s}_A|\mathcal{S}\cup\boldsymbol{s}_C)}\big /
{\lambda(\boldsymbol{s}_B|\mathcal{S}\cup\boldsymbol{s}_C)}<1$. Based on the definition of ${\lambda(\boldsymbol{s}_i|\mathcal{S}\cup\boldsymbol{s}_C)}$, we have
\begin{align}
\frac{\lambda(\boldsymbol{s}_A|\mathcal{S}\cup\boldsymbol{s}_C)}{\lambda(\boldsymbol{s}_i|\mathcal{S}\cup\boldsymbol{s}_C)}
&= \frac{\exp\left\{\left[\sum_{\boldsymbol s_j\in\mathcal{S}\cup\boldsymbol{s}_C}g(\boldsymbol s_j)\right]^T\cdot \boldsymbol{s}_A+\sum_{\boldsymbol{s}_j\in \mathcal{S}\cup\boldsymbol{s}_C}h(\boldsymbol{s}_A,\boldsymbol{s}_j)+
\boldsymbol\alpha^T\cdot pos(\boldsymbol{s}_A|\mathcal{S}\cup\boldsymbol{s}_C)\right\}}
{\exp\left\{\left[\sum_{\boldsymbol s_j\in\mathcal{S}\cup\boldsymbol{s}_C}g(\boldsymbol s_j)\right]^T\cdot \boldsymbol{s}_B+\sum_{\boldsymbol{s}_j\in \mathcal{S}\cup\boldsymbol{s}_C}h(\boldsymbol{s}_B,\boldsymbol{s}_j)+
\boldsymbol\alpha^T\cdot pos(\boldsymbol{s}_B|\mathcal{S}\cup\boldsymbol{s}_C)\right\}} \\ \notag
&= \exp\Bigg\{\left[\sum_{\boldsymbol s_j\in\mathcal{S}\cup\boldsymbol{s}_C}g(\boldsymbol s_j)\right]^T\cdot \left(\boldsymbol{s}_A-\boldsymbol{s}_B\right)
+\sum_{\boldsymbol{s}_j\in \mathcal{S}\cup\boldsymbol{s}_C} \left[h(\boldsymbol{s}_A,\boldsymbol{s}_j)-h(\boldsymbol{s}_B,\boldsymbol{s}_j)\right]\\ \notag
& \qquad \qquad \qquad \qquad \qquad \qquad \qquad \qquad \quad
+\boldsymbol\alpha^T\cdot \left[pos(\boldsymbol{s}_A|\mathcal{S}\cup\boldsymbol{s}_C)-pos(\boldsymbol{s}_B|\mathcal{S}\cup\boldsymbol{s}_C)\right]\Bigg\} \\ \notag
&= \exp\Bigg\{\left[\sum_{\boldsymbol s_j\in\mathcal{S}\cup\boldsymbol{s}_C}g(\boldsymbol s_j)\right]^T\cdot \left(\boldsymbol{s}_A-\boldsymbol{s}_B\right)
+\left[\boldsymbol h_2(\boldsymbol s_A)-\boldsymbol  h_2(\boldsymbol s_B)\right]^T\cdot \sum_{\boldsymbol{s}_j\in \mathcal{S}\cup\boldsymbol{s}_C}\boldsymbol{s}_j \\ \notag
& \qquad \qquad \qquad \qquad \qquad \qquad \qquad \qquad \qquad \qquad \qquad \quad
+\boldsymbol\alpha^T\cdot \left[pos(\boldsymbol{s}_A|\mathcal{S}\cup\boldsymbol{s}_C)-pos(\boldsymbol{s}_B|\mathcal{S}\cup\boldsymbol{s}_C)\right]\Bigg\}.
\end{align}
To make ${\lambda(\boldsymbol{s}_A|\mathcal{S}\cup\boldsymbol{s}_C)}\big /
{\lambda(\boldsymbol{s}_B|\mathcal{S}\cup\boldsymbol{s}_C)}<1$, one possible solution is
\begin{align}
\left[\sum_{\boldsymbol s_j\in\mathcal{S}\cup\boldsymbol{s}_C}g(\boldsymbol s_j)\right]^T\cdot \left(\boldsymbol{s}_A-\boldsymbol{s}_B\right)<0, \quad
&\left[\boldsymbol h_2(\boldsymbol s_A)-\boldsymbol  h_2(\boldsymbol s_B)\right]^T\cdot \sum_{\boldsymbol{s}_j\in \mathcal{S}\cup\boldsymbol{s}_C}\boldsymbol{s}_j<0, \\ \notag
&\quad \quad \quad \quad \text{and} \quad \boldsymbol\alpha^T\cdot \left[pos(\boldsymbol{s}_A|\mathcal{S}\cup\boldsymbol{s}_C)-pos(\boldsymbol{s}_B|\mathcal{S}\cup\boldsymbol{s}_C)\right]<0.
\end{align}
Summarizing the above discussions, conditions for ${f(\boldsymbol{s}_{B},\boldsymbol{s}_{C})}\big / {f(\boldsymbol{s}_{A},\boldsymbol{s}_{C})}>
{\lambda(\boldsymbol{s}_A|\mathcal{S})}\big /{\lambda(\boldsymbol{s}_B|\mathcal{S})}>1$ to hold is
\begin{align}
\left\{\begin{aligned}
& \left[\sum_{\boldsymbol s_j\in\mathcal{S}}g(\boldsymbol s_j)\right]^T\cdot \left(\boldsymbol{s}_A-\boldsymbol{s}_B\right)>0,
\quad \text{and} \quad \left[\sum_{\boldsymbol s_j\in\mathcal{S}\cup\boldsymbol{s}_C}g(\boldsymbol s_j)\right]^T\cdot \left(\boldsymbol{s}_A-\boldsymbol{s}_B\right)<0, \\
&\left[\boldsymbol h_2(\boldsymbol s_A)-\boldsymbol  h_2(\boldsymbol s_B)\right]^T\cdot \sum_{\boldsymbol{s}_j\in \mathcal{S}}\boldsymbol{s}_j>0, \quad \text{and} \quad  \left[\boldsymbol h_2(\boldsymbol s_A)-\boldsymbol  h_2(\boldsymbol s_B)\right]^T\cdot \sum_{\boldsymbol{s}_j\in \mathcal{S}\cup\boldsymbol{s}_C}\boldsymbol{s}_j<0,\\
&\boldsymbol\alpha^T\cdot \left[pos(\boldsymbol{s}_A|\mathcal{S})-pos(\boldsymbol{s}_B|\mathcal{S})\right]>0, \quad \text{and} \quad \boldsymbol\alpha^T\cdot \left[pos(\boldsymbol{s}_A|\mathcal{S}\cup\boldsymbol{s}_C)-pos(\boldsymbol{s}_B|\mathcal{S}\cup\boldsymbol{s}_C)\right]<0.  \\
\end{aligned}\right.
\end{align}

(2) In the similar way, we discuss the solutions of ${f(\boldsymbol{s}_{B},\boldsymbol{s}_{C})}\big / {f(\boldsymbol{s}_{A},\boldsymbol{s}_{C})}<
{\lambda(\boldsymbol{s}_A|\mathcal{S})}\big /{\lambda(\boldsymbol{s}_B|\mathcal{S})}<1$, and it holds when
\begin{align}
\left\{\begin{aligned}
& \left[\sum_{\boldsymbol s_j\in\mathcal{S}}g(\boldsymbol s_j)\right]^T\cdot \left(\boldsymbol{s}_A-\boldsymbol{s}_B\right)<0,
\quad \text{and} \quad \left[\sum_{\boldsymbol s_j\in\mathcal{S}\cup\boldsymbol{s}_C}g(\boldsymbol s_j)\right]^T\cdot \left(\boldsymbol{s}_A-\boldsymbol{s}_B\right)>0, \\
&\left[\boldsymbol h_2(\boldsymbol s_A)-\boldsymbol  h_2(\boldsymbol s_B)\right]^T\cdot \sum_{\boldsymbol{s}_j\in \mathcal{S}}\boldsymbol{s}_j<0, \quad \text{and} \quad  \left[\boldsymbol h_2(\boldsymbol s_A)-\boldsymbol  h_2(\boldsymbol s_B)\right]^T\cdot \sum_{\boldsymbol{s}_j\in \mathcal{S}\cup\boldsymbol{s}_C}\boldsymbol{s}_j>0,\\
&\boldsymbol\alpha^T\cdot \left[pos(\boldsymbol{s}_A|\mathcal{S})-pos(\boldsymbol{s}_B|\mathcal{S})\right]<0, \quad \text{and} \quad \boldsymbol\alpha^T\cdot \left[pos(\boldsymbol{s}_A|\mathcal{S}\cup\boldsymbol{s}_C)-pos(\boldsymbol{s}_B|\mathcal{S}\cup\boldsymbol{s}_C)\right]>0.  \\
\end{aligned}\right.
\end{align}

Experiments show that the proposed learning algorithm in Section \ref{sec:learning} can always find $\Gamma_1$, $\Gamma_2$,$\Phi_1$,$\Phi_2$ and $\boldsymbol\alpha$ satisfying the above conditions, thus can effectively predict preference reversals.
\end{proof}

In summary, the proposed additive method has good scalability when new seller joins the market, and can effectively predict preference reversals with proper selection of the parameters.

\section{Real User Test}
\label{sec:test}

We design three tasks to validate the performance of \emph{Pacos} from three aspects: the personalized ranking task, the preference reversal prediction task, and the market share prediction task. Specifically, the personalized ranking task is used to verify the ranking performance of \emph{Pacos} in personal recommendation. The preference reversal prediction task is used to verify whether the proposed method can accurately predict preference reversals. Furthermore, we also conduct an interpretability study to illustrate that the \emph{Pacos} can be used to understand the cause of preference reversals. Then, the market share prediction task is used to verify whether the proposed method can accurately predict the market share of sellers when their attributes change. At last, we evaluate the impact of the proposed learning algorithm on the performance. In the following, these experiments mentioned above are presented in detail.

\subsection{The Personalized Ranking Task}
\textbf{Dataset Description.} \quad The dataset used in the personalized ranking task comes from our prior work in \cite{li2021prima++}. In this dataset, five products with different price ranges were considered:  Xiaomi scale, Bose QC35 headphone,  Panasonic  EH-NA  series  hairdryer,  Xiaomi smartphone  with  6G  RAM  and  128G  storage,  and  Austin Air Purifier  HM  400. For each product, we collected information of real sellers from eBay, including price and seller reputation. Then, the collected sellers were randomly grouped into 10 markets, each with $4\sim6$ sellers ($N=6$). We invited $682$ subjects for an interview, including 366 males and 316 females. Each subject was asked to consider price and reputation information only, and select one seller from each market as their best choice. More details about the data collection can be found in \cite{li2021prima++}.

\textbf{Experiment Setup.} \quad The data pre-processing steps are as follows. First, following our prior work \cite{li2021prima++}, we normalize the attributes of sellers. After normalization, each attribute is in the range $[0,1]$, and a larger normalized value indicates a higher preference. In addition, following the work in \cite{mottini2017deep}, a special ``dummy'' item is introduced to pad these markets that contain items less than $N$. For example, if there are only 4 items in a market, then we pad the market with 2 ``dummy'' items. In this way, the number of items is the same and equals to $N$ in all markets, which facilitates the training procedure. In particular, for the ``dummy'' item, we set all of its attributes to zero, so that it always has zero utility and is ranked the least preferred. Therefore, it does not affect the ranking of other items in the market.

For the proposed \emph{Pacos}, to distinguish two module design methods, we use ``\emph{Pacos}-add'' to denote the additive method, and use ``\emph{Pacos}-NN'' to denote the ANN-based method. We compare \emph{Pacos}-add and \emph{Pacos}-NN with the following prior works.
\begin{itemize}
  \item The MNL-based methods, including MNL \cite{mcfadden1973conditional}, PRIMA++ \cite{li2021prima++}, LCL \cite{tomlinson2020learning}, and CDM \cite{rosenfeld2020predicting}.
  \item Machine learning-based methods, such as Naive Bayes classifier \cite{rish2001empirical}, PNN-based method \cite{mottini2017deep}, RankNet \cite{burges2005learning}, and RankSVM \cite{fine2001efficient}.
  \item The random baseline, where all sellers in a market have the same probability of being selected.
\end{itemize}

In our experiments, all methods are implemented using  the  Scikit-learn  0.22.1 and Tensorflow 1.12.0 package. To avoid the impact of randomness, we use the 5-fold cross-validation method to divide the training dataset and the test dataset, and repeat the experiment 10 times with different initial conditions. As each product contains 10 markets, 8 markets are used for training and 2 for testing in our 5-fold cross-validation method.

\textbf{Performance Metrics.} \quad Given a market $\mathcal{S}$ with $N$ sellers, we sort them in the descending order of their predicted probabilities $\{\mathcal{P}(\boldsymbol{s}_i|\mathcal{S})\}$, and the seller with the largest predicted probability is ranked the first. Let $\boldsymbol s_b$ be the selected seller and $v_b$ be its ranking position. To evaluate the ranking performance of \emph{Pacos}, following the work in \cite{li2021prima++}, the performance metrics we use are as follows.
\begin{itemize}
  \item ranking quality ($rq$): it is defined as $rq=(N-v_b)\big / (N-1)$ with $0\leq rq \leq 1$. When $s_b$ is ranked as the top one with $v_b=1$, we have $rq=1$. While when $s_b$ is ranked the last with $v_b=N$, $rq=0$.
  \item success rate ($sr$): given a positive integer $m$, the success rate $sr(m)$ refers to the frequency that a model ranks the selected seller $s_b$ in the top $m$ positions, i.e., $1\leq v_b \leq m$.
\end{itemize}
A higher value of ranking quality or success rate represents a better ranking performance.

\textbf{Experimental Results.} \quad The results of ranking quality are shown in Table \ref{tab:rq}. First, it shows that \emph{Pacos}-add and \emph{Pacos}-NN have similar performance on all products. Furthermore, the random baseline performs the worse among all methods, while \emph{Pacos}-add and \emph{Pacos}-NN achieve the highest overall accuracy. Last, PRIMA++ has the best performance among existing MNL-based methods, while the PNN-based method performs the best among existing machine learning-based methods. The results of $sr(m=1)$ are shown in the supplementary file, and we observe a similar trend. This shows that the proposed preference model can achieve higher accuracy than existing methods on the personalized ranking task, and \emph{Pacos}-add and \emph{Pacos}-NN have similar ranking performance.

\begin{table}[tbp]
  \centering
  \caption{Results of Ranking Quality in Personalized Ranking Task.}
    \begin{tabular}{ccccccc}
    \toprule
          & Air Purifier & Head phone & Hairdryer & Smartphone & Scale & Average \\
    \midrule
    Random & 0.500  & 0.500  & 0.501  & 0.500  & 0.500  & 0.500  \\
    MNL   & 0.713  & 0.780  & 0.763  & 0.799  & 0.684  & 0.748  \\
    LCL   & 0.755  & 0.832  & 0.776  & 0.810  & 0.700  & 0.774  \\
    CDM   & 0.746  & 0.632  & 0.630  & 0.569  & 0.641  & 0.644  \\
    PRIMA++ & 0.849  & 0.809  & 0.796  & 0.805  & 0.742  & 0.800  \\
    Naive Bayes & 0.798  & 0.829  & 0.778  & 0.768  & 0.774  & 0.789  \\
    RankNet & 0.849  & \textbf{0.849}  & 0.827  & 0.830  & 0.738  & 0.819  \\
    RankSVM & 0.722  & 0.787  & 0.772  & 0.804  & 0.687  & 0.754  \\
    PNN & 0.852  & 0.828  & 0.825  & 0.832  & 0.801  & 0.827  \\
    \midrule
    \emph{Pacos}-add & \textbf{0.866 } & 0.844  & 0.838  & 0.839  & 0.825  & 0.842  \\
    \emph{Pacos}-NN & 0.863  & \textbf{0.846 } & \textbf{0.844 } & \textbf{0.842 } & \textbf{0.830 } & \textbf{0.845 } \\
    \bottomrule
    \end{tabular}%
  \label{tab:rq}%
\end{table}%


\begin{figure}[tbp]
  \centering
  \includegraphics[width=8cm]{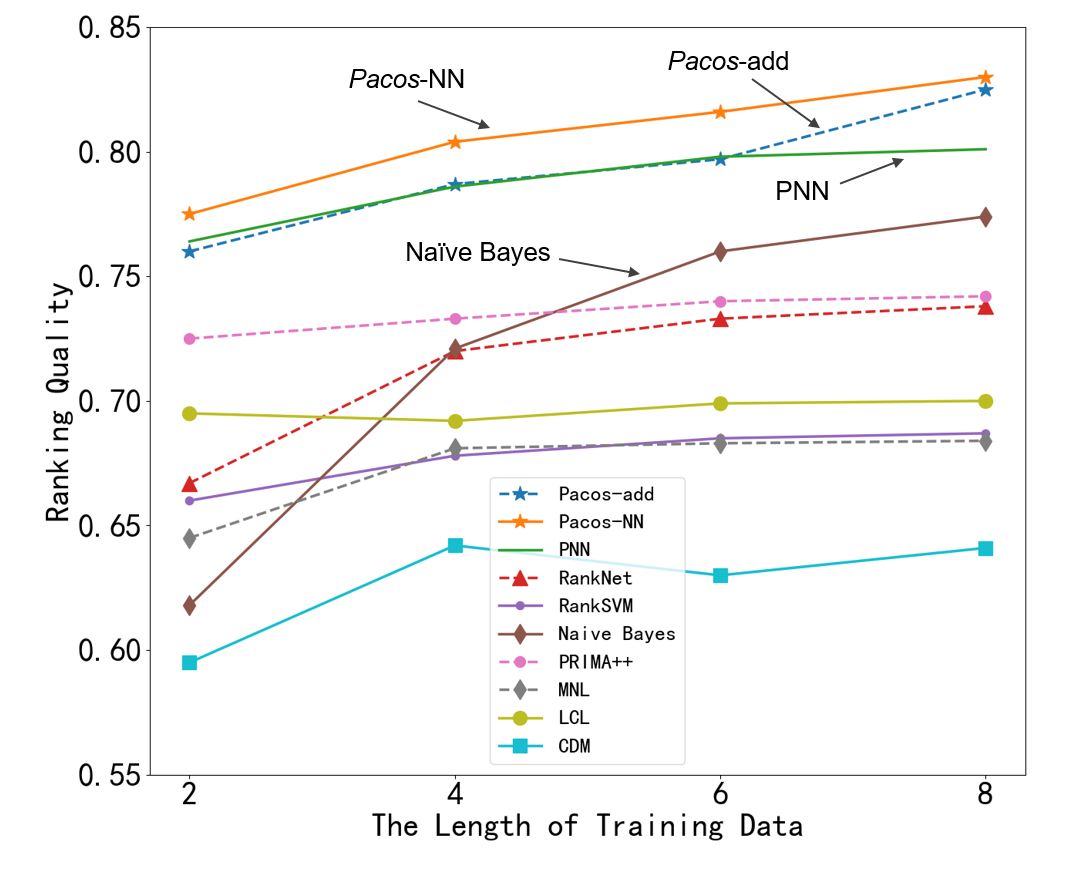}
  \caption{The impact of The Length of Training Data on Performance.}
  \label{fig:size}
\end{figure}

We also analyze the impact of the length of training data on performance. In the following, we take the Xiaomi scale as an example to show the experimental results, and we observe similar trends in other products. As shown in Fig. \ref{fig:size}, \emph{Pacos}-NN always has the best performance. When comparing \emph{Pacos}-add with \emph{Pacos}-NN, the accuracy of \emph{Pacos}-add is lower than that of \emph{Pacos}-NN when the length of training data is 2 to 6, and the accuracy of \emph{Pacos}-add is similar to that of \emph{Pacos}-NN when the length is 8. Additionally, PNN-based method gives the best performance among all prior works. When comparing \emph{Pacos}-add with PNN-based method, \emph{Pacos}-add has similar performance to the PNN-based method when the training data is less than 6, and \emph{Pacos}-add performs slightly better than the PNN-based method when the training data is 8. In summary, the proposed \emph{Pacos}-add and \emph{Pacos}-NN do not depend on a large amount of training data, and can still obtain high accuracy when only small data is available.

\subsection{The Preference Reversal Prediction Task}
\textbf{Dataset Description.} \quad Due to the lack of dataset on preference reversals in prior works, we designe a preference reversal prediction task to verify whether the proposed method is able to predict preference reversals. The product used in this experiment is the Xiaomi scale, which has two attributes: price (\textyen) and seller reputation (REP). A total of 647 subjects are invited, including 332 males and 315 females, and subjects aging from 18 to 25, 26 to 30, and 31 to 40 account for $35.4\%$, $26.9\%$, $27.8\%$ of the total subjects, respectively. The occupations of these participants include students, clerks, administration staff, technical staff, etc.

\begin{figure}[tbp]
\centering
\begin{minipage}[b]{.24\linewidth}
  \centering
  \centerline{\includegraphics[width=3.8cm]{market7-1-aligned.png}}
  \centerline{(a) market A-1}\medskip
\end{minipage}
\hfill
\begin{minipage}[b]{.24\linewidth}
  \centering
  \centerline{\includegraphics[width=3.8cm]{market7-2-aligned.png}}
  \centerline{(b) market A-2 }\medskip
\end{minipage}
\begin{minipage}[b]{.24\linewidth}
  \centering
  \centerline{\includegraphics[width=3.8cm]{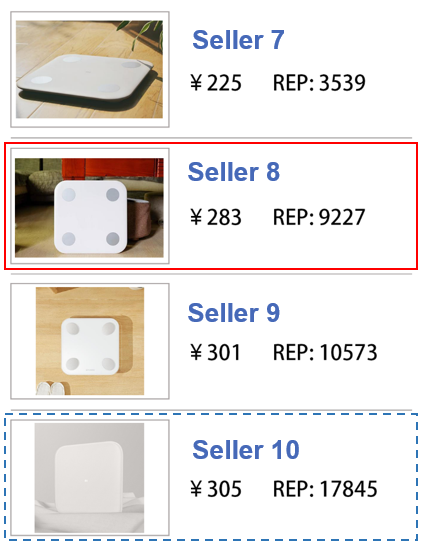}}
  \centerline{(c) market B-1}\medskip
\end{minipage}
\hfill
\begin{minipage}[b]{.24\linewidth}
  \centering
  \centerline{\includegraphics[width=3.8cm]{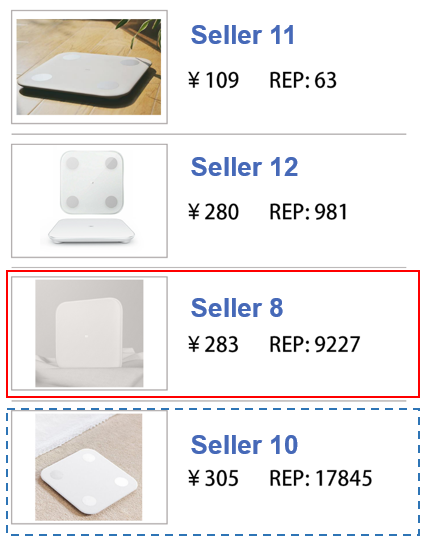}}
  \centerline{(d) market B-2 }\medskip
\end{minipage}
\caption{Two groups of markets used in the preference reversal prediction task.}
\label{fig:two_markets}
\end{figure}

In our experiment, we use the salience effect in economics \cite{bordalo2013salience} to induce preference reversals. The salience effect stems from the comparison among items, causing people to be more likely to focus on items that are more prominent in the market. Let $q(\boldsymbol{s}_i)$ be the ratio of the reputation to its price of seller $\boldsymbol{s}_i$, and define $\Delta q(\boldsymbol{s}_i)=q(\boldsymbol{s}_i)-q(\boldsymbol{s}_{i-1})$ as the $q$ value increment of seller $\boldsymbol{s}_i$ relative to seller $\boldsymbol{s}_{i-1}$, where $\Delta q(\boldsymbol{s}_1)=q(\boldsymbol{s}_1)$. The salience effect states that users have more preference for sellers with a higher $\Delta q(\boldsymbol{s}_i)$. We design two groups of markets based on the salience effect, and the detailed information is shown in Fig. \ref{fig:two_markets}. In market A-1 in the group A, the $q$ values of these sellers are $q(\boldsymbol{s}_1)=0.18,q(\boldsymbol{s}_2)=2.83,
q(\boldsymbol{s}_3)=22.16,q(\boldsymbol{s}_4)=39.91$, and the corresponding values of $\Delta q(\boldsymbol{s}_i)$ are $\Delta q(\boldsymbol{s}_1)=0.18,\Delta q(\boldsymbol{s}_2)=2.65,
\Delta q(\boldsymbol{s}_3)=19.33,\Delta q(\boldsymbol{s}_4)=17.75$, respectively. Since $\Delta q(\boldsymbol{s}_3)$ is the largest, the salience effect indicates that more users will choose seller 3 in market A-1. While in market A-2 in the same group, the $q$ values of four sellers are $q(\boldsymbol{s} _5)=14.20,q(\boldsymbol{s} _3)=22.16,q(\boldsymbol{s} _6)=25.55,q(\boldsymbol{s} _4)=39.91$, and the corresponding values of $\Delta q(\boldsymbol{s}_i)$ are $\Delta q(\boldsymbol{s}_5)=14.20,\Delta q(\boldsymbol{s}_3)=7.96,\Delta q(\boldsymbol{s}_6)=3.34,\Delta q(\boldsymbol{s}_4)=14.36$. Similarly, since $\Delta q(\boldsymbol{s}_4)$ is the largest, the salience effect indicates that more users will choose seller 4 in market A-2. We invited 647 subjects to select their best choices from each of the two markets above. Their selection data is shown in Fig. \ref{fig:context_effect_m7} (a), and we do observe the preference reversal phenomenon in this test. Similarly, we design two markets in group B using the same approach, and also observe the preference reversal phenomenon in users' decisions.

Apart from the above 4 markets, we design another 7 markets, and observe no preference reversal in these 7 markets in users' real selection data. In the preference reversal prediction task, we use these 7 markets as the training data and the 4 markets designed with the salience effect as the testing data. In this task, we compare the proposed \emph{Pacos} with the following prior methods.
\begin{itemize}
  \item MNL, which is one of the most widely used methods in choice problems.
  \item PRIMA++, which is used as an example to show the results of MNL-based methods since it has the best performance among them.
  \item PNN-based method, which is used as an example to show the results of machine learning-based methods as it performs best among them.
\end{itemize}

\begin{figure}[tbp]
\centering
\begin{minipage}[b]{.33\linewidth}
  \centering
  \centerline{\includegraphics[width=4.2cm]{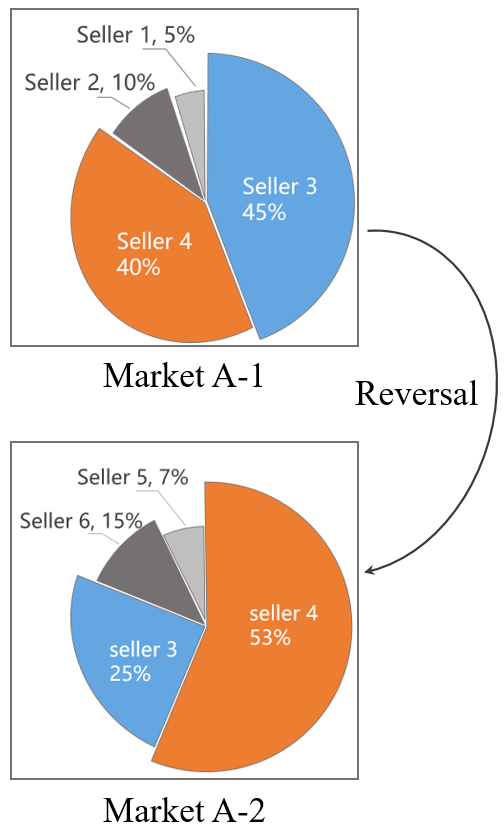}}
  \centerline{(a) real market share}\medskip
\end{minipage}
\hfill
\begin{minipage}[b]{.32\linewidth}
  \centering
  \centerline{\includegraphics[width=4.2cm]{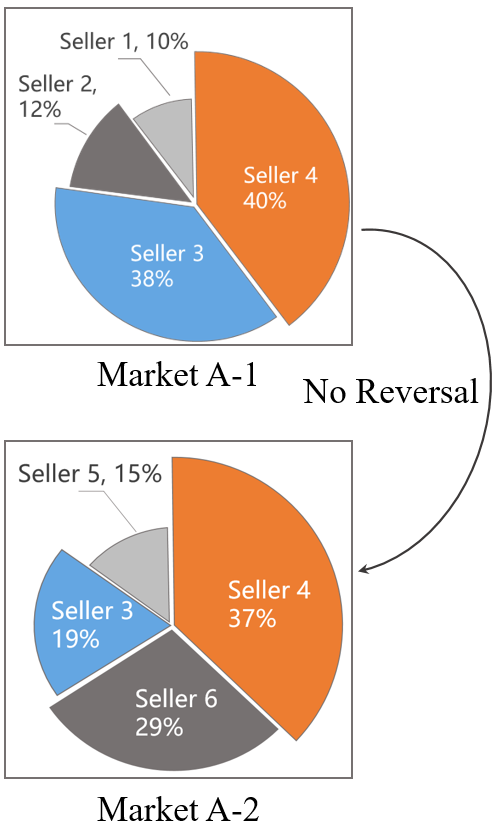}}
  \centerline{(b) \emph{Pacos}-add }\medskip
\end{minipage}
\hfill
\begin{minipage}[b]{0.33\linewidth}
  \centering
  \centerline{\includegraphics[width=4.2cm]{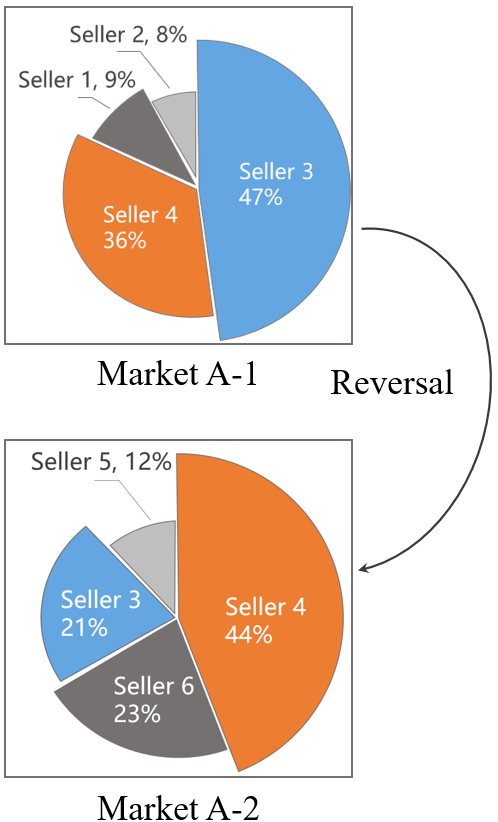}}
  \centerline{(c) \emph{Pacos}-NN }\medskip
\end{minipage}
\hfill
\begin{minipage}[b]{.33\linewidth}
  \centering
  \centerline{\includegraphics[width=4.2cm]{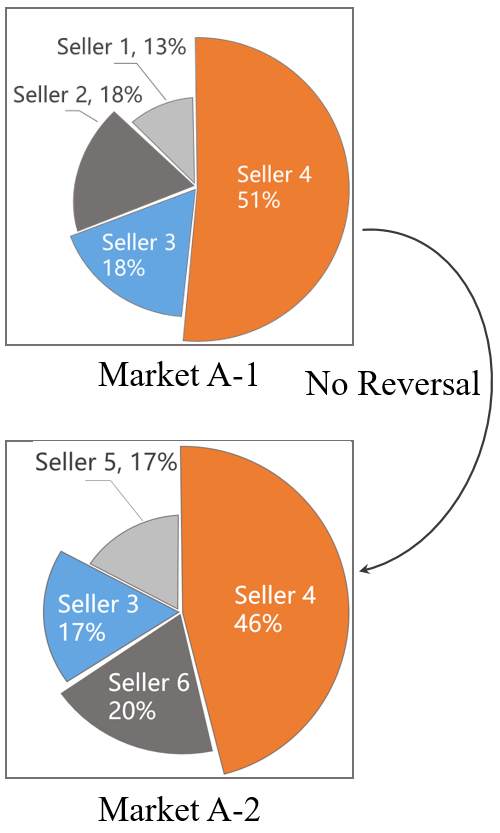}}
  \centerline{(d) MNL}\medskip
\end{minipage}
\hfill
\begin{minipage}[b]{.32\linewidth}
  \centering
  \centerline{\includegraphics[width=4.2cm]{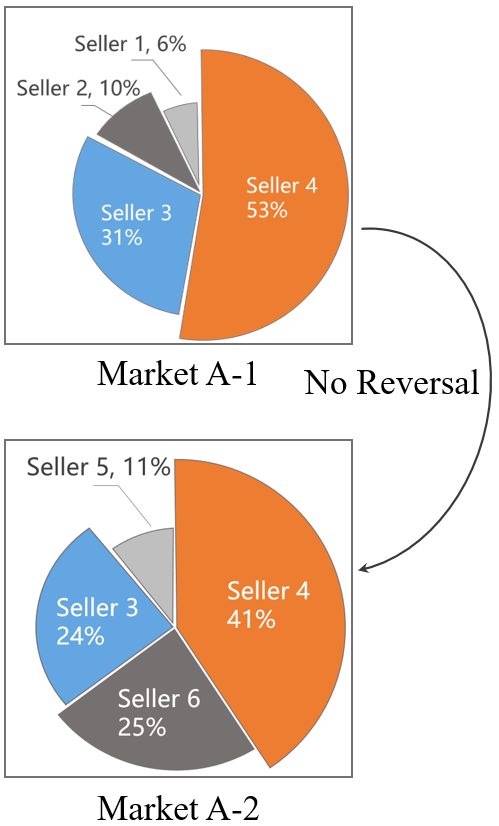}}
  \centerline{(e) PRIMA++}\medskip
\end{minipage}
\hfill
\begin{minipage}[b]{0.33\linewidth}
  \centering
  \centerline{\includegraphics[width=4.2cm]{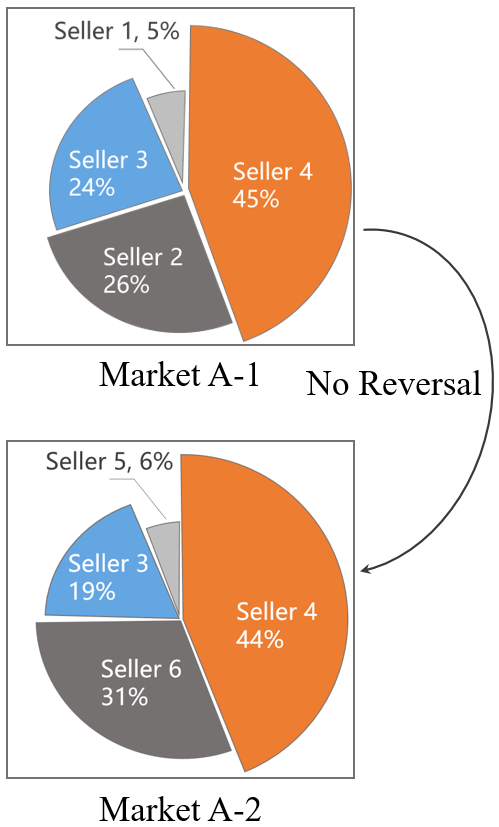}}
  \centerline{(f) PNN}\medskip
\end{minipage}
\caption{The predicted results in market A-1 and market A-2.}
\label{fig:context_effect_m7}
\end{figure}

\textbf{Performance Metrics.} \quad In this task, these methods are evaluated from two aspects. First, we consider the preference order of sellers, and evaluate whether existing methods can predict the preference reversal phenomenon. Furthermore, we evaluate these methods in terms of the accuracy of the estimated $\{\mathcal{P}(\boldsymbol s_i|\mathcal{S})\}$, and expect that the predicted $\{\mathcal{P}(\boldsymbol s_i|\mathcal{S})\}$ should be as close as possible to their true values. Considering that the ground truth of these probabilities cannot be obtained, following our prior work \cite{li2021prima++}, we introduce the market share and evaluate the accuracy of $\{\mathcal{P}(\boldsymbol s_i|\mathcal{S})\}$ by comparing the estimated probabilities with the market share. Formally, assume that a market $\mathcal{S}$ contains $N$ sellers and $M$ users are asked to select their best choices. Let $\mathcal{P}^j(\boldsymbol s_i|\mathcal{S})$ be the predicted probability that the user $j$ selects seller $\boldsymbol s_i$, and define $q_{i}=\frac{1}{M}\sum_j\mathcal{P}^j(\boldsymbol s_i)$. Let $q_{i}^*$ be the percentage of users who selects seller $\boldsymbol s_i$. Then $q_{i}^*$ can be considered as the true market share of seller $\boldsymbol s_i$, and $q_i$ can be regarded as the predicted market share of $\boldsymbol s_i$. Then we use
\begin{align}
KLD = \sum_i q_i\log_2 \frac{q_i}{q_i^*}, \quad \text{and} \quad MAE = \frac{1}{N}\sum_i\left|q_i-q_i^*\right|,
\end{align}
to evaluate the accuracy of the estimated probabilities $\{\mathcal{P}(\boldsymbol s_i|\mathcal{S})\}$, where KLD is the Kullback-Leibler divergence between $\{q_i^*\}$ and $\{q_i\}$, and MAE  is the mean absolute error between $\{q_i^*\}$ and $\{q_i\}$. Smaller values of KLD and MAE means higher accuracy of these estimated probabilities.

\textbf{Experimental Results.} \quad The predicted results of group A and B are shown in Fig. \ref{fig:context_effect_m7} and Fig. \ref{fig:context_effect_m8}, respectively. For markets in group A, more users choose seller 3 in market A-1, while more users select seller 4 in market A-2. Among all methods considered in this experiment, only the proposed \emph{Pacos}-NN can predict the preference reversal between market A-1 and market A-2. Then for markets in group B, more users choose seller 4 in market B-1, while more users select seller 3 in market B-2. The proposed \emph{Pacos}-add, \emph{Pacos}-NN and PNN-based methods can predict the preference reversal in market B-1 and market B-2. Furthermore, we examine the learned parameters $\Gamma_1$, $\Gamma_2$, $\Phi_1$, $\Phi_2$ and $\boldsymbol\alpha$, and find that they satisfy all the conditions in Theorem \ref{theo:effectiveness}.  The above results indicate that our proposed method can effectively predict the preference reversal phenomenon. More importantly, the above results are learned in the absence of preference reversals in the training set, reflecting the superiority of \emph{Pacos}. 

To evaluate the accuracy of $\{\mathcal{P}(\boldsymbol s_i|\mathcal{S})\}$, the results of KLD and MAE are in Table \ref{tab:metric_of_context_effect}. It shows that our proposed \emph{Pacos}-NN performs the best on these two metrics, PRIMA++ is the second, while MNL achieves the lowest accuracy. When comparing \emph{Pacos}-add with PNN-based method, even both methods successfully predict the preference reversal phenomenon in group B, \emph{Pacos}-NN achieves higher accuracy in estimation of the market share. In summary, experimental results demonstrate that the predicted probabilities of \emph{Pacos}-NN are more accurate than prior works.

\begin{figure}[tbp]
\centering
\begin{minipage}[b]{.33\linewidth}
  \centering
  \centerline{\includegraphics[width=4.2cm]{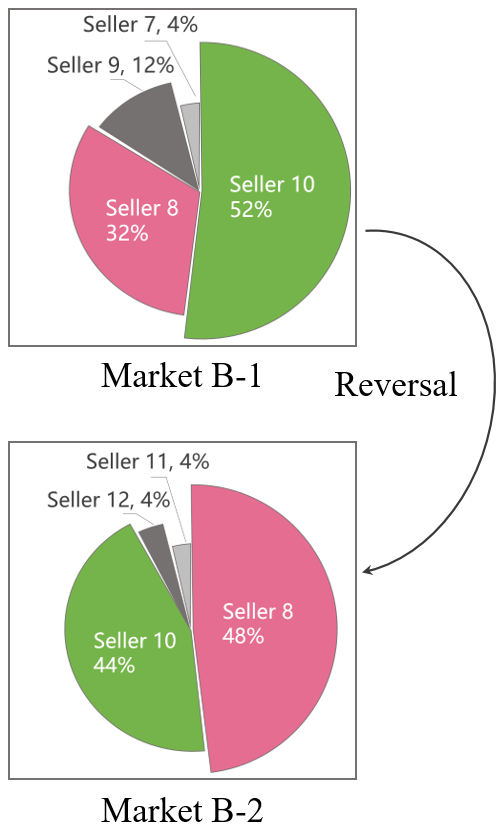}}
  \centerline{(a) real market share}\medskip
\end{minipage}
\hfill
\begin{minipage}[b]{.32\linewidth}
  \centering
  \centerline{\includegraphics[width=4.2cm]{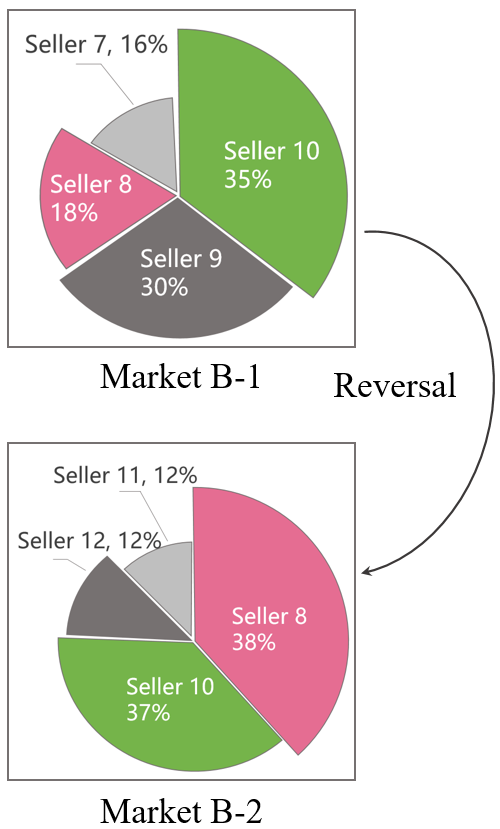}}
  \centerline{(b) \emph{Pacos}-add }\medskip
\end{minipage}
\hfill
\begin{minipage}[b]{0.33\linewidth}
  \centering
  \centerline{\includegraphics[width=4.2cm]{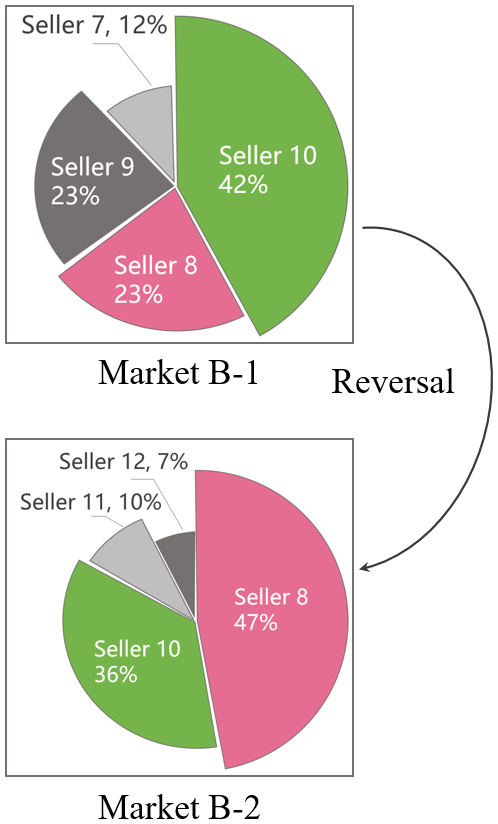}}
  \centerline{(c) \emph{Pacos}-NN }\medskip
\end{minipage}
\hfill
\hfill
\begin{minipage}[b]{.33\linewidth}
  \centering
  \centerline{\includegraphics[width=4.2cm]{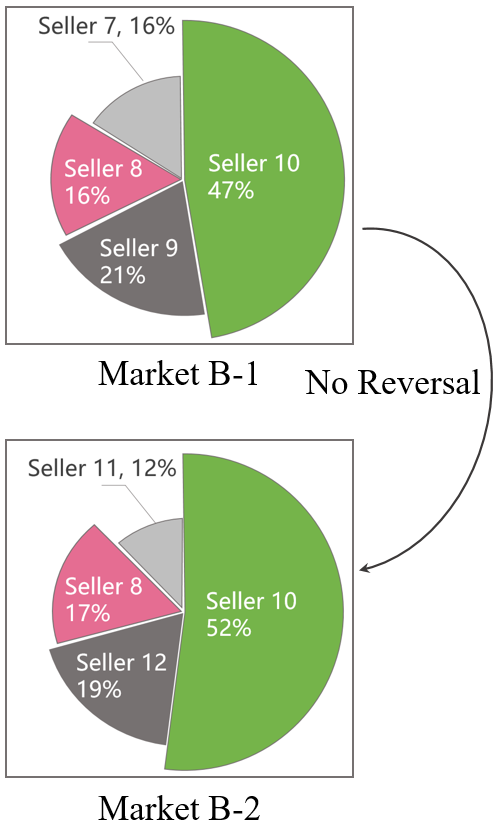}}
  \centerline{(d) MNL}\medskip
\end{minipage}
\hfill
\begin{minipage}[b]{.32\linewidth}
  \centering
  \centerline{\includegraphics[width=4.2cm]{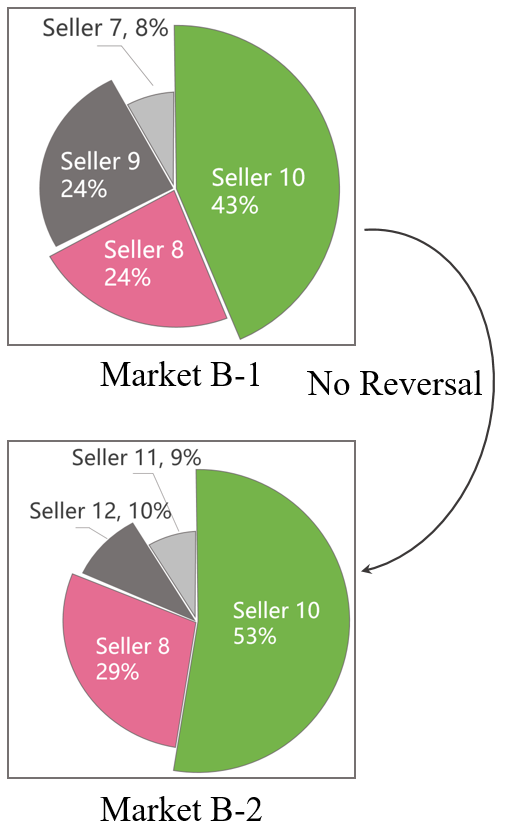}}
  \centerline{(e) PRIMA++}\medskip
\end{minipage}
\hfill
\begin{minipage}[b]{0.33\linewidth}
  \centering
  \centerline{\includegraphics[width=4.2cm]{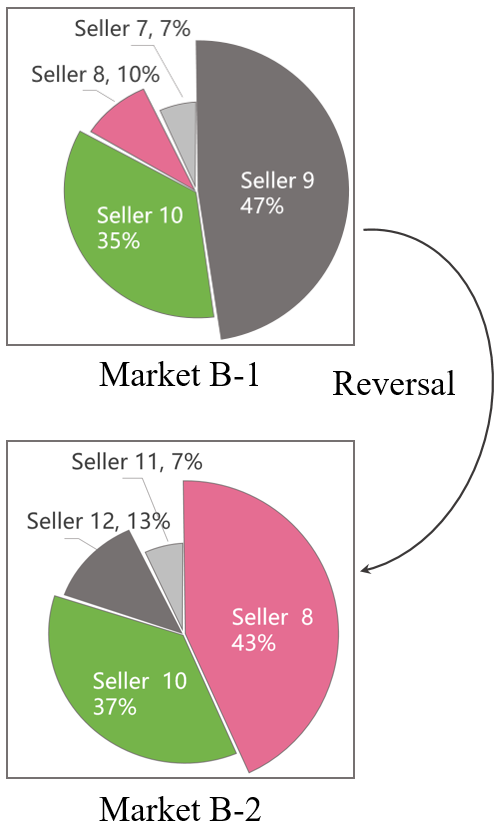}}
  \centerline{(f) PNN}\medskip
\end{minipage}
\caption{The predicted results in market B-1 and market B-2.}
\label{fig:context_effect_m8}
\end{figure}


\begin{table}[tbp]
  \centering
  \caption{The Results of KLD and MAE.}
    \begin{tabular}{cccccc}
    \toprule
          & \emph{Pacos}-add & \emph{Pacos}-NN  & PNN & MNL & PRIMA++ \\
    \midrule
    MAE   & 0.096 & \textbf{0.061} & 0.108 & 0.117 & 0.079 \\
    KLD   & 0.158 & \textbf{0.071} & 0.233 & 0.267 & 0.094 \\
    \bottomrule
    \end{tabular}%
  \label{tab:metric_of_context_effect}%
\end{table}%

\subsection{The Interpretability Study}
We conduct this study to show that the proposed \emph{Pacos} has good interpretability. In this study, by looking at the relative magnitudes of the three parts of utilities, we try to identify which factor users value more when making choices and understand their decision making process. The interpretability study is important, since existing methods either cannot predict preference reversals, or cannot offer this interpretability due to its black-box nature in design. This study can help understand the mechanisms behind preference reversals and facilitate sellers to design pricing or marketing strategies.

In this study, the proposed \emph{Pacos}-add is used as an example, as it is designed to have better interpretability than \emph{Pacos}-NN. We focus on the preference reversal prediction task, and randomly select a user who shows the preference reversal in this task. The predicted utilities in markets A-1 and A-2 obtained by \emph{Pacos}-add are visualized in Fig. \ref{fig:interpretability_study}. The results for markets B-1 and B-2 and that for other users are similar and show the same trend, and thus omitted here. In Fig. \ref{fig:interpretability_study}, the four rows represent four available sellers in a market. The first three columns denote the three parts of utilities, and the last column labeled with ``U'' denotes the total utilities. The detailed comparison utilities between pairwise sellers $h(\boldsymbol{s}_i,\boldsymbol{s}_j)$ are indicated in the dashed box. For example, in market A-1, the comparison utility (CU) of seller 3 is $4.29$, and the comparison utility of seller 3 caused by seller 4, i.e., $h(\boldsymbol{s}_3,\boldsymbol{s}_4)$, is $1.00$.

First of all, the proposed \emph{Pacos}-add can accurately predict the preference reversal, as item 3 has the largest total utility in market A-1, and item 4 has the largest total utility in market A-2. Recall that the two markets are designed using the salience effect, which stems mainly from the inter-item comparison. If we look at the estimated results of the three parts of utilities, the comparison utility dominates the sum utilities, while the attribute utility and the position utility have much smaller values. In addition, note that with our model, the estimated attribute utilities of seller 3 are almost the same in both markets (the difference is smaller than $10^{-3}$ caused by users' adaptive weights), same for its position utility in both markets. A similar phenomenon is observed for the attribute utility and position utility of seller 4. This shows that our model predicts the preference reversal in group A because it accurately learns the change in users' preferences due to the inter-item comparison in these two markets. The analysis indicates that the predicted utilities are in line with our design intentions.

If we take a closer look at the detailed comparison utility in Fig. \ref{fig:interpretability_study}, the comparison utilities between seller 3 and seller 4 are the same in both markets with $h(\boldsymbol{s}_3,\boldsymbol{s}_4)=1.00$ and $h(\boldsymbol{s}_4,\boldsymbol{s}_3)=1.12$. If we look at seller 3, in market A-1, the comparison utilities of seller 3 from seller 1 and seller 2, i.e., $h(\boldsymbol{s}_3,\boldsymbol{s}_1)$ and $h(\boldsymbol{s}_3,\boldsymbol{s}_2)$, are both positive and they are 1.25 and 1.02, respectively, and in market A-2,  the comparison utilities of seller 3 from seller 5 and seller 6, i.e., $h(\boldsymbol{s}_3,\boldsymbol{s}_5)$ and $h(\boldsymbol{s}_3,\boldsymbol{s}_6)$, are 1.14 and 0.99, respectively. That is, seller 3 has similar comparison utilities in both markets. If we look at seller 4, in market A-1, the comparison utilities of seller 4  from seller 1 and seller 2, i.e., $h(\boldsymbol{s}_4,\boldsymbol{s}_1)$ and $h(\boldsymbol{s}_4,\boldsymbol{s}_2)$, are near zero, and in market A-2, the comparison utilities of seller 4 from 5 and seller 6, i.e., $h(\boldsymbol{s}_4,\boldsymbol{s}_5)$ and $h(\boldsymbol{s}_4,\boldsymbol{s}_6)$, are $0.60$ and $1.35$, respectively. Therefore, seller 4 gains much more comparison utility from seller 5 and seller 6 than from seller 1 and seller 2. Consequently, the gain in the comparison utility from seller 5 and seller 6 in market A-2 increases the market share of seller 4, and further creates a reversal of user's preference. In summary, this study demonstrates that \emph{Pacos}-add can accurately learns the change in users' preferences and identify the cause of preference reversals.

\begin{figure}[tbp]
\centering
\begin{minipage}[b]{.45\linewidth}
  \centering
  \centerline{\includegraphics[width=7cm]{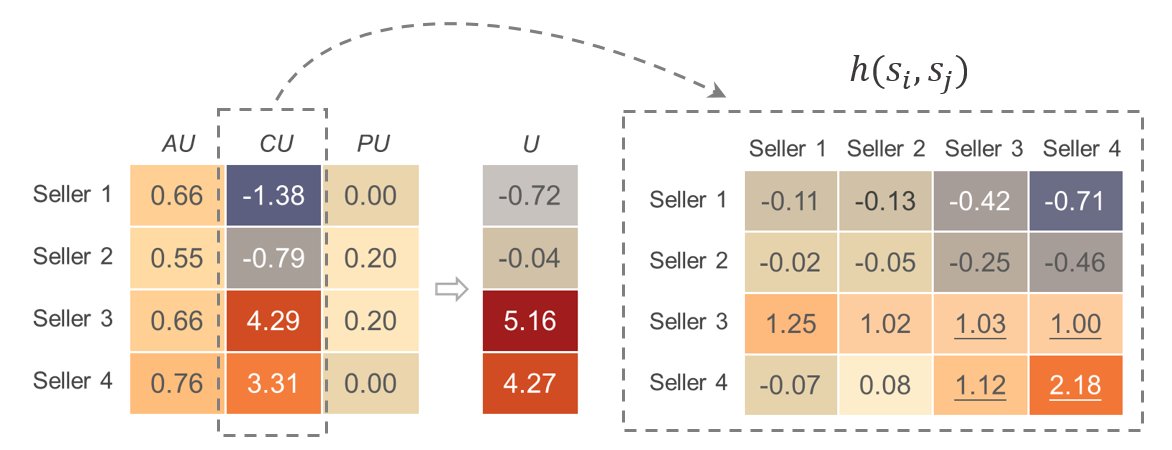}}
  \centerline{(a) Market A-1}\medskip
\end{minipage}
\hfill
\begin{minipage}[b]{.49\linewidth}
  \centering
  \centerline{\includegraphics[width=8cm]{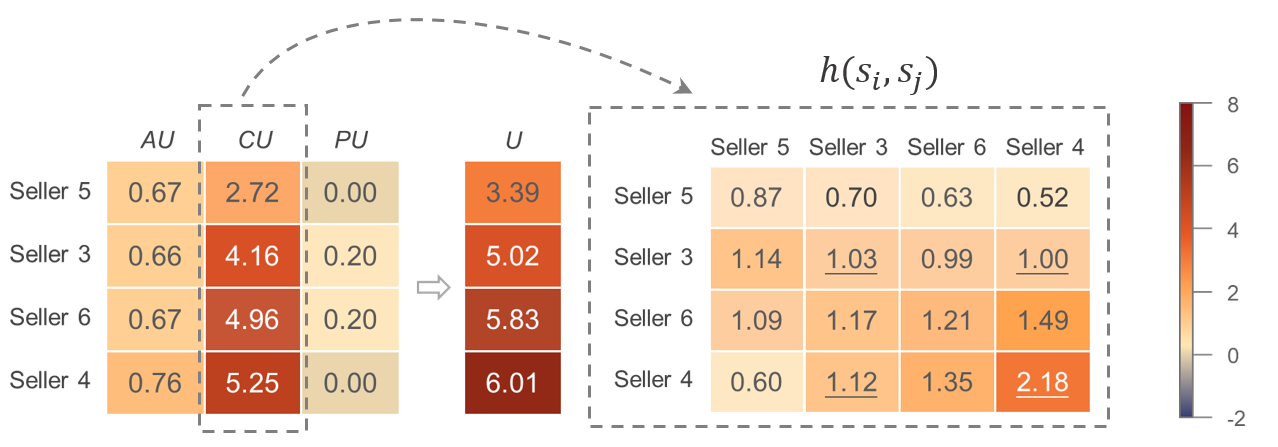}}
  \centerline{(b) Market A-2}\medskip
\end{minipage}
\caption{The illustration of the interpretability study.}
\label{fig:interpretability_study}
\end{figure}

\subsection{The Market Share Prediction Task}
\label{sec:test_attribute}

\begin{figure}[!t]
\centering
\begin{minipage}[b]{.32\linewidth}
  \centering
  \centerline{\includegraphics[width=3.8cm]{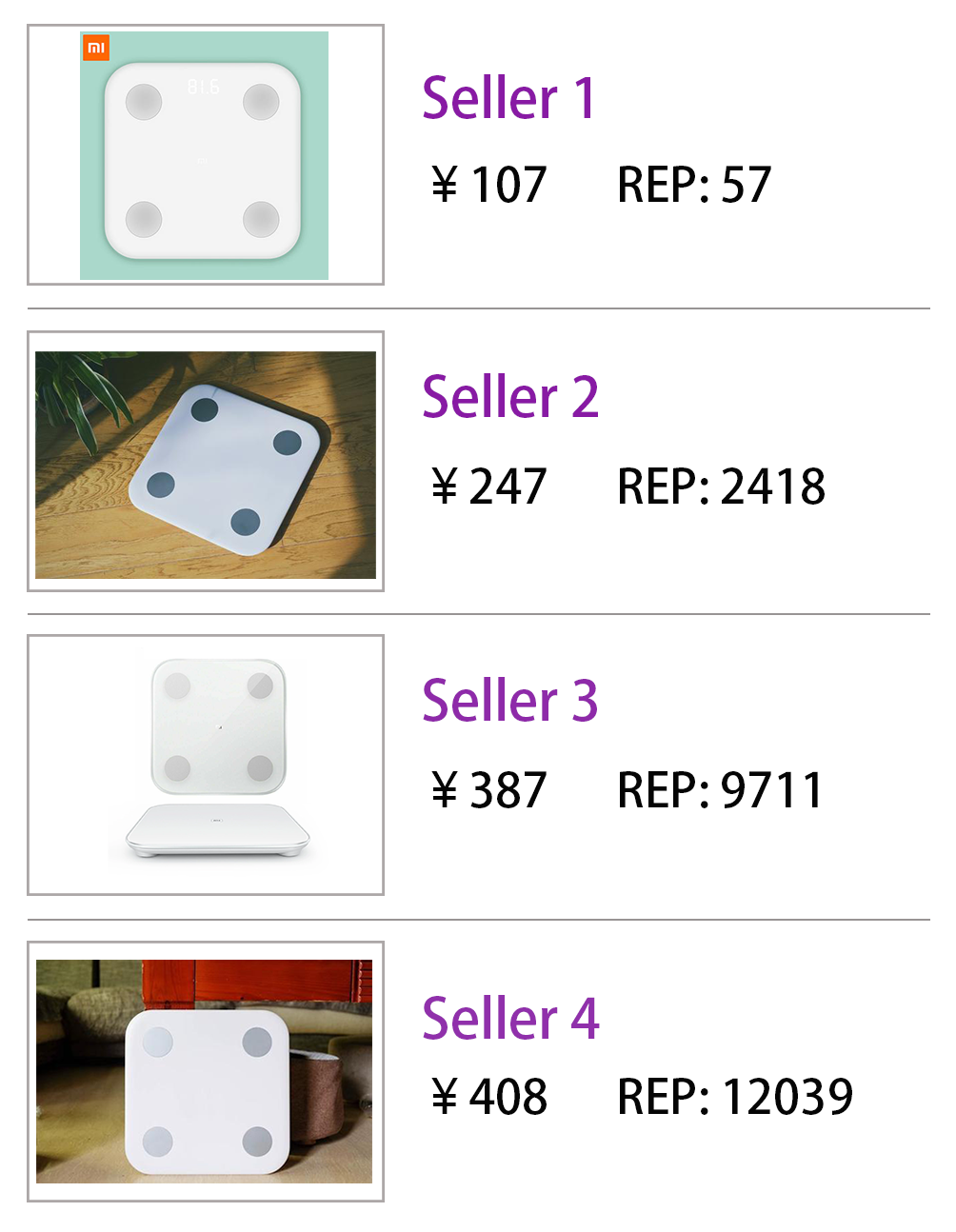}}
  \centerline{(a)}\medskip
\end{minipage}
\hfill
\begin{minipage}[b]{.67\linewidth}
  \centering
  \centerline{\includegraphics[width=11cm,height=3.8cm]{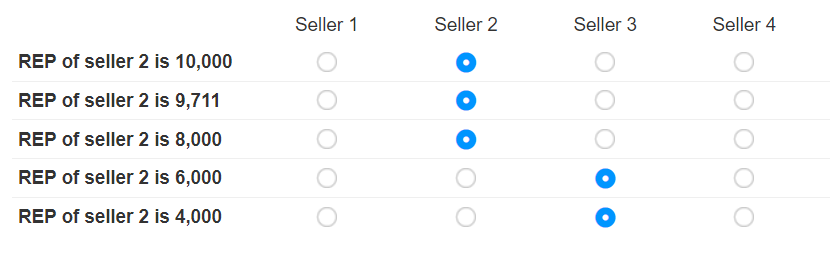}}
  \centerline{(b)}\medskip
\end{minipage}
\caption{An example in market share prediction task, where (a) the market to adjust attributes, and (b) the snapshot of user GUI in this task.}
\label{fig:market_for_attribute}
\end{figure}

\textbf{Dataset Description.} \quad In real applications, the attributes of sellers are not static, and sellers may want to increase their market shares by providing price discounts or improving reputations. Therefore, it is important to accurately predict the change in market shares caused by attribute adjustments. Due to the lack of related experiments in existing works, we design the market share prediction task to validate our model in terms of responses to attribute adjustments. In this experiment, we focus on the market shown in Fig. \ref{fig:market_for_attribute} (a), and adjust the attribute of one seller each time. For example, as shown in Fig. \ref{fig:market_for_attribute} (b), when adjusting the reputation of seller 2, we set its reputation as 10000, 9711, 8000, 6000, 4000, and observe the market share of these sellers in each adjustment. We invite the same 647 participants in the preference reversal task to complete this experiment.

\begin{figure}[tbp]
\centering
\begin{minipage}[b]{.5\linewidth}
  \centering
  \centerline{\includegraphics[width=14.8cm]{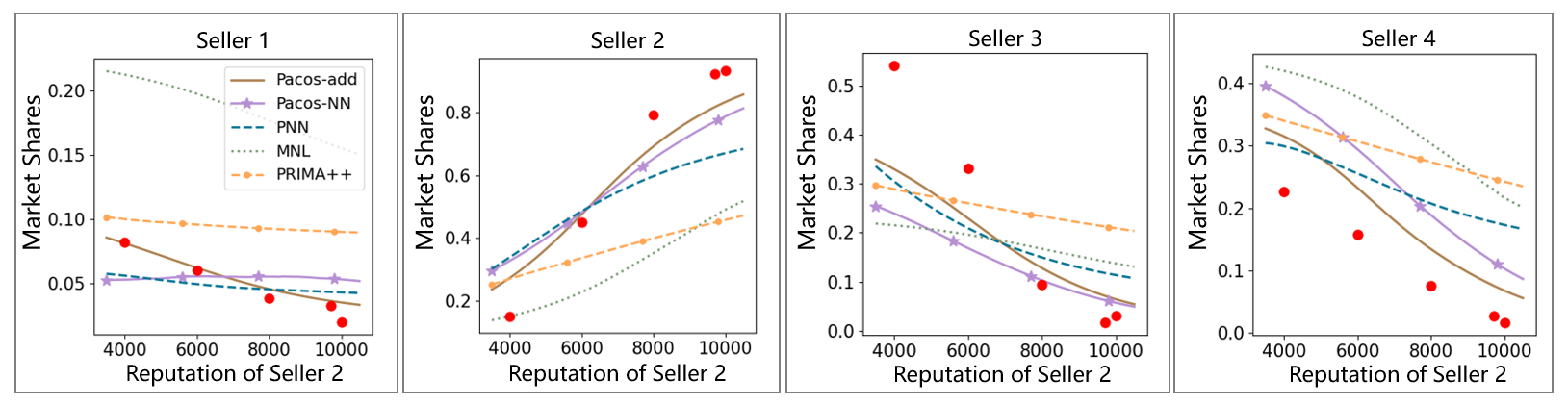}}
  \centerline{(a) adjusting the reputation of seller 2}\medskip
\end{minipage}
\hfill
\begin{minipage}[b]{.5\linewidth}
  \centering
  \centerline{\includegraphics[width=14.8cm]{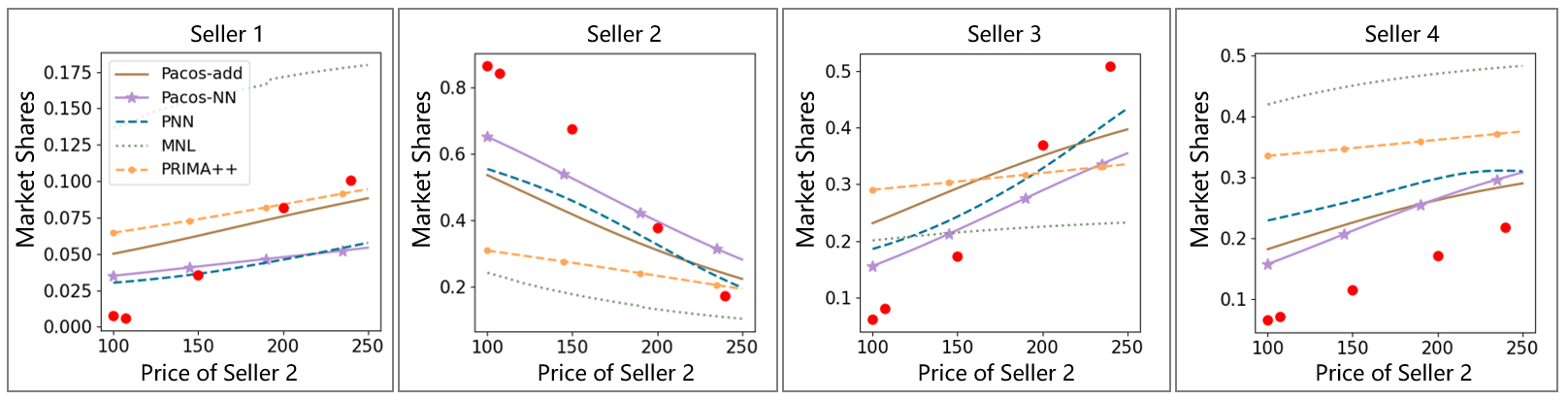}}
  \centerline{(b) adjusting the price of seller 2}\medskip
\end{minipage}
\caption{The predicted trend in market share prediction task.}
\label{fig:seller2}
\end{figure}

\textbf{Experimental Results.} \quad In this task, we compare \emph{Pacos} with MNL, PRIMA++, and PNN-based method. The results when adjusting attributes of seller 2 are shown in Fig. \ref{fig:seller2} and Table \ref{tab:attribute_change}, and results in other cases are similar and shown in the supplementary file. In Fig. \ref{fig:seller2}, the four subgraphs correspond to the estimated and real market shares four sellers. In each subgraph, the five red dots represent the true market shares of the seller, and a solid line is the simulation result of one method.

The experimental results when adjusting the reputation of seller 2 are shown in Fig. \ref{fig:seller2} (a). First, when paying attention to the market share of seller 2, it shows that the simulation results of \emph{Pacos}-add are the closest to the true values. We also observe a similar trend in the predicted results of other sellers. Further, according to the results in Table \ref{tab:attribute_change}, \emph{Pacos}-add achieves the best results on all metrics considered in our work. As for prior works, PRIMA++ is the second best regarding ranking quality and success rate, while the PNN-based method ranks the second in KLD and MAE, and MNL gives the lowest accuracy on all metrics.

The experimental results when adjusting the price of seller 2 are shown in Fig. \ref{fig:seller2} (b). The predicted market shares results of \emph{Pacos}-NN are the closest to the true values. In addition, in Table \ref{tab:attribute_change} (b), the proposed \emph{Pacos}-NN achieves the best performance on all metrics except for the $sr(m=2)$. In summary, experimental results show that the proposed \emph{Pacos}-add and \emph{Pacos}-NN perform better than prior works in the market share prediction task.


\begin{table}[tbp]
\caption{The results when adjusting (a) reputation, and (b) price of Seller 2.}
  \centering
  \subtable[]{
  \begin{tabular}{cccccc}
    \toprule
          & \emph{Pacos}-add & \emph{Pacos}-NN & PNN   & MNL   & PRIMA++ \\
    \midrule
    rq    & \textbf{0.902 } & 0.823  & 0.840  & 0.792  & 0.850  \\
    sr (m=1) & \textbf{0.775 } & 0.665  & 0.669  & 0.560  & 0.721  \\
    sr (m=2) & \textbf{0.937 } & 0.818  & 0.860  & 0.848  & 0.843  \\
    MAE   & \textbf{0.075 } & 0.120  & 0.114  & 0.202  & 0.178  \\
    KLD   & \textbf{0.105 } & 0.232  & 0.222  & 0.552  & 0.461  \\
    \bottomrule
    \end{tabular}%
  }
  \qquad
  \subtable[]{
        \begin{tabular}{cccccc}
        \toprule
              & \emph{Pacos}-add & \emph{Pacos}-NN & PNN   & MNL   & PRIMA++ \\
        \midrule
        rq    & 0.854  & \textbf{0.861 } & 0.817  & 0.600  & 0.679  \\
        sr (m=1) & 0.659 & \textbf{0.691}  & 0.617  & 0.249  & 0.413  \\
        sr (m=2) & \textbf{0.910 } & 0.905  & 0.846  & 0.623  & 0.648  \\
        MAE   & 0.117  & \textbf{0.091 } & 0.103  & 0.246  & 0.187  \\
        KLD   & 0.211  & \textbf{0.126 } & 0.163  & 0.863  & 0.532  \\
        \bottomrule
        \end{tabular}%
    }
\label{tab:attribute_change}%
\end{table}%

\subsection{Evaluation of The Learning Algorithm}
\label{sec:test_update}
As discussed in Section \ref{sec:learning}, to alleviate the imbalance issue, a novel update rule is proposed to dynamically adjust the value of $\bar B$. Here, an ablation study is conducted to verify its effectiveness, and then the impact of parameter selection on the performance is also analyzed.

\begin{table}[tbp]
  \centering
  \caption{Results of The Ablation Study}
    \begin{tabular}{p{20.25em}ccc}
    \toprule
    \multicolumn{1}{c}{} & rq    & sr ($m=1$) & sr ($m=2$) \\
    \midrule
    Attribute Utility Module Only & 0.721  & 0.446  & 0.770  \\
    All Three Utility Modules& 0.804  & 0.596  & 0.844  \\
    Three Utility Modules + $\bar B$& 0.819  & 0.614  & 0.860  \\
    Three Utility Modules + $\bar B$ + the update rule & \textbf{0.830}  & \textbf{0.642}  & \textbf{0.867}  \\
    \bottomrule
    \end{tabular}%
  \label{tab:ablation}%
\end{table}%

\textbf{The Ablation Study.} \quad The ablation study is conducted on the personalized ranking task. Since \emph{Pacos}-add and \emph{Pacos}-NN have the similar ranking performance, we use \emph{Pacos}-NN as an example to show the results, and observe the same trend in \emph{Pacos}-add. Table \ref{tab:ablation} shows the results of the ablation study, where the values are averaged on all products. First, the proposed model gives the worst results when using the attribute utility module only, and the performance is much improved when incorporating all three utility modules. Then, the introduction of $\bar B$ and the update rule further improve the model's performance. The proposed learning algorithm in Section \ref{sec:learning} achieves the highest accuracy. The ablation study indicates that the introduction of $\bar B$ and the update rule indeed improve the accuracy of \emph{Pacos}.

\textbf{Analysis of The Parameter Selection.} \quad There are four parameters in the update rule, that is, $K_{init}$, $\Delta K$, $\Delta B$, and $B_{max}$. We take $\Delta B$ and $B_{max}$ as an example to analyze the impact of parameter selection, and the analysis of $K_{init}$ and $\Delta K$ can be found in the supplementary file. In our experiment, we vary the value of $B_{max}$ from 1.0 to 3.0, and vary the value of $\Delta B$ from 0.2 to 1.0, and show the results of $sr (m=1)$. The simulation results on other metrics give the same trend and are omitted here. In Fig. \ref{fig:update_rule}, it can be seen that the success rate is very stable when we adjust the value of $B_{max}$ and $\Delta B$. A similar trend is also found in the results of $K_{init}$ and $\Delta K$. It suggests that the proposed model is insensitive to parameters in the update rule.

\begin{figure}[!t]
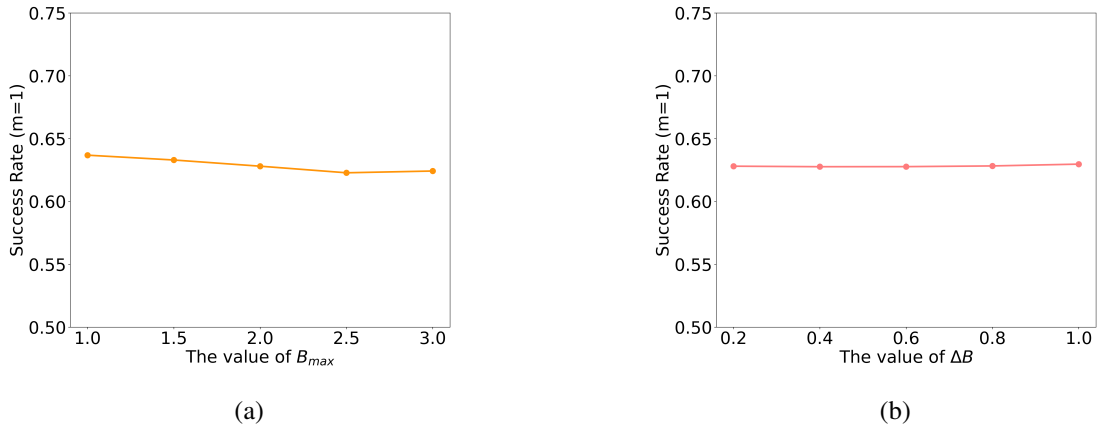

\centering
\begin{minipage}[b]{.48\linewidth}
  \centering
  \centerline{\includegraphics[width=6.5cm]{bar_b.png}}
  \centerline{(a) }\medskip
\end{minipage}
\hfill
\begin{minipage}[b]{0.48\linewidth}
  \centering
  \centerline{\includegraphics[width=6.5cm]{delta_b.png}}
  \centerline{(b) }\medskip
\end{minipage}
\caption{Impact of parameter selection.}
\vspace{-0.15in}
\label{fig:update_rule}
\end{figure}

\subsection{Summary of Real User Tests}
In summary, the proposed \emph{Pacos}-add and \emph{Pacos}-NN achieve better performance in the personalized ranking task, the preference reversal prediction task, and the market share prediction task. In addition, in the preference reversal prediction task, \emph{Pacos} can successfully and accurately predict the occurrence of preference reversals, although there are no such examples in the training set. Furthermore, \emph{Pacos}  not only predicts preference reversals, but also accurately learns the change in users' preferences and helps to identify the cause of preference reversals. 
\section{Conclusion}
\label{sec:conclusion}
In this work, we investigate context effects in choice problems and specifically study preference reversals. We identify three factors contributing to context effects, i.e., adaptive weights, inter-item comparison, and display positions. We propose a context-dependent preference model \emph{Pacos} as a unified framework to address three factors simultaneously. Two module design methods are provided, including an additive method with high interpretability, and an ANN-based method with high predictive accuracy. The theoretical analysis of the additive method demonstrates its effectiveness in dealing with preference reversals. Experimental results show that \emph{Pacos} performs better than existing methods in personalized ranking task, preference reversal prediction task, and market share prediction task, can accurately learn the change in users' preference and help to understand the cause of preference reversals. This study can help understand users' decision making mechanisms, and offer important guidelines on the product design, pricing strategies, and market demand analysis.

\section{Appendix}
\label{sec:appendix}

\begin{proposition}
\label{prop:mnl}
For the IIA issue in prior MNL-based works, we can always obtain the following inequality,
\begin{align}
\label{equ:proof_mnl}
\left[\mathcal{P}(\boldsymbol{s}_A|\mathcal{S})-\mathcal{P}(\boldsymbol{s}_B|\mathcal{S})\right]
\cdot
\left[\mathcal{P}(\boldsymbol{s}_A|\mathcal{S}\cup\boldsymbol{s}_C)
-\mathcal{P}(\boldsymbol{s}_B|\mathcal{S}\cup\boldsymbol{s}_C)\right]\geq 0.
\end{align}
\end{proposition}

\begin{proof}
The IIA issue in prior MNL-based works implies
\begin{align*}
{\mathcal{P}(\boldsymbol{s}_A|\mathcal{S})}\big/{\mathcal{P}(\boldsymbol{s}_B|\mathcal{S})}
={\mathcal{P}(\boldsymbol{s}_A|\mathcal{S}\cup\boldsymbol{s}_C)}\big/
{\mathcal{P}(\boldsymbol{s}_B|\mathcal{S}\cup\boldsymbol{s}_C)}.
\end{align*}
Subtracting 1 from both sides the equation still holds
\begin{align*}
{\mathcal{P}(\boldsymbol{s}_A|\mathcal{S})}\big/{\mathcal{P}(\boldsymbol{s}_B|\mathcal{S})}-1
={\mathcal{P}(\boldsymbol{s}_A|\mathcal{S}\cup\boldsymbol{s}_C)}\big/
{\mathcal{P}(\boldsymbol{s}_B|\mathcal{S}\cup\boldsymbol{s}_C)}-1,
\end{align*}
which can be transformed into
\begin{align*}
\left[{\mathcal{P}(\boldsymbol{s}_A|\mathcal{S})}\big/{\mathcal{P}(\boldsymbol{s}_B|\mathcal{S})}-1\right]
\cdot\left[{\mathcal{P}(\boldsymbol{s}_A|\mathcal{S}\cup\boldsymbol{s}_C)}\big/
{\mathcal{P}(\boldsymbol{s}_B|\mathcal{S}\cup\boldsymbol{s}_C)}-1\right]\geq 0.
\end{align*}
Multiplying $\mathcal{P}(\boldsymbol{s}_B|\mathcal{S})\mathcal{P}(\boldsymbol{s}_B|\mathcal{S}\cup\boldsymbol{s}_C)$ from both sides we have,
\begin{align*}
\left[\mathcal{P}(\boldsymbol{s}_A|\mathcal{S})-\mathcal{P}(\boldsymbol{s}_B|\mathcal{S})\right]
\cdot
\left[\mathcal{P}(\boldsymbol{s}_A|\mathcal{S}\cup\boldsymbol{s}_C)
-\mathcal{P}(\boldsymbol{s}_B|\mathcal{S}\cup\boldsymbol{s}_C)\right]\geq 0.
\end{align*}
\end{proof}

\bibliographystyle{IEEEbib}
\bibliography{refs}

\end{document}


\title{Supplementary File for:\\ \emph{\emph{Pacos}: Modeling Preference Reversals In Users' Context-Dependent Choices}}
\author{Qingming~Li, and H.Vicky~Zhao}
\maketitle

\section{Theoretical Analysis of The Additive Method}

\subsection{IIA issue in MNL-based methods}
\begin{proposition}
\label{prop:mnl}
For the IIA issue in prior MNL-based works, we can always obtain the following inequality,
\begin{align}
\label{equ:proof_mnl}
\left[\mathcal{P}(\boldsymbol{s}_A|\mathcal{S})-\mathcal{P}(\boldsymbol{s}_B|\mathcal{S})\right]
\cdot
\left[\mathcal{P}(\boldsymbol{s}_A|\mathcal{S}\cup\boldsymbol{s}_C)
-\mathcal{P}(\boldsymbol{s}_B|\mathcal{S}\cup\boldsymbol{s}_C)\right]\geq 0.
\end{align}
\end{proposition}

\begin{proof}
The IIA issue in prior MNL-based works implies
\begin{align*}
{\mathcal{P}(\boldsymbol{s}_A|\mathcal{S})}\big/{\mathcal{P}(\boldsymbol{s}_B|\mathcal{S})}
={\mathcal{P}(\boldsymbol{s}_A|\mathcal{S}\cup\boldsymbol{s}_C)}\big/
{\mathcal{P}(\boldsymbol{s}_B|\mathcal{S}\cup\boldsymbol{s}_C)}.
\end{align*}
Subtracting 1 from both sides the equation still holds
\begin{align*}
{\mathcal{P}(\boldsymbol{s}_A|\mathcal{S})}\big/{\mathcal{P}(\boldsymbol{s}_B|\mathcal{S})}-1
={\mathcal{P}(\boldsymbol{s}_A|\mathcal{S}\cup\boldsymbol{s}_C)}\big/
{\mathcal{P}(\boldsymbol{s}_B|\mathcal{S}\cup\boldsymbol{s}_C)}-1,
\end{align*}
which can be transformed into
\begin{align*}
\left[{\mathcal{P}(\boldsymbol{s}_A|\mathcal{S})}\big/{\mathcal{P}(\boldsymbol{s}_B|\mathcal{S})}-1\right]
\cdot\left[{\mathcal{P}(\boldsymbol{s}_A|\mathcal{S}\cup\boldsymbol{s}_C)}\big/
{\mathcal{P}(\boldsymbol{s}_B|\mathcal{S}\cup\boldsymbol{s}_C)}-1\right]\geq 0.
\end{align*}
Multiplying $\mathcal{P}(\boldsymbol{s}_B|\mathcal{S})\mathcal{P}(\boldsymbol{s}_B|\mathcal{S}\cup\boldsymbol{s}_C)$ from both sides we have,
\begin{align*}
\left[\mathcal{P}(\boldsymbol{s}_A|\mathcal{S})-\mathcal{P}(\boldsymbol{s}_B|\mathcal{S})\right]
\cdot
\left[\mathcal{P}(\boldsymbol{s}_A|\mathcal{S}\cup\boldsymbol{s}_C)
-\mathcal{P}(\boldsymbol{s}_B|\mathcal{S}\cup\boldsymbol{s}_C)\right]\geq 0.
\end{align*}
\end{proof}

\subsection{Proof of Lemma 1}
\begin{lemma}
\label{lemma:scalability}
If a new seller $\boldsymbol{s}_{C}$ joins the market $\mathcal{S}$, then
\begin{align}
\label{equ:update_lambda_pos_changes}
\lambda(\boldsymbol{s}_i|\mathcal{S}\cup\boldsymbol{s}_{C})=\lambda (\boldsymbol{s}_i|\mathcal{S})\cdot f(\boldsymbol{s}_{i},\boldsymbol{s}_{C}),
\end{align}
where $f(\boldsymbol{s}_i,\boldsymbol{s}_{C})\triangleq\exp{\left[g(\boldsymbol{ s}_{C})^T\cdot \boldsymbol{s}_i+h(\boldsymbol{s}_i,\boldsymbol{s}_{C})
+\boldsymbol\alpha^T\cdot \Delta pos(\boldsymbol{s}_i,\boldsymbol{s}_{C} )\right]}$. Here, $\Delta pos(\boldsymbol{s}_i,\boldsymbol{s}_{C})=pos(\boldsymbol{s}_i|\mathcal{S}\cup\boldsymbol{s}_{C})
-pos(\boldsymbol{s}_i|\mathcal{S})$ is a change in the display position caused by the addition of seller $\boldsymbol{s}_{C}$, and $f(\boldsymbol{s}_i,\boldsymbol{s} _{C})$ can be seen as the exponential of the change in the total utility of seller $\boldsymbol{s}_i$ caused by the addition of seller $\boldsymbol{s}_{C}$.
\end{lemma}

\begin{proof}
For the additive method, define $\lambda(\boldsymbol{s} _i|\mathcal{S})=\exp{\left[U(\boldsymbol{s}_i|\mathcal{S})\right]}$ based on Eq. \eqref{equ:uti_add}, then we have
\begin{align}
\label{equ:ori_lambda}
\lambda(\boldsymbol{s}_i|\mathcal{S})&=\exp\left\{\left[\sum_{\boldsymbol s_j\in\mathcal{S}}g(\boldsymbol s_j)\right]^T\cdot \boldsymbol{s}_i+\sum_{\boldsymbol{s}_j\in \mathcal{S}}h(\boldsymbol{s}_i,\boldsymbol{s}_j)+
\boldsymbol\alpha^T\cdot pos(\boldsymbol{s}_i|\mathcal{S})\right\}
\end{align}
When a new seller $\boldsymbol{s}_{\text{new}}$ joins the market $\mathcal{S}$, and the display position of seller $\boldsymbol{s} _i$ is $pos(\boldsymbol{s} _i|\mathcal{S}\cup\boldsymbol{s}_{\text{new}} )$ after the addition of the new seller. Then,
\begin{align}
\label{equ:lambda_new}
\lambda(\boldsymbol{s}_i|\mathcal{S}\cup\boldsymbol{s}_{C})&=\exp\left\{\left[\sum_{\boldsymbol s_j\in\mathcal{S}\cup\boldsymbol{s}_{C}}g(\boldsymbol s_j)\right]^T\cdot \boldsymbol{s}_i+\sum_{\boldsymbol{s}_j\in \mathcal{S}\cup\boldsymbol{s}_{C}}h(\boldsymbol{s}_i,\boldsymbol{s}_j)+
\boldsymbol\alpha^T\cdot pos(\boldsymbol{s}_i|\mathcal{S}\cup\boldsymbol{s}_{C})\right\}      \\    \notag
&= \exp\left\{\left[\sum_{\boldsymbol s_j\in\mathcal{S}}g(\boldsymbol s_j)\right]^T\cdot \boldsymbol{s}_i+\sum_{\boldsymbol{s}_j\in \mathcal{S}}h(\boldsymbol{s}_i,\boldsymbol{s}_j)+
\boldsymbol\alpha^T\cdot pos(\boldsymbol{s}_i|\mathcal{S})\right\} \\   \notag
&\quad \quad \cdot \exp{\left\{g(\boldsymbol{ s}_{C})^T\cdot \boldsymbol{s}_i+h(\boldsymbol{s}_i,\boldsymbol{s}_{C})
+\boldsymbol\alpha^T\cdot \left[pos(\boldsymbol{s}_i|\mathcal{S}\cup\boldsymbol{s}_{C})
-pos(\boldsymbol{s}_i|\mathcal{S})\right]\right\}} \\   \notag
&=\lambda (\boldsymbol{s}_i|\mathcal{S})\cdot f(\boldsymbol{s}_{i},\boldsymbol{s}_{C}).
\end{align}
Eq. \eqref{equ:lambda_new} means that $\lambda (\boldsymbol{s}_i|\mathcal{S})$ is simply multiplied by an update term $f(\boldsymbol{s}_i,\boldsymbol{s} _{C})$ when a new seller $\boldsymbol{s}_{C}$ joins the market.
\end{proof}

\subsection{Proof of Theorem 1}
\begin{theorem}
\label{theo:condition}
For the additive method, the sufficient and necessary condition for Eq. (13) to hold is that it is possible to find parameters satisfying
\begin{align}
\label{equ:condition}
\frac{f(\boldsymbol{s}_{B},\boldsymbol{s}_{C})}{f(\boldsymbol{s}_{A},\boldsymbol{s}_{C})}>
\frac{\lambda(\boldsymbol{s}_A|\mathcal{S})}{\lambda(\boldsymbol{s}_B|\mathcal{S})}>1 \quad \text{or} \quad
\frac{f(\boldsymbol{s}_{B},\boldsymbol{s}_{C})}{f(\boldsymbol{s}_{A},\boldsymbol{s}_{C})}<
\frac{\lambda(\boldsymbol{s}_A|\mathcal{S})}{\lambda(\boldsymbol{s}_B|\mathcal{S})}<1.
\end{align}
\end{theorem}

\begin{proof}
Based on Lemma  \ref{lemma:scalability}, we have
\begin{align}
\label{equ:lambda_AB}
\lambda(\boldsymbol{s}_A|\mathcal{S}\cup\boldsymbol{s}_{C})=\lambda (\boldsymbol{s}_A|\mathcal{S})\cdot f(\boldsymbol{s}_{A},\boldsymbol{s}_{C}), \quad \text{and} \quad
\lambda(\boldsymbol{s}_B|\mathcal{S}\cup\boldsymbol{s}_{C})=\lambda (\boldsymbol{s}_B|\mathcal{S})\cdot f(\boldsymbol{s}_{B},\boldsymbol{s}_{C}).
\end{align}

(1) When $\lambda(\boldsymbol{s}_A|\mathcal{S})-\lambda(\boldsymbol{s}_B|\mathcal{S})>0$, Eq. (13) is equivalent to $\left[\lambda(\boldsymbol{s}_A|\mathcal{S}\cup\boldsymbol{s}_C)
-\lambda(\boldsymbol{s}_B|\mathcal{S}\cup\boldsymbol{s}_C)\right]<0$. Combining with Eq. \eqref{equ:lambda_AB}, we have
\begin{align}
\label{equ:proof_lambda}
&\lambda (\boldsymbol{s}_A|\mathcal{S})\cdot f(\boldsymbol{s}_{A},\boldsymbol{s}_{C})-\lambda (\boldsymbol{s}_B|\mathcal{S})\cdot f(\boldsymbol{s}_{B},\boldsymbol{s}_{C})<0 \\  \notag
\Leftrightarrow \quad  &\lambda (\boldsymbol{s}_B|\mathcal{S})\cdot f(\boldsymbol{s}_{B},\boldsymbol{s}_{C}) > \lambda (\boldsymbol{s}_A|\mathcal{S})\cdot f(\boldsymbol{s}_{A},\boldsymbol{s}_{C})\\  \notag
\Leftrightarrow \quad & \frac{f(\boldsymbol{s}_{B},\boldsymbol{s}_{C})}{f(\boldsymbol{s}_{A},\boldsymbol{s}_{C})}
>\frac{\lambda(\boldsymbol{s}_A|\mathcal{S})}{\lambda(\boldsymbol{s}_B|\mathcal{S})} > 1.
\end{align}

If $\lambda(\boldsymbol{s}_A|\mathcal{S})-\lambda(\boldsymbol{s}_B|\mathcal{S})>0$, that is, users originally prefer seller $\boldsymbol{s}_A$ to $\boldsymbol{s}_B$ in market $\mathcal{S}$. Let ${\lambda(\boldsymbol{s}_A|\mathcal{S})}\big /{\lambda(\boldsymbol{s}_B|\mathcal{S})}$ be the ratio between users' original preferences for the two sellers. After the addition of seller $\boldsymbol{s}_C$, the update term of users' preferences for seller $\boldsymbol{s}_A$ is $f(\boldsymbol{s}_{A},\boldsymbol{s}_{C})$, and the update term for seller $\boldsymbol{s}_B$ is $f(\boldsymbol{s}_{B},\boldsymbol{s}_{C})$. Let ${f(\boldsymbol{s}_{B},\boldsymbol{s}_{C})}\big /{f(\boldsymbol{s}_{A},\boldsymbol{s}_{C})}$ be the ratio between the two update terms. Then, the preference reversal occurs when the ratio between the update terms is greater than the ratio between users' original preferences, i.e., ${f(\boldsymbol{s}_{B},\boldsymbol{s}_{C})}\big /{f(\boldsymbol{s}_{A},\boldsymbol{s}_{C})}>{\lambda(\boldsymbol{s}_A|\mathcal{S})}\big /{\lambda(\boldsymbol{s}_B|\mathcal{S})}$ .

(2) When $\lambda(\boldsymbol{s}_A|\mathcal{S})-\lambda(\boldsymbol{s}_B|\mathcal{S})<0$, Eq. (13) is equivalent to $\left[\lambda(\boldsymbol{s}_A|\mathcal{S}\cup\boldsymbol{s}_C)
-\lambda(\boldsymbol{s}_B|\mathcal{S}\cup\boldsymbol{s}_C)\right]>0$. Similar to the analysis for the first case, we have
\begin{align}
\frac{f(\boldsymbol{s}_{B},\boldsymbol{s}_{C})}{f(\boldsymbol{s}_{A},\boldsymbol{s}_{C})}
<\frac{\lambda(\boldsymbol{s}_A|\mathcal{S})}{\lambda(\boldsymbol{s}_B|\mathcal{S})}< 1.
\end{align}
If $\lambda(\boldsymbol{s}_A|\mathcal{S})<\lambda(\boldsymbol{s}_B|\mathcal{S})$, that is, users originally prefer seller $\boldsymbol{s}_B$ to $\boldsymbol{s}_A$ in market $\mathcal{S}$. Then, the preference reversal occurs when the ratio between the update terms is less than the ratio between users' original preferences, i.e., ${f(\boldsymbol{s}_{B},\boldsymbol{s}_{C})}\big /{f(\boldsymbol{s}_{A},\boldsymbol{s}_{C})}<{\lambda(\boldsymbol{s}_A|\mathcal{S})}\big /{\lambda(\boldsymbol{s}_B|\mathcal{S})}$ .
\end{proof}

\subsection{Proof of Theorem 2}
\begin{theorem}
\label{theo:effectiveness}
Based on the proposed additive method, we find one possible solution for ${f(\boldsymbol{s}_{B},\boldsymbol{s}_{C})}\big / {f(\boldsymbol{s}_{A},\boldsymbol{s}_{C})}>
{\lambda(\boldsymbol{s}_A|\mathcal{S})}\big /{\lambda(\boldsymbol{s}_B|\mathcal{S})}>1$ to hold is that $\Gamma_1$ and $\Gamma_2$ in the attribute utility module, $\Phi_1$ and $\Phi_2$ in the comparison utility module and $\boldsymbol\alpha$ in the position utility module satisfy all the following conditions
\begin{align}
\left\{\begin{aligned}
& \left[\sum_{\boldsymbol s_j\in\mathcal{S}}g(\boldsymbol s_j)\right]^T\cdot \left(\boldsymbol{s}_A-\boldsymbol{s}_B\right)>0,
\quad \text{and} \quad \left[\sum_{\boldsymbol s_j\in\mathcal{S}\cup\boldsymbol{s}_C}g(\boldsymbol s_j)\right]^T\cdot \left(\boldsymbol{s}_A-\boldsymbol{s}_B\right)<0, \\
&\left[\boldsymbol h_2(\boldsymbol s_A)-\boldsymbol  h_2(\boldsymbol s_B)\right]^T\cdot \sum_{\boldsymbol{s}_j\in \mathcal{S}}\boldsymbol{s}_j>0, \quad \text{and} \quad  \left[\boldsymbol h_2(\boldsymbol s_A)-\boldsymbol  h_2(\boldsymbol s_B)\right]^T\cdot \sum_{\boldsymbol{s}_j\in \mathcal{S}\cup\boldsymbol{s}_C}\boldsymbol{s}_j<0,\\
&\boldsymbol\alpha^T\cdot \left[pos(\boldsymbol{s}_A|\mathcal{S})-pos(\boldsymbol{s}_B|\mathcal{S})\right]>0, \quad \text{and} \quad \boldsymbol\alpha^T\cdot \left[pos(\boldsymbol{s}_A|\mathcal{S}\cup\boldsymbol{s}_C)-pos(\boldsymbol{s}_B|\mathcal{S}\cup\boldsymbol{s}_C)\right]<0,  \\
\end{aligned}\right.
\end{align}
and one possible solution for ${f(\boldsymbol{s}_{B},\boldsymbol{s}_{C})}\big / {f(\boldsymbol{s}_{A},\boldsymbol{s}_{C})}<
{\lambda(\boldsymbol{s}_A|\mathcal{S})}\big /{\lambda(\boldsymbol{s}_B|\mathcal{S})}<1$ to hold is that $\Gamma_1$, $\Gamma_2$, $\Phi_1$, $\Phi_2$ and $\boldsymbol\alpha$ satisfy
\begin{align}
\left\{\begin{aligned}
& \left[\sum_{\boldsymbol s_j\in\mathcal{S}}g(\boldsymbol s_j)\right]^T\cdot \left(\boldsymbol{s}_A-\boldsymbol{s}_B\right)<0,
\quad \text{and} \quad \left[\sum_{\boldsymbol s_j\in\mathcal{S}\cup\boldsymbol{s}_C}g(\boldsymbol s_j)\right]^T\cdot \left(\boldsymbol{s}_A-\boldsymbol{s}_B\right)>0, \\
&\left[\boldsymbol h_2(\boldsymbol s_A)-\boldsymbol  h_2(\boldsymbol s_B)\right]^T\cdot \sum_{\boldsymbol{s}_j\in \mathcal{S}}\boldsymbol{s}_j<0, \quad \text{and} \quad  \left[\boldsymbol h_2(\boldsymbol s_A)-\boldsymbol  h_2(\boldsymbol s_B)\right]^T\cdot \sum_{\boldsymbol{s}_j\in \mathcal{S}\cup\boldsymbol{s}_C}\boldsymbol{s}_j>0,\\
&\boldsymbol\alpha^T\cdot \left[pos(\boldsymbol{s}_A|\mathcal{S})-pos(\boldsymbol{s}_B|\mathcal{S})\right]<0, \quad \text{and} \quad \boldsymbol\alpha^T\cdot \left[pos(\boldsymbol{s}_A|\mathcal{S}\cup\boldsymbol{s}_C)-pos(\boldsymbol{s}_B|\mathcal{S}\cup\boldsymbol{s}_C)\right]>0.  \\
\end{aligned}\right.
\end{align}
\end{theorem}

\begin{proof}
We begin with ${f(\boldsymbol{s}_{B},\boldsymbol{s}_{C})}\big / {f(\boldsymbol{s}_{A},\boldsymbol{s}_{C})}>
{\lambda(\boldsymbol{s}_A|\mathcal{S})}\big /{\lambda(\boldsymbol{s}_B|\mathcal{S})}>1$. First for ${\lambda(\boldsymbol{s}_A|\mathcal{S})}\big /{\lambda(\boldsymbol{s}_B|\mathcal{S})}>1$, based on the definition of $\lambda(\boldsymbol{s}_i|\mathcal{S})$ in Eq. \eqref{equ:ori_lambda}, we have
\begin{align}
\frac{\lambda(\boldsymbol{s}_A|\mathcal{S})}{\lambda(\boldsymbol{s}_B|\mathcal{S})}
&= \frac{\exp\left\{\left[\sum_{\boldsymbol s_j\in\mathcal{S}}g(\boldsymbol s_j)\right]^T\cdot \boldsymbol{s}_A+\sum_{\boldsymbol{s}_j\in \mathcal{S}}h(\boldsymbol{s}_A,\boldsymbol{s}_j)+
\boldsymbol\alpha^T\cdot pos(\boldsymbol{s}_A|\mathcal{S})\right\}}
{\exp\left\{\left[\sum_{\boldsymbol s_j\in\mathcal{S}}g(\boldsymbol s_j)\right]^T\cdot \boldsymbol{s}_B+\sum_{\boldsymbol{s}_j\in \mathcal{S}}h(\boldsymbol{s}_B,\boldsymbol{s}_j)+
\boldsymbol\alpha^T\cdot pos(\boldsymbol{s}_B|\mathcal{S})\right\}} \\ \notag
&= \exp\Bigg\{\left[\sum_{\boldsymbol s_j\in\mathcal{S}}g(\boldsymbol s_j)\right]^T\cdot \left(\boldsymbol{s}_A-\boldsymbol{s}_B\right)
+\sum_{\boldsymbol{s}_j\in \mathcal{S}} \left[h(\boldsymbol{s}_A,\boldsymbol{s}_j)-h(\boldsymbol{s}_B,\boldsymbol{s}_j)\right]\\ \notag
& \qquad \qquad \qquad \qquad \qquad \qquad \qquad \quad
+\boldsymbol\alpha^T\cdot \left[pos(\boldsymbol{s}_A|\mathcal{S})-pos(\boldsymbol{s}_B|\mathcal{S})\right]\Bigg\}.
\end{align}
Recall that $h(\boldsymbol{s}_i,\boldsymbol{s}_j)=\boldsymbol h_2(\boldsymbol s_i)^T\cdot\boldsymbol{s}_j$, then
\begin{align}
&\frac{\lambda(\boldsymbol{s}_A|\mathcal{S})}{\lambda(\boldsymbol{s}_B|\mathcal{S})} = \exp\Bigg\{\left[\sum_{\boldsymbol s_j\in\mathcal{S}}g(\boldsymbol s_j)\right]^T\cdot \left(\boldsymbol{s}_A-\boldsymbol{s}_B\right)
+\left[\boldsymbol h_2(\boldsymbol s_A)-\boldsymbol  h_2(\boldsymbol s_B)\right]^T\cdot \sum_{\boldsymbol{s}_j\in \mathcal{S}}\boldsymbol{s}_j \\ \notag
& \qquad \qquad \qquad \qquad \qquad \qquad \qquad \qquad \qquad \qquad \qquad \quad
+\boldsymbol\alpha^T\cdot \left[pos(\boldsymbol{s}_A|\mathcal{S})-pos(\boldsymbol{s}_B|\mathcal{S})\right]\Bigg\}.
\end{align}
To make ${\lambda(\boldsymbol{s}_A|\mathcal{S})}\big /{\lambda(\boldsymbol{s}_B|\mathcal{S})}>1$, one possible solution is
\begin{align}
\left[\sum_{\boldsymbol s_j\in\mathcal{S}}g(\boldsymbol s_j)\right]^T\cdot \left(\boldsymbol{s}_A-\boldsymbol{s}_B\right)>0, \quad
&\left[\boldsymbol h_2(\boldsymbol s_A)-\boldsymbol  h_2(\boldsymbol s_B)\right]^T\cdot \sum_{\boldsymbol{s}_j\in \mathcal{S}}\boldsymbol{s}_j>0, \\ \notag
&\quad \quad \quad \quad \text{and} \quad \boldsymbol\alpha^T\cdot \left[pos(\boldsymbol{s}_A|\mathcal{S})-pos(\boldsymbol{s}_B|\mathcal{S})\right]>0.
\end{align}
Next, for ${f(\boldsymbol{s}_{B},\boldsymbol{s}_{C})}\big / {f(\boldsymbol{s}_{A},\boldsymbol{s}_{C})}>
{\lambda(\boldsymbol{s}_A|\mathcal{S})}\big /{\lambda(\boldsymbol{s}_B|\mathcal{S})}$, it is equivalent to $\lambda(\boldsymbol{s}_A|\mathcal{S}\cup\boldsymbol{s}_C)
-\lambda(\boldsymbol{s}_B|\mathcal{S}\cup\boldsymbol{s}_C)<0$ (see the proof of Theorem \ref{theo:condition}), that is, ${\lambda(\boldsymbol{s}_A|\mathcal{S}\cup\boldsymbol{s}_C)}\big /
{\lambda(\boldsymbol{s}_B|\mathcal{S}\cup\boldsymbol{s}_C)}<1$. Based on the definition of ${\lambda(\boldsymbol{s}_i|\mathcal{S}\cup\boldsymbol{s}_C)}$, we have
\begin{align}
\frac{\lambda(\boldsymbol{s}_A|\mathcal{S}\cup\boldsymbol{s}_C)}{\lambda(\boldsymbol{s}_i|\mathcal{S}\cup\boldsymbol{s}_C)}
&= \frac{\exp\left\{\left[\sum_{\boldsymbol s_j\in\mathcal{S}\cup\boldsymbol{s}_C}g(\boldsymbol s_j)\right]^T\cdot \boldsymbol{s}_A+\sum_{\boldsymbol{s}_j\in \mathcal{S}\cup\boldsymbol{s}_C}h(\boldsymbol{s}_A,\boldsymbol{s}_j)+
\boldsymbol\alpha^T\cdot pos(\boldsymbol{s}_A|\mathcal{S}\cup\boldsymbol{s}_C)\right\}}
{\exp\left\{\left[\sum_{\boldsymbol s_j\in\mathcal{S}\cup\boldsymbol{s}_C}g(\boldsymbol s_j)\right]^T\cdot \boldsymbol{s}_B+\sum_{\boldsymbol{s}_j\in \mathcal{S}\cup\boldsymbol{s}_C}h(\boldsymbol{s}_B,\boldsymbol{s}_j)+
\boldsymbol\alpha^T\cdot pos(\boldsymbol{s}_B|\mathcal{S}\cup\boldsymbol{s}_C)\right\}} \\ \notag
&= \exp\Bigg\{\left[\sum_{\boldsymbol s_j\in\mathcal{S}\cup\boldsymbol{s}_C}g(\boldsymbol s_j)\right]^T\cdot \left(\boldsymbol{s}_A-\boldsymbol{s}_B\right)
+\sum_{\boldsymbol{s}_j\in \mathcal{S}\cup\boldsymbol{s}_C} \left[h(\boldsymbol{s}_A,\boldsymbol{s}_j)-h(\boldsymbol{s}_B,\boldsymbol{s}_j)\right]\\ \notag
& \qquad \qquad \qquad \qquad \qquad \qquad \qquad \qquad \quad
+\boldsymbol\alpha^T\cdot \left[pos(\boldsymbol{s}_A|\mathcal{S}\cup\boldsymbol{s}_C)-pos(\boldsymbol{s}_B|\mathcal{S}\cup\boldsymbol{s}_C)\right]\Bigg\} \\ \notag
&= \exp\Bigg\{\left[\sum_{\boldsymbol s_j\in\mathcal{S}\cup\boldsymbol{s}_C}g(\boldsymbol s_j)\right]^T\cdot \left(\boldsymbol{s}_A-\boldsymbol{s}_B\right)
+\left[\boldsymbol h_2(\boldsymbol s_A)-\boldsymbol  h_2(\boldsymbol s_B)\right]^T\cdot \sum_{\boldsymbol{s}_j\in \mathcal{S}\cup\boldsymbol{s}_C}\boldsymbol{s}_j \\ \notag
& \qquad \qquad \qquad \qquad \qquad \qquad \qquad \qquad \qquad \qquad \qquad \quad
+\boldsymbol\alpha^T\cdot \left[pos(\boldsymbol{s}_A|\mathcal{S}\cup\boldsymbol{s}_C)-pos(\boldsymbol{s}_B|\mathcal{S}\cup\boldsymbol{s}_C)\right]\Bigg\}.
\end{align}
To make ${\lambda(\boldsymbol{s}_A|\mathcal{S}\cup\boldsymbol{s}_C)}\big /
{\lambda(\boldsymbol{s}_B|\mathcal{S}\cup\boldsymbol{s}_C)}<1$, one possible solution is
\begin{align}
\left[\sum_{\boldsymbol s_j\in\mathcal{S}\cup\boldsymbol{s}_C}g(\boldsymbol s_j)\right]^T\cdot \left(\boldsymbol{s}_A-\boldsymbol{s}_B\right)<0, \quad
&\left[\boldsymbol h_2(\boldsymbol s_A)-\boldsymbol  h_2(\boldsymbol s_B)\right]^T\cdot \sum_{\boldsymbol{s}_j\in \mathcal{S}\cup\boldsymbol{s}_C}\boldsymbol{s}_j<0, \\ \notag
&\quad \quad \quad \quad \text{and} \quad \boldsymbol\alpha^T\cdot \left[pos(\boldsymbol{s}_A|\mathcal{S}\cup\boldsymbol{s}_C)-pos(\boldsymbol{s}_B|\mathcal{S}\cup\boldsymbol{s}_C)\right]<0.
\end{align}
Summarizing the above discussions, conditions for ${f(\boldsymbol{s}_{B},\boldsymbol{s}_{C})}\big / {f(\boldsymbol{s}_{A},\boldsymbol{s}_{C})}>
{\lambda(\boldsymbol{s}_A|\mathcal{S})}\big /{\lambda(\boldsymbol{s}_B|\mathcal{S})}>1$ to hold is
\begin{align}
\left\{\begin{aligned}
& \left[\sum_{\boldsymbol s_j\in\mathcal{S}}g(\boldsymbol s_j)\right]^T\cdot \left(\boldsymbol{s}_A-\boldsymbol{s}_B\right)>0,
\quad \text{and} \quad \left[\sum_{\boldsymbol s_j\in\mathcal{S}\cup\boldsymbol{s}_C}g(\boldsymbol s_j)\right]^T\cdot \left(\boldsymbol{s}_A-\boldsymbol{s}_B\right)<0, \\
&\left[\boldsymbol h_2(\boldsymbol s_A)-\boldsymbol  h_2(\boldsymbol s_B)\right]^T\cdot \sum_{\boldsymbol{s}_j\in \mathcal{S}}\boldsymbol{s}_j>0, \quad \text{and} \quad  \left[\boldsymbol h_2(\boldsymbol s_A)-\boldsymbol  h_2(\boldsymbol s_B)\right]^T\cdot \sum_{\boldsymbol{s}_j\in \mathcal{S}\cup\boldsymbol{s}_C}\boldsymbol{s}_j<0,\\
&\boldsymbol\alpha^T\cdot \left[pos(\boldsymbol{s}_A|\mathcal{S})-pos(\boldsymbol{s}_B|\mathcal{S})\right]>0, \quad \text{and} \quad \boldsymbol\alpha^T\cdot \left[pos(\boldsymbol{s}_A|\mathcal{S}\cup\boldsymbol{s}_C)-pos(\boldsymbol{s}_B|\mathcal{S}\cup\boldsymbol{s}_C)\right]<0.  \\
\end{aligned}\right.
\end{align}

(2) In the similar way, we discuss the solutions of ${f(\boldsymbol{s}_{B},\boldsymbol{s}_{C})}\big / {f(\boldsymbol{s}_{A},\boldsymbol{s}_{C})}<
{\lambda(\boldsymbol{s}_A|\mathcal{S})}\big /{\lambda(\boldsymbol{s}_B|\mathcal{S})}<1$, and it holds when
\begin{align}
\left\{\begin{aligned}
& \left[\sum_{\boldsymbol s_j\in\mathcal{S}}g(\boldsymbol s_j)\right]^T\cdot \left(\boldsymbol{s}_A-\boldsymbol{s}_B\right)<0,
\quad \text{and} \quad \left[\sum_{\boldsymbol s_j\in\mathcal{S}\cup\boldsymbol{s}_C}g(\boldsymbol s_j)\right]^T\cdot \left(\boldsymbol{s}_A-\boldsymbol{s}_B\right)>0, \\
&\left[\boldsymbol h_2(\boldsymbol s_A)-\boldsymbol  h_2(\boldsymbol s_B)\right]^T\cdot \sum_{\boldsymbol{s}_j\in \mathcal{S}}\boldsymbol{s}_j<0, \quad \text{and} \quad  \left[\boldsymbol h_2(\boldsymbol s_A)-\boldsymbol  h_2(\boldsymbol s_B)\right]^T\cdot \sum_{\boldsymbol{s}_j\in \mathcal{S}\cup\boldsymbol{s}_C}\boldsymbol{s}_j>0,\\
&\boldsymbol\alpha^T\cdot \left[pos(\boldsymbol{s}_A|\mathcal{S})-pos(\boldsymbol{s}_B|\mathcal{S})\right]<0, \quad \text{and} \quad \boldsymbol\alpha^T\cdot \left[pos(\boldsymbol{s}_A|\mathcal{S}\cup\boldsymbol{s}_C)-pos(\boldsymbol{s}_B|\mathcal{S}\cup\boldsymbol{s}_C)\right]>0.  \\
\end{aligned}\right.
\end{align}
\end{proof}

\section{Real User Test}

\subsection{The Personalized Ranking Task}

The results of $sr$ when $m=1$ and $m=2$ are shown in Table \ref{tab:sr1} and \ref{tab:sr2}, respectively. Same as the results of $rq$ in the manuscript, the proposed \emph{Pacos}-add and \emph{Pacos}-NN achieve similar accuracy, which are much better than prior works.

\begin{table}[htbp]
  \centering
  \caption{Results of Success Rate ($m=1$)}
    \begin{tabular}{ccccccc}
    \toprule
          & Air Purifier & Head phone & Hair Dryer & Smart Phone & Scale & Average \\
    \midrule
    Random & 0.234  & 0.223  & 0.246  & 0.236  & 0.244  & 0.237  \\
    MNL   & 0.433  & 0.577  & 0.535  & 0.592  & 0.398  & 0.507  \\
    PRIMA++ & 0.635  & 0.635  & 0.598  & 0.583  & 0.500  & 0.590  \\
    Naive Bayes & 0.595  & 0.634  & 0.582  & 0.544  & 0.560  & 0.583  \\
    RankNet & 0.635  & 0.650  & 0.633  & 0.616  & 0.471  & 0.601  \\
    RankSVM & 0.398  & 0.580  & 0.535  & 0.579  & 0.395  & 0.498  \\
    Pointer NN & 0.652  & 0.624  & 0.655  & \textbf{0.631 } & 0.601  & 0.633  \\
    \midrule
    Ours-add & \textbf{0.680 } & \textbf{0.653 } & 0.668  & 0.647  & 0.630  & 0.656  \\
    Ours-NN & 0.673  & 0.664  & \textbf{0.683 } & \textbf{0.656 } & \textbf{0.642 } & \textbf{0.664 } \\
    \bottomrule
    \end{tabular}%
  \label{tab:sr1}%
\end{table}%

\begin{table}[htbp]
  \centering
  \caption{Results of Success Rate ($m=2$)}
    \begin{tabular}{ccccccc}
    \toprule
          & Air Purifier & Head phone & Hair Dryer & Smart Phone & Scale & Average \\
    \midrule
    Random & 0.470  & 0.446  & 0.492  & 0.470  & 0.490  & 0.474  \\
    MNL   & 0.738  & 0.774  & 0.805  & 0.822  & 0.738  & 0.775  \\
    PRIMA++ & 0.894  & 0.801  & 0.834  & 0.841  & 0.787  & 0.831  \\
    Naïve Bayes & 0.821  & 0.831  & 0.808  & 0.783  & 0.812  & 0.811  \\
    RankNet & 0.900  & \textbf{0.863 } & 0.870  & \textbf{0.870 } & 0.787  & 0.858  \\
    RankSVM & 0.779  & 0.781  & 0.819  & 0.839  & 0.743  & 0.792  \\
    Pointer NN & 0.898  & 0.829  & 0.858  & 0.862  & 0.840  & 0.858  \\
    \midrule
    Ours-add & \textbf{0.905 } & 0.838  & 0.866  & 0.858  & 0.865  & 0.866  \\
    Ours-NN & 0.903  & 0.837  & \textbf{0.871 } & 0.865  & \textbf{0.867 } & \textbf{0.869 } \\
    \bottomrule
    \end{tabular}%
  \label{tab:sr2}%
\end{table}%

\subsection{The Market Share Prediction Task}

In the market share prediction task, we vary the attributes of items according to TableVII in the manuscript, and observe the market share of each item. The results when adjusting attributes of seller 2 are shown in Fig. \ref{fig:seller2}. In Fig. \ref{fig:seller2}, the four subgraphs correspond to the estimated and real market shares four sellers. In each subgraph, the five red dots represent the true market shares of the seller, and a solid line is the simulation result of one method. The experimental results when adjusting the reputation of seller 2 are shown in Fig. \ref{fig:seller2} (a). When paying attention to the market share of seller 2, it shows that the simulation results of \emph{Pacos}-add are the closest to the true values. We also observe a similar trend in the predicted results of other sellers. The experimental results when adjusting the price of seller 2 are shown in Fig. \ref{fig:seller2} (b). The predicted market shares results of \emph{Pacos}-NN are the closest to the true values. The accuracy of the predicted market shares are displayed in Table \ref{tab:attribute_change}. As shown in these results, the proposed \emph{Pacos}-add always output the trend that is the most similar to the true values, and gives the accuracy that outperforms prior works in most cases.

\begin{figure}[ht]
\centering
\begin{minipage}[b]{0.48\linewidth}
  \centering
  \centerline{\includegraphics[width=6.5cm]{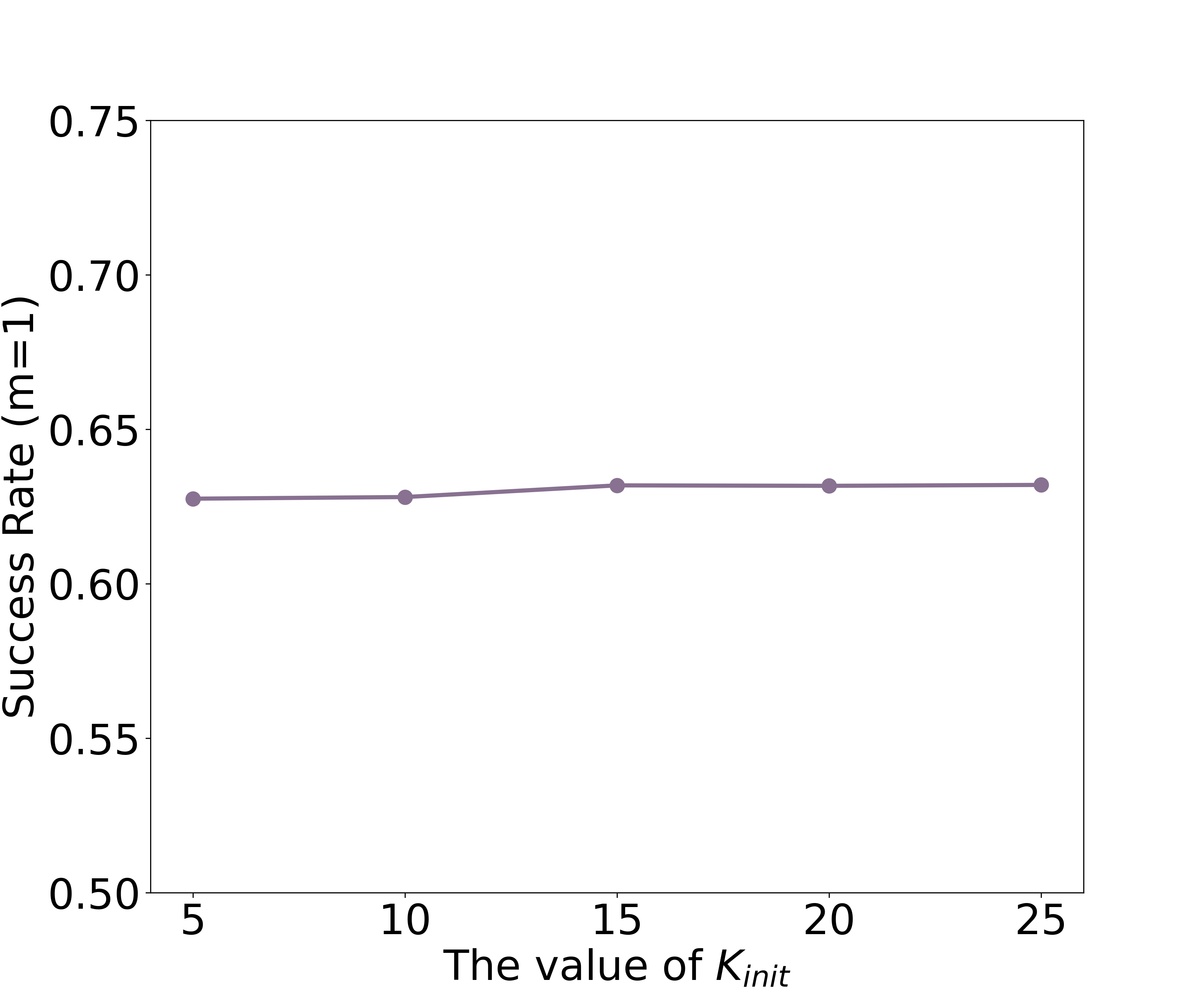}}
  \vspace{-0.1in}
  \centerline{(a) }\medskip
\end{minipage}
\hfill
\begin{minipage}[b]{0.48\linewidth}
  \centering
  \centerline{\includegraphics[width=6.5cm]{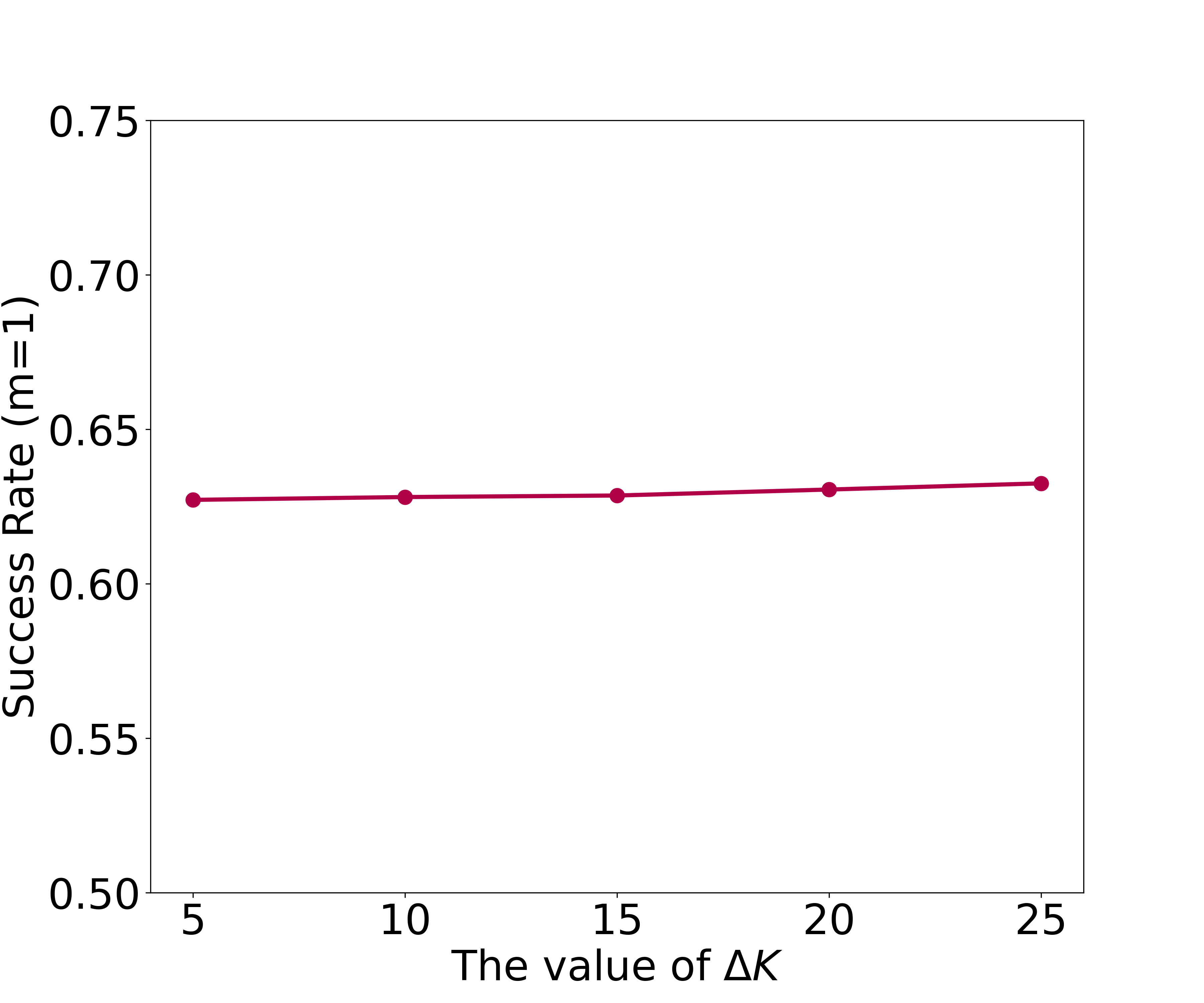}}
  \vspace{-0.1in}
  \centerline{(b) }\medskip
\end{minipage}
\vspace{-.1in}
\caption{Effects of parameter selection.}
\label{fig:update_rule}
\end{figure}

\begin{table}
\caption{The accuracy when adjust (a) the price of item 2, (b) the price of item 3, (c) the reputation of item 3, and (4) the price of item 4.}
\centering
\subtable[]{
    \begin{tabular}{cccccc}
    \toprule
          & Ours-add & Ours-NN & PNN   & MNL   & PRIMA++ \\
    \midrule
    rq    & 0.854  & \textbf{0.861 } & 0.817  & 0.600  & 0.679  \\
    sr (m=1) & \textbf{0.659 } & 0.691  & 0.617  & 0.249  & 0.413  \\
    sr (m=2) & \textbf{0.910 } & 0.905  & 0.846  & 0.623  & 0.648  \\
    MAE   & 0.117  & \textbf{0.091 } & 0.103  & 0.246  & 0.187  \\
    KLD   & 0.211  & \textbf{0.126 } & 0.163  & 0.863  & 0.532  \\
    \bottomrule
    \end{tabular}%
\label{1P}
}
\quad
\subtable[]{
    \begin{tabular}{cccccc}
    \toprule
          & Ours-add & Ours-NN & PNN   & MNL   & PRIMA++ \\
    \midrule
    rq    & 0.910  & 0.887  & \textbf{0.911 } & 0.832  & 0.899  \\
    sr (m=1) & 0.787  & 0.749  & 0.787  & 0.688  & \textbf{0.793 } \\
    sr (m=2) & 0.948  & 0.922  & \textbf{0.952 } & 0.883  & 0.909  \\
    MAE   & \textbf{0.068 } & 0.079  & 0.071  & 0.160  & 0.154  \\
    KLD   & \textbf{0.118 } & 0.139  & 0.190  & 0.415  & 0.386  \\
    \bottomrule
    \end{tabular}%
\label{2P}
}
\quad
\subtable[]{
    \begin{tabular}{cccccc}
    \toprule
          & Ours-add & Ours-NN & PNN   & MNL   & PRIMA++ \\
    \midrule
    rq    & \textbf{0.878 } & 0.842  & 0.830  & 0.788  & 0.872  \\
    sr (m=1) & 0.700  & 0.669  & 0.629  & 0.605  & \textbf{0.716 } \\
    sr (m=2) & \textbf{0.939 } & 0.876  & 0.881  & 0.834  & 0.909  \\
    MAE   & \textbf{0.105 } & 0.137  & 0.127  & 0.191  & 0.173  \\
    KLD   & \textbf{0.214 } & 0.348  & 0.269  & 0.593  & 0.506  \\
    \bottomrule
    \end{tabular}%
\label{2R}
}
\quad
\subtable[]{
    \begin{tabular}{cccccc}
    \toprule
          & Ours-add & Ours-NN & PNN   & MNL   & PRIMA++ \\
    \midrule
    rq    & 0.911  & 0.766  & 0.750  & 0.872  & \textbf{0.936 } \\
    sr (m=1) & 0.813  & 0.627  & 0.588  & 0.750  & \textbf{0.843 } \\
    sr (m=2) & 0.943  & 0.758  & 0.752  & 0.903  & \textbf{0.966 } \\
    MAE   & \textbf{0.108 } & 0.165  & 0.176  & 0.157  & 0.174  \\
    KLD   & \textbf{0.227 } & 0.463  & 0.533  & 0.408  & 0.484  \\
    \bottomrule
    \end{tabular}%

\label{3P}
}
\label{tab:attribute_change}
\end{table}

\begin{figure}[tbp]
\centering
\begin{minipage}[b]{.5\linewidth}
  \centering
  \centerline{\includegraphics[width=14.8cm]{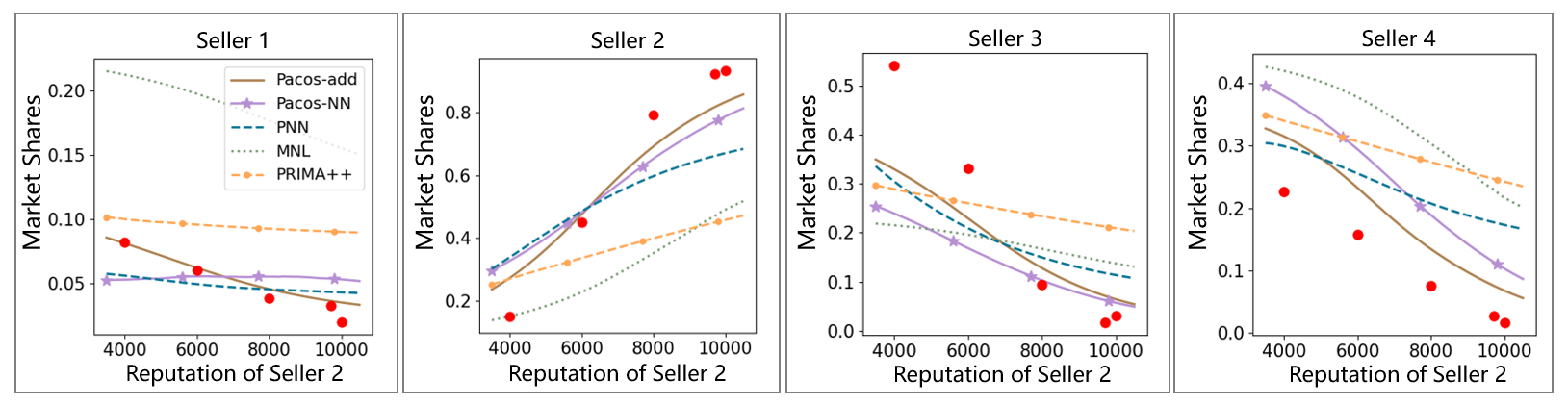}}
  \centerline{(a) adjusting the reputation of seller 2}\medskip
\end{minipage}
\hfill
\begin{minipage}[b]{.5\linewidth}
  \centering
  \centerline{\includegraphics[width=14.8cm]{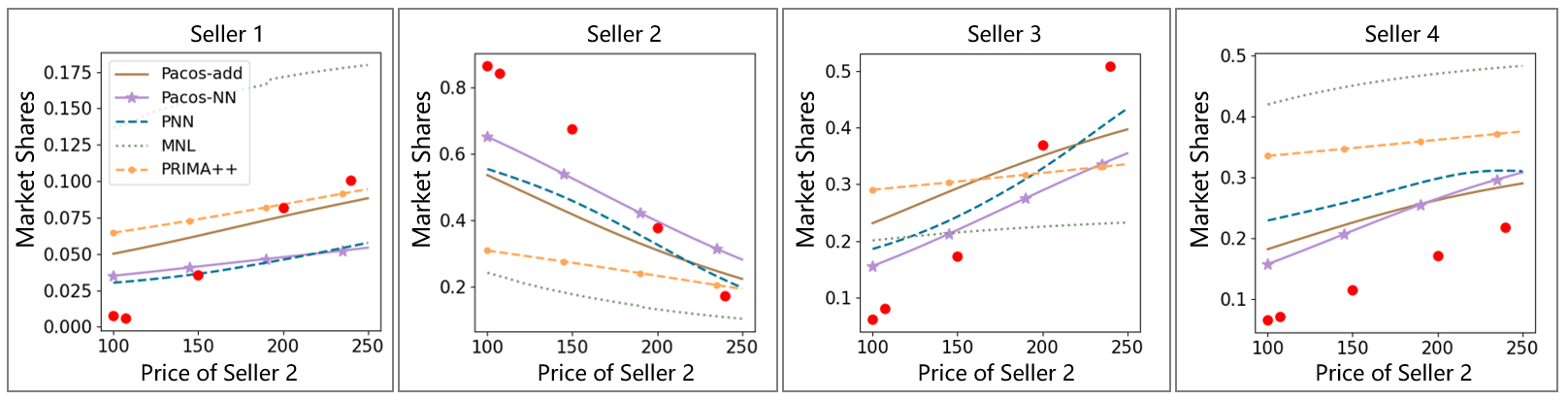}}
  \centerline{(b) adjusting the price of seller 2}\medskip
\end{minipage}
\caption{The predicted trend in market share prediction task.}
\label{fig:seller2}
\end{figure}


\subsection{The Update Rule}
To analyze the impact of $K_{init}$ and $\Delta K$, in our experiment, we vary the value of $K_{init}$ and $\Delta K$ from 5 to 25,  respectively, and the results of $sr (m=1)$ are shown in Fig. \ref{fig:update_rule}. It can be seen from Fig. \ref{fig:update_rule} that the success rate is stable when $K_{init}$ and $\Delta K$ changes, indicating that the proposed model is not sensitive to the selection of $K_{init}$ and $\Delta K$.

\bibliographystyle{IEEEbib}
\bibliography{refs}